\documentclass[11pt]{article}
\usepackage[margin=1in]{geometry}
\usepackage{amsmath,amsfonts,amssymb,amsthm}
\usepackage{graphicx}
\usepackage{enumerate}
\usepackage{bbm}
\usepackage{verbatim}
\usepackage{hyperref,color}
\usepackage[capitalize,nameinlink]{cleveref}
\usepackage[dvipsnames]{xcolor}
\hypersetup{
	colorlinks=true,
	pdfpagemode=UseNone,
	citecolor=OliveGreen,
	linkcolor=NavyBlue,
	urlcolor=Magenta,
	pdfstartview=FitW
}
\usepackage{appendix}
\usepackage{makecell}
\usepackage{tablefootnote}
\crefname{appsec}{Appendix}{Appendices}
\usepackage{tikz}
\usepackage{pgfplots}
\usepackage{xifthen}
\usepackage{subcaption}
\usepackage{hhline}
\usepackage[normalem]{ulem}

\theoremstyle{plain}
\newtheorem{theorem}{Theorem}[section]

\newtheorem{lemma}[theorem]{Lemma}
\newtheorem{corollary}[theorem]{Corollary}

\newtheorem{fact}[theorem]{Fact}

\theoremstyle{definition}
\newtheorem{definition}[theorem]{Definition}

\newtheorem*{assumption*}{Assumption}

\theoremstyle{remark}
\newtheorem{remark}[theorem]{Remark}

\crefname{lemma}{Lemma}{Lemmas}
\crefname{theorem}{Theorem}{Theorems}
\crefname{definition}{Definition}{Definitions}
\crefname{fact}{Fact}{Facts}
\crefname{claim}{Claim}{Claims}
\crefname{proposition}{Proposition}{Propositions}

\newcommand{\dif}{\,\mathrm{d}}
\newcommand{\ind}{\mathbbm{1}}

\newcommand{\norm}[1]{\left\lVert #1 \right\rVert}

\newcommand{\ip}[2]{\left\langle #1 , #2 \right\rangle}

\newcommand{\N}{\mathbb{N}}

\newcommand{\R}{\mathbb{R}}

\newcommand{\EE}{\mathcal{E}}

\newcommand{\TT}{\mathcal{T}}

\newcommand{\XX}{\mathcal{X}}

\newcommand{\qandq}{\quad\text{and}\quad}

\newcommand{\e}{\mathrm{e}}
\renewcommand{\epsilon}{\varepsilon}
\renewcommand{\emptyset}{\varnothing}
\newcommand{\set}[1]{\left\{#1\right\}}
\newcommand{\tuple}[1]{\left(#1\right)} 
\newcommand{\inner}[2]{\left\langle #1,#2\right\rangle}
\newcommand{\tp}{\tuple}

\newcommand{\abs}[1]{\left\vert#1\right\vert}
\newcommand{\ctp}[1]{\left\lceil#1\right\rceil}

\def\*#1{\boldsymbol{#1}} 
\def\+#1{\mathcal{#1}} 
\def\-#1{\mathrm{#1}} 
\def\=#1{\mathbb{#1}} 
\def\!#1{\mathfrak{#1}} 

\def\oPr{\mathop{\mathrm{Pr}}}
\renewcommand{\Pr}[2][]{ \ifthenelse{\isempty{#1}}
  {\oPr\left[#2\right]}
  {\oPr_{#1}\left[#2\right]} } 

\def\oE{\mathop{\mathbb{E}}}
\newcommand{\E}[2][]{ \ifthenelse{\isempty{#1}}
  {\oE\left[#2\right]}
  {\oE_{#1}\left[#2\right]} }

\def\oVar{\mathrm{Var}}
\newcommand{\Var}[2][]{ \ifthenelse{\isempty{#1}}
  {\oVar\left[#2\right]}
  {\oVar_{#1}\left[#2\right]} }

\def\oEnt{\mathrm{Ent}}
\newcommand{\Ent}[2][]{ \ifthenelse{\isempty{#1}}
  {\oEnt\left[#2\right]}
  {\oEnt_{#1}\left[#2\right]} }

\newcommand{\PhiEnt}[2][]{ \ifthenelse{\isempty{#1}}
  {\oEnt^\phi\left[#2\right]}
  {\oEnt^\phi_{#1}\left[#2\right]} }

\newcommand{\fixed}[1]{#1}

\title{Rapid Mixing at the Uniqueness Threshold}

\author{Xiaoyu Chen\thanks{Massachusetts Institute of Technology, USA. Email: \texttt{xiaoyu@mit.edu}.} 
\and Zongchen Chen\thanks{Georgia Institute of Technology, USA. Email: \texttt{chenzongchen@gatech.edu}.}
\and Yitong Yin\thanks{Nanjing University, China. Email: \texttt{yinyt@nju.edu.cn}.}
\and Xinyuan Zhang\thanks{Nanjing University, China. Email: \texttt{zhangxy@smail.nju.edu.cn}.}
}
\date{\today}

\pgfplotsset{compat=1.18} 
\begin{document}

\maketitle

\begin{abstract}
Over the past decades, a fascinating computational phase transition has been identified in sampling from Gibbs distributions. 
Specifically, for the hardcore model on graphs with $n$ vertices and maximum degree $\Delta$, 
the computational complexity of sampling from the Gibbs distribution,
defined over the independent sets of the graph with vertex-weight $\lambda>0$,
undergoes a sharp transition at the critical threshold $\lambda_c(\Delta) := \frac{(\Delta-1)^{\Delta-1}}{(\Delta-2)^\Delta}$, known as the tree-uniqueness threshold:
\begin{itemize}
    \item In the uniqueness regime where $\lambda<\lambda_c(\Delta)$, a local Markov chain for sampling from the Gibbs distribution known as Glauber dynamics  has an optimal mixing time of $O(n \log n)$.
    \item In the non-uniqueness regime where $\lambda>\lambda_c(\Delta)$, the Glauber dynamics exhibits exponential mixing time; furthermore, the sampling problem becomes intractable unless $\mathsf{RP}=\mathsf{NP}$.
\end{itemize}
The computational complexity at the critical point $\lambda = \lambda_c(\Delta)$ remains poorly understood, 
as previous algorithmic and hardness results all required a constant slack from this threshold.

In this paper, we resolve this open question at the critical phase transition threshold, thus completing the picture of the computational phase transition. 
We show that for the hardcore model on graphs with maximum degree $\Delta\ge 3$ at the uniqueness threshold $\lambda = \lambda_c(\Delta)$,
the mixing time of Glauber dynamics is upper bounded by a polynomial in $n$, but is not nearly linear in the worst case: 
specifically, it falls between $\tilde{O}\left(n^{(2+4\e)+O(1/\Delta)}\right)$ and $\Omega\left(n^{4/3}\right)$.

For the Ising model (either antiferromagnetic or ferromagnetic), we establish similar results.
For the Ising model on graphs with maximum degree $\Delta\ge 3$ at the critical temperature $\beta$ where $|\beta| = \beta_c(\Delta)$, 
with the tree-uniqueness threshold $\beta_c(\Delta)$   defined by $(\Delta-1)\tanh\beta_c(\Delta)=1$, 
we show that the mixing time of Glauber dynamics is upper bounded by \fixed{$\tilde{O}\left(n^{3 + O(1/\Delta)}\right)$} and lower bounded by $\Omega\left(n^{3/2}\right)$ in the worst case.
For the Ising model specified by a critical interaction matrix $J$ with $\norm{J}_2=1$, we obtain an upper bound $\tilde{O}(n^{3/2})$ for the mixing time, matching the lower bound $\Omega\left(n^{3/2}\right)$ on the complete graph up to a logarithmic factor.

Our mixing time upper bounds  hold regardless of whether the maximum degree $\Delta$ is constant.
These bounds are derived from a new interpretation and analysis of the localization scheme method introduced by Chen and Eldan \cite{CE22}, applied to the field dynamics for the hardcore model and the proximal sampler for the Ising model.
As key steps in both our upper and lower bounds,  
we establish sub-linear upper and lower bounds for spectral independence at the critical point for worst-case instances.
	
\end{abstract}

\thispagestyle{empty}

\newpage 
\thispagestyle{empty}

\tableofcontents

\thispagestyle{empty}

\newpage

\setcounter{page}{1}

\section{Introduction}\label{sec:introduction}

One of the most remarkable critical phenomena in computation established over the past decades is the computational phase transition in two-spin systems at the uniqueness threshold.

The hardcore model is a quintessential example of a two-spin system that exhibits such a phase transition.
Let $G=(V,E)$ be a graph, and let $\+I(G)$ be the set of all independent sets of $G$. 
The Gibbs distribution $\mu$ of the hardcore model on $G$ defines a distribution over $\+I(G)$, 
where each independent set $I \in \+I(G)$ is assigned a probability:
\begin{align*}
    \mu(I) = \frac{\lambda^{|I|}}{Z},
\end{align*}
where $\lambda>0$ is a parameter known as the \emph{fugacity},
and $Z = \sum_{I \in \+I(G)} \lambda^{|I|}$ is the normalizing constant, referred to as the \emph{partition function}.

Given an instance of the hardcore model on a graph $G$ with fugacity $\lambda > 0$, 
a fundamental problem is to generate random samples approximately from the associated Gibbs distribution $\mu$.
This problem is computationally equivalent to the approximate counting problem of estimating the partition function $Z$,
which has broad applications in statistical physics and statistical inference.
Intuitively, these computational tasks become more intractable as $\lambda$ increases, 
since finding large independent sets is known to be computationally challenging.

In a seminal work~\cite{dyer2002counting}, Dyer, Frieze, and Jerrum raised the question of the critical threshold for the fugacity $\lambda>0$ beyond which an efficient sampler ceases to exist for hardcore models on graphs with maximum degree $\Delta$.
They questioned\footnote{To be precise, the original wording in \cite{dyer2002counting} states: ``One might even rashly conjecture (though we shall not do so) that this critical $\lambda$ is the same as that marking the boundary between unique and multiple Gibbs measures in the independent set (hard core gas) model in the regular infinite tree of degree $\Delta$ (the so-called Bethe lattice).''} 
whether this critical threshold for computation actually coincides with the tree uniqueness threshold $\lambda_c(\Delta) = \frac{(\Delta-1)^{\Delta-1}}{(\Delta-2)^\Delta}$, 
which indicates the phase transition of the uniqueness of Gibbs measure on the infinite $\Delta$-regular tree,
as well as the decay of correlation within such a measure.

This beautiful computational phase transition regarding sampling from the hardcore model has been established in the past decades.
A significant step was made by Weitz~\cite{weitz2006counting}, who showed that if $\lambda < \lambda_c(\Delta)$, 
then there exists a polynomial-time deterministic approximation algorithm for the partition function on graphs with constant maximum degree $\Delta$.
Weitz's result has been highly influential, stimulating numerous follow-up works,
including Barvinok's algorithm based on Taylor's theorem and the interpolation of zero-free polynomials~\cite{Bar06,peters2019conjecture}.
Due to standard self-reducibility, these counting results extend straightforwardly to approximate sampling.
Meanwhile, Sly~\cite{sly2010computational} showed that in a neighborhood just beyond the uniqueness threshold $\lambda_c(\Delta)$, there is no polynomial-time algorithm for approximate counting or sampling unless $\mathsf{RP}=\mathsf{NP}$.
This intractability was later extended to the entire non-uniqueness regime $\lambda > \lambda_c(\Delta)$ \cite{sly2012computational,GSV16}.

While these near-critical results provide theoretical foundations for the computational phase transition of sampling, 
the algorithms based on deterministic counting tend to be slow. 
This is because these algorithms typically rely on enumerating connected substructures of logarithmic size.
As a result, while their running time is polynomial in the size of the graph $n$, the exponent typically depends on model parameters, 
such as the maximum degree of the graph.

On the other hand, 
randomized approaches for sampling and counting, such as Markov chain Monte Carlo (MCMC) algorithms, can be much more efficient.
Among these, Glauber dynamics, also called (random-scan) Gibbs sampling or heat-bath dynamics, is perhaps one of the simplest and most popular MCMC method for sampling from Gibbs distributions.
When applied to the hardcore model, if the current independent set is $I_t$, then the next one $I_{t+1}$ is generated as follows:

\begin{enumerate}
    \item Choose a vertex $v$ uniformly at random;
    \item If $I_t \cup \{v\} \notin \+I(G)$, then set $I_{t+1} = I_t$;
    \item Otherwise, set $I_{t+1} = I_t \cup \{v\}$ with probability $\frac{\lambda}{1+\lambda}$ and otherwise set $I_{t+1} = I_t \setminus \{v\}$ with the remaining probability.
\end{enumerate}

It is well known that the Glauber dynamics is ergodic and  converges to the target distribution~$\mu$. 
The mixing time characterizes the number of steps the chain requires to be sufficiently close to~$\mu$. 
Formally, it is defined as the minimum $t$ such that the distribution of $I_t$ is $(1/4)$-close to $\mu$ in the total variation distance  regardless of the initialization. See \cref{subsec:MC-basics} for details.

Despite the simplicity of Glauber dynamics, 
establishing its mixing time can be extremely challenging. 
For years, progress on rapid mixing in the near-critical regime~\cite{vigoda2001note} lagged behind that of deterministic counting algorithms.
Recently, however, 
Anari, Liu and Oveis Gharan introduced a novel and powerful tool called \emph{spectral independence} for analyzing Glauber dynamics \cite{anari2020spectral}. 
Since its introduction, the spectral independence method has seen great success in obtaining optimal (nearly linear) mixing time for a large family of models under various parameter regimes, see e.g. \cite{chen2021optimal,liu2021coupling,blanca2022mixing,anari2022entropic,CE22,CFYZ22optimal,chen2023near,chen2023coloring,chen2024fast,chen2024stability,anari2024universality}.

Roughly speaking, the notion of spectral independence characterizes the decay of correlations in a spectral way so that the mixing of Glauber dynamics could be proved inductively.
Moreover, the spectral independence property (with the constant independent of $n$) occurs in the same regime as previously established correlation decay properties such as strong spatial mixing \cite{weitz2006counting} or zero-freeness of partition function \cite{peters2019conjecture}. 
In particular, all these nice properties happen if and only if the fugacity is below the critical value $\lambda_c(\Delta)$.

For the hardcore model, a fascinating and drastic transition in the mixing of Glauber dynamics for sampling from the Gibbs distribution
has been established at the uniqueness threshold:
\begin{itemize}
\item \textbf{Uniqueness regime} \cite{anari2020spectral,chen2021optimal,anari2022entropic,CE22,CFYZ22optimal}: For any graph of $n$ vertices and maximum degree $\Delta$, if $\lambda < \lambda_c(\Delta)$, then the mixing time of Glauber dynamics is $O(n \log n)$,
where the constant in $O(\cdot)$ depends only on the slackness of $\lambda$ to $\lambda_c(\Delta)$.
\item \textbf{Non-uniqueness regime} \cite{mossel2009hardness,sly2010computational,sly2012computational,GSV16}: Fix any constant $\Delta \ge 3$. 
There exists an infinite family of graphs of maximum degree $\Delta$, such that for any graph from the family of $n$ vertices, if $\lambda > \lambda_c(\Delta)$, then the mixing time of Glauber dynamics is $\exp(\Omega(n))$.
Furthermore, if $\lambda > \lambda_c(\Delta)$, there is no polynomial-time algorithm for approximate sampling on graphs with maximum degree $\Delta$ unless $\mathsf{RP}=\mathsf{NP}$.
\end{itemize}

The critical behavior described above provides an almost complete picture for the mixing time of Glauber dynamics for the hardcore model on graphs with maximum degree $\Delta$: as the parameter $\lambda$ increases from the subcritical (uniqueness) regime to the supercritical (non-uniqueness) regime, the mixing time transitions from nearly linear to exponential, with the phase transition occurring at the critical point $\lambda_c(\Delta)$.
However, one important case remains unresolved: the sampling problem \textbf{at the uniqueness threshold}, where $\lambda = \lambda_c(\Delta)$.

In fact, not only do we not know the mixing time of Glauber dynamics at criticality, we barely know anything rigorously, regarding either tractability or hardness, of the critical hardcore model. 
All previous algorithms in the subcritical regime, whether deterministic or randomized, cease to be efficient when $\lambda$ approaches $\lambda_c(\Delta)$ from below.
Similarly, the proofs of hardness in the supercritical regime in \cite{sly2012computational,GSV16} fail when $\lambda$ approaches $\lambda_c(\Delta)$ from above. 
As far as we know, there are no rigorously efficient algorithms or established computational hardness for the critical hardcore model on general graphs of maximum degree $\Delta$ or on random $\Delta$-regular graphs. 

The situation is similar for the Ising model, another extensively studied two-spin system defined on graphs; see \cref{subsec:Ising} for the definition. 
While we still  know almost nothing about the critical Ising model on general graphs, 
we have gained more insights into specific families of graphs.
For the ferromagnetic mean-field Ising model (i.e., on the complete graph), it has been shown that at the critical temperature, the Glauber dynamics mixes in time $\Theta(n^{3/2})$ \cite{LLP10,ding2009meanfield}. 
On the planar grid $\mathbb{Z}^2$, Lubetzky and Sly showed that the mixing time of the Glauber dynamics for the critical Ising model is polynomial but not linear in the size of the graph \cite{LS12}.
More recently, Bauerschmidt and Dagallier proved polynomial mixing time upper bounds in $\mathbb{Z}^d$ for $d \ge 5$ \cite{BD24}. 
These results confirm the conjectures from statistical physics that the Ising model undergoes a critical slowdown: the mixing time evolves from nearly linear in the subcritical regime, to polynomial but not nearly linear at criticality, and finally to exponential in the supercritical regime.

We further note that these critical results for the Ising model on complete graphs or lattice do not carry over to the hardcore model. 
In the hardcore model, both complete graphs and complete bipartite graphs are trivial at the critical point;
while on the planar grid, the exact value of the critical point remains unknown, with only numerical predictions available  \cite{blanca2019phase}.

\subsection{Rapid mixing for the critical hardcore model}

We establish polynomial mixing time for the hardcore model at critical fugacity.

Let $G = (V,E)$ be an undirected graph, and let $\lambda \in \mathbb{R}_{>0}$ denote the fugacity. 
Recall that the Gibbs distribution $\mu$ for the hardcore model, supported on the collection of independent sets $\+I(G)$,  is given by:
\begin{align*}
\forall I \in \+I(G), \quad \mu(I) \propto \lambda^{\abs{I}}.
\end{align*}
Let $\lambda_c(\Delta) = \frac{(\Delta-1)^{\Delta-1}}{(\Delta-2)^\Delta}$ be the critical fugacity for degree $\Delta \ge 3$.
For a graph $G$ of maximum degree $\Delta \ge 3$, the \emph{critical hardcore model} on $G$ is the hardcore model with critical fugacity $\lambda = \lambda_c(\Delta)$.

\begin{theorem}[Hardcore Model at the Critical Fugacity]\label{thm:hardcore}
    Consider the Glauber dynamics for the critical hardcore model.
    \begin{itemize}
    \item \textnormal{(Upper bound)} For any graph $G$ of $n$ vertices and maximum degree $\Delta \ge 3$, the mixing time of Glauber dynamics for the critical hardcore model on $G$ is $O\left( n^{(2+4\e) + \frac{4\e}{\Delta-2}} \log \Delta \right)$.
    \item \textnormal{(Lower bound)} Fix a constant $\Delta \ge 3$. There exists an infinite family of graphs of maximum degree $\Delta$, such that for any graph $G$ from the family of $n$ vertices, the mixing time of Glauber dynamics for the critical hardcore model on $G$ is $\Omega\tp{n^{4/3}}$.
    \end{itemize}
\end{theorem}

Our mixing time upper bound holds for any maximum degree $\Delta\ge 3$,
regardless of whether $\Delta$ is bounded by a constant.
Furthermore, our lower bound $\Omega\tp{n^{4/3}}$ refutes the optimal near linear mixing time, 
effectively separating the critical case from the subcritical case.

\subsection{Rapid mixing for the critical Ising model}
\label{subsec:Ising}

Let $G=(V,E)$ be an undirected graph. Let $\beta \in \mathbb{R}$ denote the inverse temperature, and let $\*h \in \mathbb{R}^V$ denote the external fields. The Gibbs distribution for the Ising model is given by:
\begin{align}\label{eq:graph-Ising}
\forall \*x \in \{-1,+1\}^V, \quad \mu(x) \propto \exp\tp{\frac{\beta}{2} \*x^\intercal A_G \*x + \*h^\intercal \*x},
\end{align}
where $A_G$ is the adjacency matrix of $G$. When $\beta>0$, the Ising model is called \emph{ferromagnetic}, while it is called \emph{antiferromagnetic} when $\beta<0$.

The tree-uniqueness threshold of the Ising model (with worst-case or zero external fields $\*h$)  is characterized by a critical inverse temperature $\beta_c(\Delta)$ for every degree $\Delta\ge 3$, defined by:
\[
(\Delta-1) \tanh \beta_c(\Delta) =1.
\]
For a graph $G$ of maximum degree $\Delta$, the \emph{critical ferromagnetic (respectively, antiferromagnetic) Ising model} on $G$ is the Ising model at critical temperature $\beta = \beta_c(\Delta)$ (respectively, $\beta = -\beta_c(\Delta)$) with zero external fields.

An optimal mixing time of  $O(n \log n)$ has been established for the Glauber dynamics within the uniqueness regime when $|\beta|< \beta_c(\Delta)$ \cite{anari2022entropic,chen2021optimal,CE22}.
And slow mixing can be similarly deduced from \cite{mossel2009hardness,sly2010computational,sly2012computational,GSV16} when $|\beta| > \beta_c(\Delta)$.
We remark that while for the antiferromagnetic Ising model the sampling problem is computationally hard when $\beta < -\beta_c(\Delta)$ \cite{sly2010computational,sly2012computational,GSV16}, for the ferromagnetic Ising model there exists an efficient approximate sampler for all $\beta > 0$ \cite{JS93,RW99}.

We prove upper and lower bounds on the mixing time 
of the Glauber dynamics for the Ising model at the critical temperature.

\begin{theorem}[Ising Model at the Critical Temperature]\label{thm:Ising-graphical}
    Consider the Glauber dynamics for the critical Ising model (either ferromagnetic or antiferromagnetic) with zero external fields.
    \begin{itemize}
    \item \textnormal{(Upper bound)} For any graph $G$ of $n$ vertices and maximum degree $\Delta \ge 3$, the mixing time of Glauber dynamics for the critical Ising model on $G$ is \fixed{$O\tp{n^{3+ \frac{4}{\Delta - 2}} \log n}$};
    \item \textnormal{(Lower bound)} Fix a constant $\Delta \ge 3$. There exists an infinite family of graphs of maximum degree $\Delta$, such that for any graph $G$ from the family of $n$ vertices, the mixing time of Glauber dynamics for the critical Ising model on $G$ is $\Omega(n^{3/2})$.
    \end{itemize}
\end{theorem}

Like in \Cref{thm:hardcore}, our mixing time upper bound holds for any maximum degree $\Delta \ge 3$ possibly unbounded. 
Furthermore, the upper bound also holds under arbitrary external fields, though the constant in big-O in the mixing time inevitably depends on the magnitude of the fields (this can be avoided by initializing the Glauber dynamics from a state aligned with the signs of the external fields $\*h$).

Compared to \Cref{thm:hardcore} for the hardcore model, the gap between the upper and lower bounds in \Cref{thm:Ising-graphical} is much smaller.
Notably, the exponent $3/2$ in the lower bound aligns with the critical mean-field ferromagnetic Ising model on the complete graph \cite{LLP10,ding2009meanfield}.

More generally, the Ising model can be specified by an \emph{interaction matrix} $J\in\mathbb{R}^{n\times n}$ which is symmetric and positive semidefinite. The Gibbs distribution $\mu$ is then defined as follows:
\begin{align}
\forall \*x \in \{-1,+1\}^n,\quad  \mu(\*x) \propto \exp\tp{\frac{1}{2}\*x^\intercal J \*x + \*h^\intercal \*x}.\label{eq:Ising-model-interaction-matrix}
\end{align}
Here the assumption of $J$ being positive semidefinite is without loss of generality, since the Ising model with interaction matrix $J+a I$ for any $a \in \mathbb{R}$ is equivalent to the original one.
Note that $J = \beta A_G$ recovers the previous definition \eqref{eq:graph-Ising} for the Ising model defined on a graph $G$.

Recent works have established $O(n \log n)$ mixing time of Glauber dynamics when the interaction matrix satisfies $\norm{J}_2 < 1$ \cite{EKZ22,anari2022entropic,CE22}. 
Very recently, it was shown that the sampling and counting problems become intractable once $\norm{J}_2 > 1$ \cite{Kun24,GKK24}.
These results suggest that the critical Ising model with general interaction matrices should correspond to the case where $\norm{J}_2 = 1$.
Indeed, for the mean-field Ising model on the complete graph, where the interaction matrix is a rank-one matrix given by $J = \frac{\beta'}{n} \*1_n \*1_n^\intercal$, the critical temperature occurs exactly at $\beta' = 1$, in which case $\norm{J}_2 = 1$.

The following theorem establishes a nearly tight bound for the mixing time of the Glauber dynamics for the Ising model specified by an interaction matrix $J$ with critical $2$-norm $\norm{J}_2 = 1$.

\begin{theorem}[Ising Model at the Critical Interaction Norm]\label{thm:Ising-interaction}
    Consider the Glauber dynamics for the Ising model with interaction matrix $J$ and zero external fields.
    \begin{itemize}
    \item \textnormal{(Upper bound)} For any symmetric and positive semidefinite $J \in \mathbb{R}^{n \times n}$ with $\norm{J}_2 = 1$, the mixing time of Glauber dynamics for the Ising model specified by $J$ is $O\tp{n^{3/2} \log n}$.
    \item \textnormal{(Lower bound~\cite{LLP10})} The mixing time of Glauber dynamics for the critical mean-field Ising model specified by $J = \frac{1}{n} \*1_n \*1_n^\intercal$ is $\Omega(n^{3/2})$.
    \end{itemize}
\end{theorem}

As discussed after \cref{thm:Ising-graphical}, the mixing time upper bound in \cref{thm:Ising-interaction} holds under arbitrary external fields as well.

\begin{remark}
	Concurrently and independently to our work, Prodromidis and Sly \cite{PS24+} also showed polynomial mixing of the Glauber dynamics for the critical ferromagnetic Ising model. After posting our paper, we became aware that Bauerschmidt, Bodineau, and Dagallier had previously obtained a similar polynomial mixing result as \cite{PS24+} in their survey \cite[Example 6.19]{BBD24}. Both \cite{PS24+,BBD24} focused on the ferromagnetic Ising model on graphs by extending the approach from \cite{BD24}. Besides, the mixing time upper bounds they obtained are large polynomials, while in our work we aim to obtain a sharp exponent as much as we can.
\end{remark}

\subsection{Deterministic counting at criticality}

As discussed earlier, previous deterministic counting algorithms fail at the critical point including Weitz's self-avoiding walk tree algorithm and Barvinok's Taylor polynomial interpolation method.
We show that there exist simple deterministic counting algorithms for the partition function which work for the critical hardcore and Ising model, with a running time $\exp({O}(n^c))$ for some constant $c < 1$; in particular, the running time is sub-exponential in $n$.

\begin{theorem}\label{thm:main-deterministic}
Let $\Delta \ge 3$ be a constant. 
There exist deterministic algorithms for estimating the partition functions for the hardcore and Ising models at criticality on graphs of $n$ vertices and maximum degree $\Delta$, achieving a relative error of $1\pm n^{-\Omega(1)}$  in time that is sub-exponential in~$n$.
\end{theorem}

We remark that one could attempt to utilize Weitz's algorithm by replacing the exponential decay with the polynomial decay in the strong spatial mixing property, which would yield a significantly weakened approximation, such as a PTAS for the log-partition function (see~\cite[Theorem 5.5]{song2019counting}). 
However, in order to obtain an estimation within $O(n^{-1})$ error bound, the depth of the truncated self-avoiding walk tree needs to be $\Omega(n^2)$, which is even larger than the maximum depth of the self-avoiding walk tree. 
Thus, it may fail to obtain a better algorithm than the brute-force enumeration by simply applying Weitz's algorithm when at criticality.

\paragraph{Organization of the paper}
The rest of the paper is organized as follows.
In \cref{sec:prelim} we give preliminaries. 
In \cref{sec:upper-bound} we present a new interpretation of localization schemes introduced in \cite{CE22} with a new analysis for upper bounding the mixing time of Glauber dynamics.
\fixed{In \cref{sec:SI-criticality} we establish $O(\sqrt{n})$ spectral independence for the critical Ising model with interaction norm $1$, which allows us to obtain a sharper mixing time upper bound in this case.}
We establish the lower bounds on the mixing time in \cref{sec:LB}.
Finally, we present deterministic counting algorithms in \cref{sec:deterministic} and prove \cref{thm:main-deterministic}.

\fixed{
	\begin{remark}
    \label{rmk:bug}
		In previous versions of the paper, as well as in the proceeding of the Annual ACM Symposium on Theory of Computing (STOC 2025), we claimed that the critical Ising model is $O(\sqrt{n})$-spectrally independent, in both the graphical setting with $|\beta| = \beta_c(\Delta)$ and the general interaction setting with $\norm{J}_2 = 1$. We later found a serious loophole in our proof for the graphical case, invalidating our claim of $O(\sqrt{n})$-spectral independence. 
		For this reason, our mixing result for the critical graphical Ising model states as $\tilde{O}\left(n^{3 + O(1/\Delta)}\right)$, worse than the previous, falsely claimed $\tilde{O}\left(n^{2 + O(1/\Delta)}\right)$ upper bound. 
		While we still believe that $O(\sqrt{n})$-spectral independence should hold, we are unable to fix our previous proof or find an alternative approach.
		Meanwhile, our results for the general interaction setting with $\norm{J}_2 = 1$ remain unaffected and valid.
	\end{remark}
}

\section{Preliminaries}\label{sec:prelim}

\subsection{Notations}
Let $\mu$ be a distribution over $\{-1,+1\}^n$ and $\tau \in \{-1,+1\}^\Lambda$ be a partial configuration on $\Lambda \subseteq [n]$ such that $\Pr[X \sim \mu]{X_\Lambda = \tau} > 0$. The distribution $\mu^\tau$ follows from the law of $X \sim \mu$ conditioned on $X_\Lambda = \tau$. For any $\*\lambda \in \mathbb{R}_{\ge 0}$, the distribution $\*\lambda * \mu$ satisfies:
\begin{align} \label{eq:def-magnetization}
    \forall \sigma \in \{-1,+1\}^n,\quad (\*\lambda * \mu)(\sigma) \propto \prod_{i \in [n]:\, \sigma_i = +1} \lambda_i \cdot \mu(\sigma).
\end{align}
The operation in \eqref{eq:def-magnetization} is usually called magnetizing a distribution with local fields~\cite{chen2021rapid} or tilting the measure by external field~\cite{anari2022entropic} in previous works.
\subsection{Markov chains and mixing times}
\label{subsec:MC-basics}
\subsubsection{Basic definitions}
Let $\Omega$ be a finite state space, and $(X_t)_{t \ge 0}$ be a Markov chain over $\Omega$ with transition matrix $P$. 
\begin{itemize}
    \item The Markov chain is \emph{irreducible} if for all $x, y \in \Omega$, there exists $t \ge 0$ with $P^t(x,y) > 0$;
    \item The Markov chain is \emph{aperiodic} if $\gcd\{t \ge 1:P^t(x,x)\} = 1$ for all $x \in \Omega$.
\end{itemize}
The fundamental theorem of Markov chains shows that the Markov chain $(X_t)_{t \ge 0}$ has a unique \emph{stationary distribution}, i.e. a distribution $\mu$ satisfying $\mu P = \mu$, if $(X_t)_{t \ge 0}$ is irreducible and aperiodic.

Let $(X_t)_{t \ge 0}$ be a Markov chain with transition matrix $P$ and stationary distribution $\mu$. The \emph{mixing time} measures the speed of the Markov chain $(X_t)_{t \ge 0}$ converging to its stationary distribution $\mu$, which is formally defined as
\begin{align*}
T_{\mathrm{mix}}:=\max_{x \in \Omega} \min\set{t: d_{\mathrm{TV}}\tp{P^t(x,\cdot),\mu} < \frac{1}{4} },
\end{align*}
where $d_{\mathrm{TV}}(\cdot,\cdot)$ denotes the \emph{total variational distance}.

Let $\mu$ be a distribution supported over the Boolean hypercube $\{-1,+1\}^n$. The Glauber dynamics on $\mu$ updates a configuration $X$ to $Y$ according to the following rules:
\begin{enumerate}
\item Select integer $1 \le i \le n$ uniformly at random;
\item Sample $Y$ from distribution $\mu$ conditioning on $Y_j = X_j$ for all $j\neq i$.
\end{enumerate}

\subsubsection{Functional inequalities}
Let $\mu$ be a distribution over a finite state space $\Omega$, and $(X_t)_{t \ge 0}$ be a Markov chain over $\Omega$ with the transition matrix $P$ and the stationary distribution $\mu$. The \emph{expectation}, \emph{variance}, \emph{entropy}, and the inner product is defined as follows:
\begin{itemize}
\item Expectation: For all $f:\Omega \to \mathbb{R}$, $\E[\mu]{f} := \sum_{x \in \Omega} \mu(x) f(x)$;
\item Variance: For all $f:\Omega \to \mathbb{R}$, $\Var[\mu]{f} := \E[\mu]{f} - \tp{\E[\mu]{f}}^2$;
\item Entropy: For all $f:\Omega \to \mathbb{R}_{\geq 0}$, $\Ent[\mu]{f} := \E[\mu]{f\log f}-\E[\mu]{f} \log \E[\mu]{f}$;
\item Inner product: For all $f,g:\Omega \to \mathbb{R}$, $\inner{f}{g}_\mu = \E[\mu]{f \cdot g}$.
\end{itemize}
In particular, we will use the convention $0 \log 0 = 0$ in the definition of entropy.

\begin{definition}[Poincar\'e inequality and modified log-Sobolev inequality]
    Let $(X_t)_{t \ge 0}$ be a Markov chain over $\Omega$ with the transition matrix $P$ and stationary distribution $\mu$. \begin{enumerate}
        \item We say the \emph{Poincar\'e inequality} holds with constant $\lambda > 0$ if
        \begin{align*}
        \forall f:\Omega \to \mathbb{R},\quad \lambda \Var[\mu]{f} \le \+E_P(f,f),
        \end{align*}
        where $\+E_P(f,g) := \inner{f}{(I-P)g}$ is the \emph{Dirichlet form}. 
        \item We say the \emph{modified log-Sobolev inequality} holds with constant $\rho > 0$ if
        \begin{align*}
            \forall f:\Omega \to \mathbb{R}_{\geq 0},\quad \rho \Ent[\mu]{f} \le \+E_P(f,\log f).
        \end{align*}
    \end{enumerate}
\end{definition}
By the Courant-Fischer theorem, the Poincar\'e inequality holds with constant $\lambda$ if and only if $\lambda \le 1-\lambda_2(P)$, where $\lambda_2(P)$ is the second largest eigenvalue of $P$, and $1-\lambda_2(P)$ is the \emph{spectral gap}. 

Suppose the eigenvalues of $P$ are non-negative. The \emph{relaxation time} $T_{\mathrm{rel}}$ is defined as the inverse of the spectral gap, specifically $T_{\mathrm{rel}} = \frac{1}{1-\lambda_2(P)}$. By~\cite[Theorem 12.4]{levin2017markov}, the mixing time of $P$ can be bounded as follows:
 \begin{align*}
    T_{\mathrm{mix}} \le T_{\mathrm{rel}} \cdot \log \frac{4}{\mu_{\min}}, \quad \text{where $\mu_{\min}:=\min_{x \in \Omega} \mu(x)$.}
 \end{align*}
With \emph{modified log-Sobolev inequality}, a better mixing time bound can be achieved~\cite{bobkov2006modified}:
\begin{align*}
T_{\mathrm{mix}} \le \frac{1}{\rho} \tp{\log \log \frac{1}{\mu_{\min}} + 3}.
\end{align*}

For distribution supported over the Boolean hypercube $\{-1,+1\}^n$, we introduce the \emph{approximate tensorization of variance and entropy}.
\begin{definition}[Approximate tensorization of variance/entropy~\cite{caputo2015approximate}]
    Suppose a distribution $\mu$ is supported over a subset $\Omega$ of the Boolean hypercube $\{-1,+1\}^n$. The distribution $\mu$ satisfies the \emph{approximate tensorization of variance} with constant $C$, if 
\begin{align*}
\forall f \in \Omega \to \mathbb{R}, \quad \Var[\mu]{f} \le C \cdot \sum_{v \in [n]} \E{\Var[v]{f}},
\end{align*}
where $\E[]{\Var[v]{f}}:= \E[X \sim \mu]{\Var[\mu^{X_{[n] \setminus \{v\}}}]{f}}$ denotes the average of local variance. 
Here, the conditional distribution $\mu^{X_{[n] \setminus \{v\}}}$ is $Y \sim \mu$ conditioned on the event that $X_w = Y_w$ for all $w \in [n] \setminus \{v\}$.
Similarly, the distribution $\mu$ exhibits the \emph{approximate tensorization of entropy} with constant $C$, if
\begin{align*}
\forall f \in \Omega \to \mathbb{R}_{\geq 0}, \quad \Ent[\mu]{f} \le C \cdot \sum_{v \in [n]} \E[]{\Ent[v]{f}}.
\end{align*}
where $\E[]{\Ent[v]{f}}$ is defined accordingly.
\end{definition}

We note that the definition of variance/entropy tensorization can be defined accordingly to general $\phi$-entropy ($\phi$-divergence).
We refer to \Cref{def:approximate-tensorization} for details.

The approximate tensorization of variance/entropy is closely related to the Poincar\'e inequality and modified log-Sobolev inequalities for the Glauber dynamics on the target distribution $\mu$, and thereby implies the an upper bound for the mixing time of Glauber dynamics. 
\begin{lemma}[Folklore] \label{lem:tensorization-implies-mixing}
    Let $\mu$ be a distribution over a subset of the Boolean hypercube $\{-1,+1\}^n$. 
    \begin{enumerate}
    \item The distribution $\mu$ satisfies the approximate tensorization of variance with constant $C$ if and only if the Poincar\'e inequality for the Glauber dynamics on $\mu$ holds with constant $\frac{1}{Cn}$.
    As a corollary, the mixing time of Glauber dynamics is bounded by $Cn \log \frac{4}{\mu_{\min}}$.
    \item If $\mu$ satisfies the approximate tensorization of entropy with constant $C$, then the modified log-Sobolev inequality for the Glauber dynamics holds with constant $\frac{1}{Cn}$. As a corollary, the mixing time of Glauber dynamics is bounded by $Cn \tp{\log \log \frac{1}{\mu_{\min}}+3}$.
    \end{enumerate}
\end{lemma}

We refer to~\cite[Fact 3.5 and Fact A.3]{chen2021optimal} for the proof of this lemma.

\subsubsection{Spectral independence}

\begin{definition}[Spectral independence~\cite{anari2020spectral}]
Let $\mu$ be a distribution over a subset of the Boolean hypercube $\{-1,+1\}^n$. Let $X\sim \mu$. The \emph{influence matrix} $\Psi_{\mu}$ of $\mu$ is defined as
\begin{align*}
\Psi_{\mu}(i,j) = \begin{cases}
    \Pr[]{X_j=+1\mid X_i = +1} - \Pr[]{X_j = +1 \mid X_i = -1} & \text{if } \E{X_i} \in (-1,1),\\
    0 & \text{otherwise.}
    \end{cases}
\end{align*}
The distribution $\mu$ is $C$-\emph{spectrally independent} if $\lambda_{\max}(\Psi_{\mu}) \le C$.
Furthermore, we say $\mu$ is \emph{$C$-spectrally independent for all pinning} if for every feasible pinning $\tau$, $\lambda_{\max}(\Psi_{\mu^\tau}) \leq C$.
\end{definition}

\begin{remark}
    For a set $S \in 2^{[n]}$, we can represent it as some vector $X \in \set{-1,+1}^n$ by setting $X_i = +1$ for $i \in S$ and $X_j = -1$ for $j\not\in S$.
    Hence, any distribution $\nu$ on $2^{[n]}$ can be considered as a distribution on $\set{-1,+1}^n$.
    And the notion of spectral independence can be defined for the distribution $\nu$ accordingly.
\end{remark}

We list several equivalent definitions of spectral independence.

\begin{lemma}\label{lem:SI-equiv}
Let $\mu$ be a distribution over a subset $\+X$ of the Boolean hypercube $\{-1,+1\}^n$. The following statements are equivalent. 
\begin{enumerate}
    \item\label{item:SI-1} $\mu$ is $C$-spectrally independent;
    \item\label{item:SI-2} $\mathrm{Cov}(\mu) \preceq C  \cdot \mathrm{diag}\set{\Var[X \sim \mu]{X_i}}_{i \in [n]}$, where $\mathrm{Cov}(\mu)$ is the covariance matrix of $\mu$ defined as:
    $$\mathrm{Cov}(\mu)_{i,j} = \E[X \sim \mu]{X_i X_j} - \E[X \sim \mu]{X_i} \E[X \sim \mu]{X_j},$$
    and we use the notation $A\preceq B$ to indicate that $A - B$ is negative semi-definite;
    \item\label{item:SI-3} For all functions $f = \frac{\nu}{\mu}$, where $\nu$ is a distribution absolutely continuous with respect to $\mu$,
    \begin{align}\label{eq:SI-equiv-form}
        \sum_{i \in [n]} p_i \tp{\frac{q_i}{p_i} - 1}^2 + (1-p_i) \tp{\frac{1-q_i}{1-p_i}-1}^2 \le C \cdot \Var[\mu]{f},
    \end{align}
    where $q_i$ denotes $\Pr[X\sim \nu]{X_i=+1}$, and $p_i$ denotes $\Pr[X\sim \mu]{X_i = +1}$. 
\end{enumerate}
\end{lemma}
\begin{remark}
    The third statement~\eqref{eq:SI-equiv-form} is an analogue of the entropic independence introduced in~\cite{anari2021entropic}, which evaluates the ``independence'' of homogeneous distributions via the higher-order random walk $P_{n \leftrightarrow 1} = D_{n \to 1} U_{1 \to n}$.
\end{remark}
\begin{remark} \label{rem:SI-domain}
  The condition $\-{Cov}(\mu) \preceq C \cdot \-{diag}\set{\Var[X\sim \mu]{X_i}}_{i\in [n]}$ can be defined similarly for distributions $\mu$ over $\set{0,1}^{n}$.
  If $X \sim \mu$, then let $Y = 2X - \*1$, where $\*1$ is the all $1$ vector.
  Let the distribution of $Y$ be $\nu$.
  It is direct to see that $\nu$ is a distribution over $\set{-1,+1}^{n}$.
  We can also verify that $\-{Cov}(\nu) = 4 \-{Cov}(\mu)$ and $\Var[Y\sim \nu]{Y_i} = 4 \Var[X\sim \mu]{X_i}$ for each $1 \leq i \leq n$.
\end{remark}
\begin{proof}
The equivalence of~\Cref{item:SI-1,item:SI-2} follows from the identity:
\begin{align} \label{eq:SI-ratio}
    \Psi_{\mu}(i,j) = \frac{\mathrm{Cov}(\mu)_{i,j}}{\Var[X \sim \mu]{X_i}}.
\end{align}
The equivalence of~\Cref{item:SI-1,item:SI-3} requires the knowledge on high dimensional expander. We refer to~\cite{stefankovic2023lecture} for the background. 
Let $\mu^{\mathrm{hom}}$ and $\nu^{\mathrm{hom}}$ be the homogenization of $\mu$ and $\nu$ respectively, and $f^{\mathrm{hom}} = \frac{\mu^{\mathrm{hom}}}{\mu^{\mathrm{hom}}}$ be the homogenization of $f$.
The left-hand-side of~\eqref{eq:SI-equiv-form} is $n$ times the variance of $\frac{\nu D_{n \to 1}}{\mu D_{n \to 1}} = U_{1 \to n} f$, where $D_{n \to 1}$ is the higher-order down walk from dimension $n$ faces (subsets with $n$ elements) to dimension $1$ faces, and $U_{1 \to n}$ is its adjoint operator. Formally,
\begin{align}\label{eq:equiv-1}
    \nonumber \sum_{i \in [n]} p_i \tp{\frac{q_i}{p_i} - 1}^2 + (1-p_i) \tp{\frac{1-q_i}{1-p_i}-1}^2 &= n \Var[\mu^{\mathrm{hom}} D_{n \to 1}]{U_{1 \to n} f^{\mathrm{hom}}} \\ 
    &= n\inner{f^{\mathrm{hom}}-1}{P_{n \leftrightarrow 1}\tp{f^{\mathrm{hom}}-1}}_{\mu^{\mathrm{hom}}},
\end{align}
where $P_{n \leftrightarrow 1} = D_{n \to 1} U_{1 \to n}$ is the higher-order down-up walk. By Courant--Fischer theorem,
\begin{align}\label{eq:equiv-2}
    \lambda_2(P_{1 \leftrightarrow n})=\lambda_2(P_{n \leftrightarrow 1}) = \sup_{f^{\mathrm{hom}}: \{-1,+1\}^n \to \mathbb{R}_{\ge 0}} \frac{\inner{f^{\mathrm{hom}}-1}{P_{n \leftrightarrow 1}\tp{f^{\mathrm{hom}}-1}}_{\mu^{\mathrm{hom}}}}{\Var[\mu^{\mathrm{hom}}]{f^{\mathrm{hom}}}}.
\end{align}
By~\cite[Theorem 1.3]{anari2020spectral}, the second eigenvalue of the higher-order down-up walk $P_{1 \leftrightarrow n}$ satisfies 
\begin{align}\label{eq:equiv-3}
\lambda_2(P_{1 \leftrightarrow n})=\frac{1}{n} + \frac{n-1}{n} \lambda_2(P^{\perp}_{1 \leftrightarrow n}) = \frac{1}{n} + \frac{n-1}{n} \cdot \frac{1}{n-1} \tp{\lambda_{\max}(\Psi_{\mu}) - 1}= \frac{1}{n} \lambda_{\max}\tp{\Psi_{\mu}},
\end{align}
where $P_{1 \leftrightarrow n}^\perp$ denotes the non-lazy up-down walk of $P_{1 \leftrightarrow n}$.
The equivalence of~\Cref{item:SI-1,item:SI-3} then follows from~\eqref{eq:equiv-1},~\eqref{eq:equiv-2},~\eqref{eq:equiv-3} and $\Var[\mu^{\mathrm{hom}}]{f^{\mathrm{hom}}} = \Var[\mu]{f}$.
\end{proof}

\subsection{Self-avoiding walk tree}
\label{sec:def-SAW-tree}
Let $G = (V, E)$ be a graph.
Given a vertex $r \in V$, a self-avoiding walk $v_0, v_1, \cdots, v_\ell$  from $r$ (i.e., $v_0 = r$) is a sequence of vertices such that $v_i$ and $v_{i+1}$ are adjacent ($0 \leq i <\ell$); and $v_0, \cdots, v_{\ell-1}$ are distinct vertices (but $v_\ell$ may equal to some vertex in $v_0, \cdots, v_{\ell-2}$).
Intuitively, the walk from $r$ is forced to stop once it reaches some vertex that has been passed before.

Let $\mu$ be the Gibbs distribution of a Ising model on the graph $G$ with parameter $\beta$ and inhomogeneous external field $(\lambda_v)_{v\in V}$.
Fix an order $\prec$ on $V$.
For every vertex $r \in V$, the self-avoiding walk (SAW) tree $\+T^{\mathrm{SAW}}_{G,r}$ is a tree with pinning and external field~\cite{weitz2006counting}.
Each vertex $u$ of $\+T^{\mathrm{SAW}}_{G,r}$ represents a self-avoiding walk $p_u$ on $G$ initiated from $r$.
The vertex $u$ is usually considered as a copy of $v_\ell$ (the last vertex in $p_u$); and therefore carries the same external field $\lambda_{v_\ell}$.
Vertices in $\+T^{\mathrm{SAW}}_{G,r}$ are organized as follow:
\begin{itemize}
    \item For vertices $u, v$ in the SAW tree, $u$ is a child of $v$ if and only if $p_u$ extends $p_v$ (i.e., $p_v$ can be obtained by removing the last vertex in $p_u$).
    \item For a leaf node $u$ in the SAW tree with the path $p_u$ be $v_0, \cdots, v_\ell$, either degree of $v_\ell$ in $G$ is 1, or $v_\ell = v_i$ for some $0\leq i \leq \ell-2$.
    In the later case, $u$ will have a pinning $+1$ (resp. $-1$) if $v_{i+1} \prec v_{\ell-1}$ (resp. $v_{i+1} \succ v_{\ell-1}$).
\end{itemize}
The pinnings, external fields, and parameter $\beta$ then define a Gibbs distribution $\mu_{\+T}$  of Ising model on the SAW tree.
For each vertex $u$ in the SAW tree, we also define $\mu_{\+T_u}$ as $\mu_{\+T}$ projected to the subtree rooted at $u$.

\subsection{Large deviation and local limit theorem}
The following local limit theorem plays a key role in our lower bound proofs.
\begin{lemma}[\cite{richter1958multi}]\label{lem:llt}
    Let $(X_i)_{i \ge 1}$ be independent and identically distributed random variables supported over a finite subset $\Omega=\{\*x_1,\*x_2,\ldots,\*x_k\} \subseteq \mathbb{Z}^d$, and $S_n = \sum_{i=1}^n X_i$. 
    Let $\*\mu = \E{X_1}$ and $\Sigma = \mathrm{Cov} \tp{X_1}$. 
    Suppose $X_1$ is not supported over any sublattice and $\mathrm{det}(\Sigma) > 0$. There exist constants $C_1,C_2 > 0$ (only relies on $X_1$) such that the following holds for sufficiently large $n$:
    
    For any $\*x \in \mathbb{Z}^d$ satisfying $\norm{\*x - n \*\mu}_{\infty} \le 3 n^{2/3}$, it holds that
    \begin{align}\label{eq:llt}
        C_1 \le \frac{\Pr{S_n = \*x}}{n^{-d/2} \exp\tp{-\frac{1}{2n} (\*x - n \*\mu)^T \Sigma^{-1} (\*x - n \*\mu)}} \le C_2.
    \end{align}
    
    Furthermore, if $X_1$ is symmetric, i.e.  $\Pr{X_1  = \*\mu + \*x} = \Pr{X_1 = \*\mu-\*x}$ for any $\*x \in \Omega$, then for all $\*x \in \mathbb{Z}^d$ with $\norm{\*x-n\*\mu}_{\infty} \le 3 n^{3/4}$, the Gaussian approximation~\eqref{eq:llt} holds.
\end{lemma}

\begin{remark}
    The local limit theorem stated in~\cite{richter1958multi} is of the following form:
    \begin{align*}
       \frac{\Pr{S_n = \*x}}{\phi(\*x)} = \exp\tp{n\sum_{k=3}^{+\infty} Q_k\tp{\frac{\*x - n\*\mu}{n}}} \tp{1+ O\tp{\norm{\frac{\*x - n\*\mu}{n}}_{\infty}}},
    \end{align*}
    where $\phi(\*x) = (2\pi)^{-d/2} \mathrm{det}(n\Sigma)^{-1/2} \exp\tp{-\frac{1}{2n} (\*x - n \*\mu)^T \Sigma^{-1} (\*x - n \*\mu)}$, the lattice point $\*x$ satisfying $\norm{\*x - n \*\mu}_{\infty} = o(n)$, and $Q_k$ is a homogeneous polynomial of degree $k$ whose coefficients are determined by the distribution of $X_1$. This implies~\eqref{eq:llt} when $\norm{\*x- n \*\mu}_{\infty} \le 3n^{2/3}$. Furthermore, when random variables are symmetric, $Q_3$ vanishes, allowing the Gaussian approximation~\eqref{eq:llt} holds for a larger regime $\norm{\*x - n \*\mu}_{\infty} \le 3n^{3/4}$. This follows from that the semi-invariants (the coefficients in the Taylor's expansion of the logarithm of the moment generating function) of degree $3$ of a symmetric random variable with mean $\*0$ are zeroes. As a corollary, the constant $3$ in $\norm{\*x-n \*\mu}_{\infty} \le 3n^{2/3}$ or $3n^{3/4}$ can be replaced by arbitrary universal constant.
\end{remark}

\subsection{Birthday paradox}
Our percolation proof for spectral independence utilizes the following birthday bound.
\begin{lemma}\label{lem:birthday-paradox}
Suppose $X_0,X_1,\ldots,X_{n}$ are independent and identically distributed copies of a random variable $X$ on $[n]$. Let $T$ be the smallest index $i$ with $X_i = X_j$ for some $0 \le j <i$, then
\begin{align*}
\Pr{T \ge \ell} \le \exp\tp{-\frac{\ell(\ell-1)}{2n}}.
\end{align*}
\end{lemma}

\begin{proof}
Let $p_k = \Pr{X = k}$. Note that
\begin{align}\label{eq:elementary-symmetric-polynomial}
    \Pr{T \ge \ell} = \ell! \sum_{1 \le x_1<x_2<\ldots<x_\ell \le n} \prod_{i=1}^\ell p_{x_i}.
\end{align}
To upper bound the elementary symmetric polynomial, we view it as a optimization program:
\begin{equation}\label{eq:optimization-birthday}
\begin{aligned}
\max_{\*p} \quad & \sum_{1 \le x_1<x_2<\ldots<x_\ell \le n} \prod_{i=1}^\ell p_{x_i}\\
\textrm{s.t.} \quad & 
\sum_{i=1}^n p_i = 1.
\end{aligned}
\end{equation}

By a standard application of the method of Lagrange multiplier, the maximum is achieved when $p_1=p_2=\ldots=p_k$ and $p_{k+1}=\ldots=p_n = 0$. Combining~\eqref{eq:elementary-symmetric-polynomial}, we have
\begin{align*}
\Pr{T \ge \ell} \le \max_{\ell \le k \le n} \ell! \binom{k}{\ell} \frac{1}{k^\ell} \le \exp\tp{-\frac{\ell (\ell-1)}{2n}}.
\end{align*}
This concludes the proof.
\end{proof}

\subsection{The probabilistic method}
In our lower bound proofs, particularly for the constructions of hard instances, 
we utilize the following non-standard application of the probabilistic method.

\begin{lemma}\label{lem:probabilistic-method}
Let $X$ and $Y$ be random variables over a {finite} support $\Omega \subseteq \mathbb{R}$ where $Y>0$ holds. 
Then there exists a {realization} $(x,y)$ 
of $(X,Y)$  such that
\begin{align*}
    \frac{x}{y} \ge \frac{\E[]{X}}{\E[]{Y}}.
\end{align*}
\end{lemma}

\begin{proof}
    Let $Z = \E[]{Y} \cdot X - \E[]{X} \cdot Y$. Then we have $\E[]{Z}=0$. 
    By the probabilistic method, there must exist a realization $z$ of $Z$ and corresponding realization $(x,y)$ of $(X,Y)$, such that $0\le z=\E[]{Y} \cdot x - \E[]{X} \cdot y$, and $y>0$ as $Y>0$. This implies $x/y\ge {\E[]{X}}/{\E[]{Y}}$.
\end{proof}

\section{Upper Bounds}\label{sec:upper-bound}
Our proof of upper bounds involves two steps.
In this section, we present a new proof of the localization scheme approach developed in \cite{CE22} and upper bounds the mixing time by an integral of the spectral independence constant, a fact that has been implicitly implied in \cite{CE22}.
Combining known bounds on the spectral independence and crude bounds near the critical point, we establish rapid mixing of Glauber dynamics for both critical hardcore and Ising models.
In \cref{sec:SI-criticality}, we establish spectral independence for the critical Ising model with a bound of $O(\sqrt{n})$.
This allows us to obtain a sharper exponent in the mixing time for the critical Ising model,
not far from the lower bound which will be proved in \cref{sec:LB}.

    This section is organized as follows. In~\Cref{sec:overview-continuous-downup}, we outline the proof of the mixing time upper bound for the hardcore model by interpreting the field dynamics as a continuous-time down-up walk. In~\Cref{subsec:information-theoretic}, we introduce the information-theoretic view of localization schemes, including the field dynamics and proximal samplers as working examples. Finally, we apply this information-theoretic perspective to critical hardcore model and Ising model in~\Cref{subsec:hardcore-upper,subsec:Ising-upper} respectively.

\subsection{Backgrounds and proof overview}
\label{sec:overview-continuous-downup} 

\subsubsection{Matroid, down-up walks, and field dynamics}

Perhaps it is fair to say that the spectral independence approach and its successful and exciting applications all starts from the breakthrough in sampling basis of matroids \cite{ALOV24}.
We begin by reviewing the method for analyzing the basis exchange process for sampling basis of matroids, also called the down-up walks.

Let $V$ be a set of $n$ elements. 
We do not formally define matroids here but instead consider $\mu$ to be a distribution supported on a family of $r$-subsets of $V$ that satisfy nice properties, e.g. spanning trees of a given graph.
The down and up walks are naturally defined for such distributions.
For any $0\le k \le r$, we define the down walk $D_{r \to k}$ as the Markov kernel from a given $r$-subset $T$ to a random $k$-subset $S\subseteq T$ by removing $r-k$ elements from $T$ chosen uniformly at random.
The up walk $U_{k \to r}$ is the time-reversal for $D_{r \to k}$, i.e., the adjoint operator. 
More precisely, if $\mu_r = \mu$ is the original distribution over $r$-subsets and $\mu_k = \mu_r D_{r \to k}$ is the resulted distribution over $k$-subsets from the down walk, 
then the up walk $U_{k \to r}$ is defined such that 
\begin{align}\label{eq:down-up-adjoint}
    \mu_k (S) U_{k \to r} (S,T) = \mu_r(T) D_{r \to k} (T,S)
\end{align}
for any $k$-subset $S$ and $r$-subset $T$. 
These Markov kernels are crucial for inductively proving rapid mixing of the down-up walk, defined by the kernel $P_{r \leftrightarrow (r-1)} = D_{r \to (r-1)} U_{(r-1) \to r}$. 
One way to explain the idea of the proof is that we study the mixing property for the whole family of associated down-up walks $P_{r \leftrightarrow k} = D_{r \to k} U_{k \to r}$ and use induction on $k$. Crucially, the inductive step requires one to study the mixing property of a local walk and the overall proof forms an elegant local-to-global argument \cite{KO20,CGM21}.

Meanwhile, if the distribution $\mu$ we want to sample from is not supported on subsets of fixed size, such down and up walks are not suitable for the exact same analysis.
To get around this, Anari, Liu and Oveis Gharan \cite{anari2020spectral} noticed a simple and clever way: They define a bijective mapping from subsets of $V$ to $n$-subsets of $U = V \times \{0,1\}$, i.e., each subset $T \subseteq V$ one-to-one corresponds to an $n$-subset $\{(v,\ind_{\{v \in T\}}): v \in V\} \subseteq U$.
Thus, any distribution $\mu$ over subsets of $V$ of any size can be understood as an equivalent distribution $\pi$ over (legal) $n$-subsets of $U$.
The intuition behind this idea is homogenization of the associated partition function.
In particular, the down-up walk for $\pi$ recovers the Glauber dynamics for sampling from $\mu$, and the local-to-global analysis also can be applied in a generalized form \cite{AL20}, giving rise to the notion of spectral independence.

While the homogenization trick works nicely, the resulting mixing time has poor dependency on other parameters such as the maximum degree of the graph or the spectral independence constant. 
To improve this, Chen, Feng, Yin and Zhang \cite{chen2021rapid} developed a new approach. They introduced a new Markov chain named as field dynamics, and compared its mixing property with the Glauber dynamics.
Their approach can also be viewed as a homogenization but in a highly non-trivial way.
Later, the same authors showed that the mixing time of Glauber dynamics is nearly linear in the uniqueness regime, independent of the maximum degree \cite{CFYZ22optimal}.

At about the same time, Chen and Eldan \cite{CE22} developed a new framework for proving mixing bounds for Markov
chains. Not only they established the same result for both field dynamics and Glauber dynamics as in \cite{CFYZ22optimal}, their approach more generally applies to a large family of Markov chains, including for example down-up walks for basis of matroids and the proximal sampler for continuous (log-concave) distributions.
Chen and Eldan study the field dynamics and other Markov chains by defining and understanding the associated localization schemes, a measure-valued stochastic process with martingale behaviors. 
The approach developed in \cite{CE22} relies heavily on tools from stochastic calculus.

One of the contributions of this work, from our point of view, is to present a more intuitive explanation of the localization scheme method from \cite{CE22}.
Our starting point is the following key observation: \emph{the field dynamics is the down-up walk run in continuous time}.
Thus, previous approaches for analyzing down-up walks on fixed-size subsets can be naturally carried over to, for example, the hardcore model over independent sets of any size.
In particular, we are able to recover previous results on field dynamics \cite{CFYZ22optimal,CE22} and furthermore obtain new mixing results all the way up to the critical point, and our proof does not require much knowledge from stochastic calculus.

\subsubsection{Field dynamics is continuous-time down-up walk}

Let us present more details on this observation.
For any $\theta \in [0,1]$, we define $D_{1 \to \theta}$ to be the Markov kernel that, given a subset $T$, generates a random subset $S \subseteq T$ by independently including each element of $T$ into $S$ with probability $\theta$.
Namely, each element in $T$ flips a coin independently to decide if it should be kept or removed.
We note that the Markov kernel $D_{1 \to \theta}$ is a natural analog of the (discrete-time) down walk $D_{r \to k}$ introduced earlier (i.e., removing $r-k$ elements uniformly at random).

Let $\mu_1 = \mu$ be the original distribution, and $\mu_\theta = \mu_1 D_{1 \to \theta}$ be the distribution resulted from the down walk. Note that $\E[\mu_\theta]{|S|} = \theta \, \E[\mu_1]{|T|}$, and thus $\theta$ can be understood as the expected relative size of the subset (analogous to $\mu_k$ on $k$-subsets).

These Markov kernels $D_{1\to \theta}$ in fact describes a stochastic process $(X_t)_{t\in[0,1]}$ on $2^V$ 
which we simply refer to it as the \emph{continuous-time down walk}.
It is formally defined as follows: $X_0$ is a random subset generated from $\mu_1$, and for each $t \in [0,1]$,
\begin{align}
	X_t = \{v \in V: v \in X_0 \wedge r_v > t \}
\end{align}
where $\{r_v\}_{v \in V}$ is a collection of independent random variables uniformly distributed on $[0,1]$.
In particular, $X_1 = \emptyset$. 
In other words, we place a clock at each element $v$ in $X_0$ which rings at a random time $r_v$ independently and uniformly distributed on $[0,1]$, and when the clock rings $v$ is removed.
The continuous-time down walk $(X_t)_{t\in[0,1]}$ is a decreasing process and a time-inhomogeneous Markov process.
We observe that for any $t \in [0,1]$, 
it holds
\begin{align*}
    \Pr{ X_t = S \mid X_0 = T } = D_{1 \to (1-t)} (T,S),
\end{align*}
and the distribution of $X_t$ is exactly $\mu_{1-t}$.

As in the analysis for matroids, we need the analogous notions for the up walks.
Mathematically, they are just Markov kernels $U_{\theta \to 1}$ adjoint to the down walk kernels $D_{1 \to \theta}$ (i.e., satisfying an analog of \eqref{eq:down-up-adjoint}).
However, it would be more intuitive to understand them through the time reversal of the continuous-time down walk, which we call the \emph{continuous-time up walk} and denote it by $(Y_\theta)_{\theta\in[0,1]}$.
Specifically, a path $(Y_\theta)_{\theta\in[0,1]}$ of the continuous-time up walk can be generated by sampling $Y_1$ from $\mu_1$, and putting independent clocks at each element in $Y_1$, so that when the clock at $v$ rings we add $v$ to the current set.

From the general Markov process theory, the continuous-time up walk $(Y_\theta)_{\theta\in[0,1]}$ is also a time-inhomogeneous Markov process and, crucially, it depends on the target distribution $\mu_1$ (observe that this feature is also shared by the matroid up walks).
It is an increasing process and its final state $Y_1$ is distributed as $\mu_1$.
Moreover, both $Y_\theta$ and $X_{1-\theta}$ have the same distribution $\mu_\theta$.
The transitions of $(Y_\theta)_{\theta \in[0,1]}$ are described by the Markov kernels $U_{\theta \to 1}$:
\begin{align*} 
    \Pr{Y_1 = T \mid Y_\theta = S} = U_{\theta \to 1}(S,T).
\end{align*}
In particular, $D_{1 \to \theta}$ and $U_{\theta \to 1}$ are adjoint to each other.

Given the Markov kernels $D_{1 \to \theta}$ and $U_{\theta \to 1}$ for the continuous-time down and up walks respectively, we are now ready to define the \emph{continuous-time down-up walk}, whose Markov kernel is given by $P_{1 \leftrightarrow \theta} = D_{1 \to \theta} U_{\theta \to 1}$ where $\theta \in [0,1]$ is a parameter.
A moment of thought reveals that the continuous-time down-up walk consists of two steps, where in the first step we drop each element from the current subset independently with probability $1-\theta$, and in the second step we try to add back elements according to a suitable conditional distribution.
It turns out that this is \emph{exactly} the field dynamics introduced in \cite{chen2021rapid}. 
In other words, what field dynamics does is exactly running the down-up walk but in continuous time.

With this understanding in mind, to analyze mixing properties of the field dynamics, we could simply follow the inductive approach taken for analyzing the (discrete-time) down-up walk.
The only distinction is that the inductive step involving taking differences of $P_{r \leftrightarrow k}$ for adjacent $k$'s now becomes taking derivative of $P_{1 \leftrightarrow \theta}$ with respect to $\theta$.

We will give a complete proof of rapid mixing in \cref{subsec:information-theoretic} soon assuming entropic stability.
Here we only explain the high-level idea.
Let $\nu$ be a distribution over $2^V$.
To prove rapid mixing of the field dynamics, we need to show that the $\chi^2$ or KL divergence $D( \cdot \,\Vert\, \cdot )$ from $\nu$ to $\mu$ contracts along the continuous-time down walk.
Recall that $\mu_\theta = \mu D_{1 \to \theta}$ and let
$\nu_\theta = \nu D_{1 \to \theta}$.
We wish to show that for all $\theta \in [0,1]$,
\begin{align}\label{eq:half-step-contraction}
    D \left( \nu_\theta \,\Vert\, \mu_\theta \right)
    \le (1-\gamma(\theta)) D \left( \nu \,\Vert\, \mu \right),
\end{align}
where $\gamma(\theta)$ describes the contraction rate.
Rapid mixing of the field dynamics follows immediately from \eqref{eq:half-step-contraction}.
The local-to-global argument essentially corresponds to upper bounding the derivative of the divergence $D \left( \nu_\theta \,\Vert\, \mu_\theta \right)$ with respect to $\theta$. Namely, if we can prove for some $\alpha(\theta)$ that 
\begin{align}\label{eq:ov-entropic-stability}
    \frac{\dif}{\dif \theta} D \left( \nu_\theta \,\Vert\, \mu_\theta \right) 
    \le \alpha(\theta) \left( D \left( \nu \,\Vert\, \mu \right) - D \left( \nu_\theta \,\Vert\, \mu_\theta \right) \right),
\end{align}
then by Gr\"{o}nwall's inequality, \eqref{eq:half-step-contraction} immediately follows.
The inequality \eqref{eq:ov-entropic-stability} is basically the entropic stability condition introduced in \cite{CE22}, a notion generalizing both spectral and entropic independence.

\subsection{An information-theoretic view of localization schemes}
\label{subsec:information-theoretic}
In this subsection, we aim to provide an intuitive interpretation of the localization schemes introduced in \cite{CE22}, 
requiring minimal knowledge of stochastic calculus.
Our approach can be seen as the information-theoretic view of stochastic localization \cite{EM22}, extended also to the field dynamics.

This section is organized as follows.
In \cref{subsubsec:noising-denoising} we study localization schemes by introducing a pair of noising and denoising Markov processes which are time reversals of each other, and 
in \cref{subsubsec:kernel-chain} we define the associated Markov kernels as well as the Markov chains induced by them; our running examples are the field dynamics in the discrete space $2^V$ (applied to the hardcore model) and the proximal sampler in the continuous space $\R^n$ (applied to the Ising model).
In \cref{subsubsec:app-conserv-ent}, we introduce the approximate conservation of entropy which implies, for example, rapid mixing of the field dynamics and proximal sampler, and furthermore rapid mixing of the Glauber dynamics under mild extra assumptions.
Finally, in \cref{subsubsec:ent-stab} we define entropic stability introduced in \cite{CE22} (stated slightly differently) which implies approximate conservation of entropy; this implication can be viewed as a continuous analog of the local-to-global analysis for matroid down-up walks.

\subsubsection{Noising and denoising processes}
\label{subsubsec:noising-denoising}
Let $\mu$ represent the target distribution over a state space $\XX \subseteq \R^n$, from which we seek to generate random samples.
For simplicity we assume $\XX$ is \emph{finite} in this paper, which is sufficient for our applications in hardcore and Ising models.

\paragraph{Noising Process.} 
Consider a stochastic process $(X_t)_{t \in \TT}$ taking values in $\XX$, 
where $X_0$ is drawn from the target distribution $\mu$.
This process progressively modifies the distribution over time, and is referred to as the \emph{noising process} (also called \emph{forward process} or \emph{down walk} under various settings).
In this paper, we focus on the continuous-time case where $\TT = [0,1]$,  though extending this to the discrete-time case, where $\TT = \{0,1,\dots, T\}$, $T \in \N$, is straightforward.
The stochastic process $(X_t)$ we consider here is a \emph{Markov process}, though it is generally \emph{not} time-homogeneous. 
For simplicity we assume the law of $X_1$ is a Dirac measure concentrated at $\*0$; in other words, $(X_t)$ always arrives at the origin at the end of the process.

\paragraph{Denoising Process.}
Now, consider the time reversal of $(X_t)_{t \in [0,1]}$,  denoted by $(Y_\theta)_{\theta \in [0,1]}$.
The time-reversed process $(Y_\theta)_{\theta \in [0,1]}$ is also a \emph{Markov process}, though, like the original process, it is generally \emph{not} time-homogeneous. 
We refer to $(Y_\theta)_{\theta \in [0,1]}$ as the \emph{denoising process} (also called \emph{backward process} or \emph{up walk} under various settings); crucially, it depends on $\mu$, the distribution of $Y_1$.
Note that we have the distributional equality $Y_\theta \overset{d}{=} X_{1-\theta}$.
The initial point $Y_0 = \*0$ is set to the origin, which corresponds to $X_1$ in the noising process.
This denoising process is called an \emph{observation process} in \cite{Mon23} and is more important in the analysis.

\paragraph{\underline{\normalfont\textit{Examples of noising and denoising processes.}}}
Before moving on, we look at some examples first.
The important noising/denoising processes that will be considered in this paper is the continuous-time down/up walk and the Brownian bridge/F\"ollmer processes.
\begin{definition}[continuous-time down and up walks] \label{definition-decreasing-process}
    Let $\mu$ be a distribution over $2^V$ where  $V=[n]$.
The \emph{continuous-time down walk} $(X_t)_{t\in[0,1]}$ with initial distribution $X_0\sim \mu$ is a continuous-time stochastic process constructed as follows.
Each $v\in V$ is associated with an independent random variable $r_v$, uniformly distributed on $[0,1]$.
    For each $t \in [0,1]$, the process is defined by:
    \[
    X_t = \{v \in V: v \in X_0 \wedge t < r_v \}.
    \]
    As $t$ increases, elements  are progressively removed from $X_t$, until eventually $X_1 = \emptyset$. 

    The associated denoising process $(Y_\theta)_{\theta \in[0,1]}$, i.e., the \emph{continuous-time up walk} will be defined as a revealing or observing process. 
    Let $Y_0 = \emptyset$ and $Y_1 \sim \mu$ be sampled in the beginning.
    Each vertex $v \in V$ is associated with an independent random variable $s_v$, uniformly distributed on $[0,1]$.
    Then, for each $\theta \in [0,1]$, the process is defined by:
    \begin{align*}
        Y_\theta = \set{v \in V: v \in Y_1 \land \theta > s_v}.
    \end{align*}
    As $\theta$ increases, the elements in $Y_1$ is revealed progressively.
\end{definition}

We remark that the continuous-time up walk $(Y_\theta)_{\theta\in[0,1]}$ is a Markov process since it is the time reversal of $(X_t)_{t\in[0,1]}$, even though the way we describe it does not indicate this.

The continuous time down/up walks are defined in the discrete space $2^V$ for set families.
We now consider the analog of them in the continuous space $\R^n$.
Let $\mu$ be the target distribution in $\R^n$.
In the discrete case of set families, the continuous-time down walk starts from a random subset, gradually removes elements, and eventually reaches the empty set at time $t=1$. 
In $\R^n$, a natural analog of the empty set is the origin point $\*0$, and starting from a random point $X \sim \mu$, we wish to design a random continuous path (i.e., a stochastic process) connecting $X$ and the origin $\*0$.
This motivates us to consider the \emph{(standard) Brownian bridge process}, which is the (standard) Brownian motion $(B_t)_{t \in [0,1]}$ conditioned on $B_1 = \*0$.
The following standard fact about the Brownian bridge process is helpful to us.

\begin{fact}\label{fact:BB}
    Let $(B_t)_{t \in [0,1]}$ denote the standard Brownian motion. 
    Then the continuous-time stochastic process $(B^{\mathrm{br}}_t)_{t \in [0,1]}$ defined as
    \[
        B^{\mathrm{br}}_t = (1-t) B_{\frac{t}{1-t}} 
    \]
    is a standard Brownian bridge process whose probability distribution is the conditional distribution of $(B_t)_{t \in [0,1]}$ subject to the condition that $B_1 = \*0$.
    Similarly, the process $(B^{\mathrm{br}'}_t)_{t \in [0,1]}$ defined as
    \[
        B^{\mathrm{br}'}_t = t B_{\frac{1-t}{t}} 
    \]
    is also a standard Brownian bridge process.
\end{fact}

We now define the continuous-space analog of down and up walks using Brownian bridges. 
For the noising process, we run the Brownian bridge but further assuming starting at a random point (generated from $\mu$ plus a pushforward step); we still call it the Brownian bridge process for simplicity.
For the denoising process, we recover the F{\"o}llmer process and stochastic localization (with a fixed driving matrix).

\begin{definition} \label{definition-brownian-bridge}
Let $(B^{\mathrm{br}}_t)_{t \in [0,1]}$ be a $r$-dimensional standard Brownian bridge.
Let $\mu$ be an initial distribution over $\XX \subseteq \mathbb{R}^n$,
and let $L \in \mathbb{R}^{r \times n}$ be a driving matrix. 
The \emph{Brownian bridge process} $(X_t)_{t \in [0,1]}$ with driving matrix $L$
starts at $X_0 \sim \mu$, and for all $t \in (0,1]$, evolves as
\begin{align}\label{eq:BBP}
    X_t = (1-t) L X_{0} + B^{\mathrm{br}}_t.
\end{align}
At time $t = 0$, it jumps from $X_0$ to $L X_0$.
Then, as $t$ increases, it evolves from $LX_0$ to $\*0$ via a Brownian bridge.

Let $(B_\theta^{\mathrm{br'}})_{\theta \in [0,1]}$ be a $r$-dimensional standard Brownian bridge.
The \emph{F{\"o}llmer process} $(Y_\theta)_{\theta\in [0,1]}$ 
starts at $Y_0 = 0$ and is defined by
first sample $Y_1 \sim \mu$ and for $\theta \in [0,1)$, it evolves as
\begin{align}\label{eq:FP}
    Y_\theta = \theta L Y_1 + B^{\mathrm{br'}}_\theta.
\end{align}
As $t$ increases, it evolves from $\*0$ to $L Y_1$.
And, finally, when $\theta = 1$, it jumps from $L Y_1$ to $Y_1$.
\end{definition}

The F{\"o}llmer process \cite{follmer2005,follmer2006} has attracted a lot of attention in recent years and is especially helpful in establishing various functional inequalities, see for examples \cite{Lehec13,ELS20,EM20,Mik21,KP23}.
The F{\"o}llmer process is also known to be closely related to stochastic localization, see \cite[Lemma 4.1]{MS24}.
Here our treatment of the F{\"o}llmer process is slightly informal. 
For a standard and formal definition, one needs to write down the stochastic differential equation and show the solution has the same law as the stochastic process defined by \eqref{eq:FP}.
For our purposes, however, we only need the existence of the F{\"o}llmer process as a Markov process and the explicit forms of Markov kernels associated with it (defined in the next subsection). This is sufficient for us to derive rapid mixing of associated Markov chains and further the Glauber dynamics.

\begin{remark}
    By \cref{fact:BB}, the Brownian bridge process \eqref{eq:BBP} can be equivalently represented as 
    \begin{align*}
        X_t = (1-t) \left( L X_{0} + B_{\frac{t}{1-t}} \right).
    \end{align*}
    Taking $\tilde{X}_t = (1+t) X_{\frac{t}{1+t}}$ for $t \in (0,\infty)$, we obtain
    \begin{align}\label{eq:tilde-X}
        \tilde{X}_t = LX_0 + B_t,
    \end{align}
    which is a more common representation in literature for the noising process.
    Similarly, the F{\"o}llmer process \eqref{eq:FP} can be equivalently represented as
    \begin{align*}
        Y_\theta = (1-\theta) \left( \frac{\theta}{1-\theta} L Y_1 + B_{\frac{\theta}{1-\theta}} \right).
    \end{align*}
    Taking $\tilde{Y}_t = (1+t) Y_{\frac{t}{1+t}}$ for $t \in (0,\infty)$, we obtain
    \begin{align}\label{eq:tilde-Y}
        \tilde{Y}_t = t L Y_1 + B_t,
    \end{align}
    which is again a more common representation in literature for the denoising process or stochastic localization \cite{EM22,KP23}.
    We remark that \cref{eq:tilde-Y} is called the \emph{linear observation process} in \cite{Mon23}.
\end{remark}

\subsubsection{Associated Markov kernels and Markov chains}
\label{subsubsec:kernel-chain}
Given a pair of noising process $(X_t)_{t \in [0,1]}$ and denoising process $(Y_\theta)_{\theta \in [0,1]}$, we study the Markov kernels associated with them and define a family of discrete-time Markov chains from these kernels.
These Markov chains are not only important themselves, but also allow us to derive rapid mixing of the Glauber dynamics.
A summary of noising and denoising processes discussed in this paper, and their associated Markov chains is provided in \cref{table:summary}.

{
\renewcommand{\arraystretch}{1.4}
\begin{table}[t]
\centering
\begin{tabular}{|c|c|c|}
\hline
Markov chain & \makecell{Field dynamics\\(\Cref{def:field-dynamics})} & \makecell{Proximal sampler\\(\Cref{def:proximal-sampler})} \\ \hhline{|=|=|=|}
Domain\tablefootnote{Though we only apply the proximal sampler to discrete distributions, it is applicable to continuous-space distributions.}  & $2^{[n]}$ & $\mathbb{R}^n$ \\ \hline
Null information state\tablefootnote{The state that $X_1$ equals almost surely.} & $\emptyset$ & $\*0$ \\ \hline 
Noising process & \makecell{Continuous-time down walk\\
(\Cref{definition-decreasing-process})} & \makecell{Brownian bridge process\\(\Cref{definition-brownian-bridge})} \\ \hline
Denoising process & \makecell{Continuous-time up walk\\
(\Cref{definition-decreasing-process})} & \makecell{F{\"o}llmer process \\(\Cref{definition-brownian-bridge})} \\ \hline
\end{tabular}
\caption{Two examples of localization schemes: Field dynamics and Proximal sampler}
\label{table:summary}
\end{table}
}

\paragraph{Markov Kernels.}
Both the noising process $(X_t)$ and the denoising process $(Y_\theta)$ can be characterized by a collection of double-indexed Markov kernels $(P_{\theta \to \eta})_{0\le \eta \le \theta \le 1}$ and their time reversals $(Q_{\eta \to \theta})_{0\le \eta \le \theta \le 1}$. 
For any $0 \le \eta \le \theta \le 1$, the following holds:
\begin{align*}
	P_{\theta \to \eta}(x, \cdot) 
	:= \Pr{ X_{1-\eta} \in \cdot \mid X_{1-\theta} = x }
	= \Pr{ Y_\eta \in \cdot \mid Y_\theta = x }; \\
    Q_{\eta \to \theta}(x, \cdot)
    := \Pr{ Y_\theta \in \cdot \mid Y_\eta = x } 
    = \Pr{ X_{1-\theta} \in \cdot \mid X_{1-\eta} = x }.
\end{align*}
Observe that for all $\theta \in [0,1]$, it holds $Q_{\theta \to \theta} = P_{\theta \to \theta} = \mathrm{id}$, and for all $0 \le \eta \le \zeta \le \theta \le 1$, it holds $Q_{\eta \to \theta} = Q_{\eta \to \zeta} Q_{\zeta \to \theta}$ and $P_{\theta \to \eta} = P_{\theta \to \zeta} P_{\zeta \to \eta}$.

\begin{remark}
    The measure-valued stochastic process $\left( Q_{\theta \to 1}(Y_\theta,\cdot) \right)_{\theta \in [0,1]}$ where $(Y_\theta)_{\theta \in [0,1]}$ is the denoising process, corresponds to the \emph{localization process} $(\nu_\theta)_{\theta \in [0,1]}$ in~\cite{CE22}. Specifically, when $\theta = 0$, the measure $Q_{\theta \to 1}(Y_\theta,\cdot) = Q_{0 \to 1}(\*0,\cdot)$ is always the target distribution $\mu$ since $Y_0 = \*0$; and when $\theta = 1$, the measure $Q_{\theta \to 1}(Y_\theta,\cdot) = \*\delta_{Y_1}$ is a random Dirac measure where $Y_1$ is distributed as $\mu$. 
    However, it is non-trivial to rigorously verify that this is indeed a localization process. 
    We remark that in \cite{CE22} the notion of localization processes is defined abstractly and generally, without the presence of noising and denoising processes.
\end{remark}

\paragraph{Densities.}
Recall that $X_0$ is generated from the target distribution $\mu$.
Let $\mu_\theta$ denote the law of $Y_\theta \overset{d}{=} X_{1-\theta}$ (so $\mu_1 = \mu$); that is, $\mu_\theta := \mu_1 P_{1 \to \theta}$.
More generally, it holds 
\begin{align*}
    \mu_\theta P_{\theta \to \eta} = \mu_\eta
    \qandq
    \mu_\eta Q_{\eta \to \theta} = \mu_\theta.
\end{align*}
We remark that $Q_{\eta \to \theta}$ and $P_{\theta \to \eta}$ are adjoint to each other, i.e., for any two integrable functions $f,g: \R^n \to \R$, it holds $\ip{f}{P_{\theta \to \eta} g}_{\mu_\theta} = \ip{Q_{\eta \to \theta} f}{g}_{\mu_\eta}$.

\paragraph{Markov Chains.}
The noising and denoising processes allow us to define a family of discrete-time Markov chains associated with them.
For any $\theta \in [0,1]$, we define the Markov kernel
\begin{align*}
	K_{1 \leftrightarrow \theta} = P_{1 \to \theta} Q_{\theta \to 1}.
\end{align*}
Since $P_{1 \to \theta}$ and $Q_{\theta \to 1}$ are adjoint to each other, the Markov kernel $K_{1 \leftrightarrow \theta}$ is reversible with respect to $\mu = \mu_1$.
Observe that the chain $P_{1 \to 1} = \mathrm{id}$ does not move at all, while the chain $P_{1 \to 0}$ moves to the stationary distribution $\mu$ in a single step.

\paragraph{\underline{\normalfont\textit{Examples of associated Markov chains.}}}
The examples of noising and denoising processes can be used to define families of discrete-time Markov chains as described above.
Here, we do short calculations to see how the continuous-time down-up walks recover the field dynamics \cite{chen2021rapid}, and the Brownian bridge and F{\"o}llmer processes recover the proximal sampler \cite{lee2021structured}.

The continuous-time down-up walk is exactly the field dynamics introduced in \cite{chen2021rapid}.  
We adopt the continuous-time down walk as the noising process starting from $X_0\sim\mu$.
For $\theta \in [0,1]$. The Markov kernel of this process, $P_{1 \to \theta}(S,\cdot)$, independently discards each element $v \in S$ with probability $1-\theta$.  
Correspondingly, the Markov kernel $Q_{\theta \to 1}(S,\cdot)$ of the continuous-time up walk samples a subset $T \supseteq S$ with probability proportional to $(1-\theta)*\mu (T)$, where $(1-\theta)*\mu$ is obtained by magnetizing $\mu$ with local fields $1-\theta$ (see \eqref{eq:def-magnetization}). In particular, we have:
\begin{align}
\forall T \subseteq [n], \quad \tp{(1-\theta)*\mu}(T) \propto (1-\theta)^{\abs{T}} \cdot \mu(T). \label{eq:def-mu-external-field}
\end{align}
Specifically, the Markov kernel for this denoising process is given by:
\begin{align*}
Q_{\theta \to 1}(S,T) \propto \Pr[]{X_0 = T \text{ and } X_{1-\theta} = S} = \mu(T) (1-\theta)^{\abs{T}-\abs{S}} \theta^{\abs{S}},
\end{align*}
which means $Q_{\theta \to 1}(S,\cdot)$ is the distribution $(1-\theta) * \mu$ conditioned on all $v\in S$ being occupied.
The continuous-time down-up walk $K_{1 \leftrightarrow \theta} = P_{1 \to \theta} Q_{\theta \to 1}$ then recovers the field dynamics, formally defined as below.
\begin{definition}[Field dynamics, \cite{chen2021rapid}]\label{def:field-dynamics}
    Let $\mu$ be a distribution over $2^V$, where $V$ is a finite ground set. Let $\theta \in (0,1)$ be a parameter. The field dynamics with parameter $\theta$ is a Markov chain with stationary distribution $\mu$.
    Given a subset $T \in 2^V$, it updates $T$ to a new subset $T'$ as:
    \begin{enumerate}
        \item (Noising Step) Remove each $v \in T$ with probability $1-\theta$ independently, and denote the resulting subset by $S \subseteq T$;
        \item (Denoising Step) Draw a new subset $T' \in 2^V$ from the distribution $(1-\theta) * \mu$ conditioned on $T' \supseteq S$.
    \end{enumerate}
\end{definition}

The Brownian bridge and F{\"o}llmer processes recover the proximal sampler introduced in \cite{lee2021structured}.
%
Consider the Brownian bridge process $(X_t)_{t \in [0,1]}$ with driving matrix $L \in \mathbb{R}^{r\times n}$ that starts from $X_0 \sim \mu$.
For a fixed $\theta \in [0,1)$ and configuration $\*x \in \{-1,+1\}^n$, the noising process $P_{1 \to \theta}(\*x,\cdot)$ first adds Gaussian noise to vector $L \*x$ and then scales it by a factor of $\theta$. The Gaussian noise follows a multivariate Gaussian distribution with covariance matrix $\frac{1-\theta}{\theta} I_r$, where $I_r$ is the $r\times r$ identity matrix. For any state $\*y \in \mathbb{R}^r$, the F\"ollmer process $Q_{\theta \to 1}(\*y,\cdot)$, by definition, satisfies
\begin{align*}
Q_{\theta \to 1}(\*y,\*z) &\propto \mu(\*z)\exp\tp{-\frac{1}{2\theta (1-\theta)} \tp{\theta L \*z-\*y}^\intercal \tp{\theta L \*z - \*y} }\\
&\propto \mu(\*z)\exp\tp{-\frac{\theta}{2(1-\theta)}\*z^\intercal L^\intercal L \*z + \frac{1}{1-\theta} \*y^\intercal L \*z }, \quad \forall \*z \in \XX.
\end{align*}
The Markov kernel $K_{1 \leftrightarrow \theta} = P_{1 \to \theta} Q_{\theta \to 1}$ corresponds to the proximal sampler defined below.
\begin{definition}[Proximal sampler, \cite{lee2021structured}]\label{def:proximal-sampler}
    Let $\mu$ be a distribution over a finite state space $\XX \subseteq \mathbb{R}^n$.
    Let $\theta \in (0,1)$ be a parameter and $L \in \mathbb{R}^{r\times n}$ be a matrix. 
    The proximal sampler with parameter $\theta$ and driving matrix $L$ is a Markov chain with stationary distribution $\mu$.
    Given a vector $\*x \in \XX$, it updates $\*x$ to a new vector $\*z$ as:
    \begin{enumerate}
        \item (Noising Step) Sample $\*y \in \R^r$ from $\+{N}(\theta L\*x, \theta(1-\theta) I_r)$, which is the Gaussian distribution on $\mathbb{R}^r$ with mean $\theta L \*x$ and covariance $\theta(1-\theta) I_r$;
        \item (Denoising Step) Sample a new vector $\*z \in \XX$ from the distribution $\nu$ defined over $\XX$ as:
        \begin{align*}
            \nu(\*z) \propto \mu(\*z) \exp\tp{-\frac{\theta}{2(1-\theta)} \*z^\intercal L^\intercal L\*z + \frac{1}{1-\theta} \*y^\intercal L \*z}, \quad \forall \*z \in \XX.
        \end{align*}
    \end{enumerate}
\end{definition}

We now apply the proximal sampler to the Ising model.
Let $\XX = \{\pm 1\}^n$ and let distribution $\mu(\*x) \propto \exp\tp{\frac{1}{2} \*x^\intercal L^\intercal L \*x + \*h^\intercal \*x}$ for $\*x \in \XX$ be the Gibbs distribution of the Ising model with interaction matrix $J = L^\intercal L$ and external fields $\*h$.
Take the same matrix $L$ as the driving matrix in the proximal sampler.
Then for any $\theta \in (0,1)$ and $\*y \in \R^r$, we have
\begin{align*}
    Q_{\theta \to 1}(\*y, \*z) \propto \exp\tp{\frac{1}{2} \tp{1-\frac{\theta}{1-\theta}} \*z^\intercal L^\intercal L \*z + \*h_\star^\intercal \*z}, \quad \forall \*z \in \{\pm 1\}^n
\end{align*}
where $\*h_\star = \*h + \frac{1}{1-\theta} L^\intercal \*y$; 
that is, $Q_{\theta \to 1}(\*y,\cdot)$ is the Gibbs distribution of another Ising model.
In particular, if we take $\theta = 1/2$ then $Q_{\theta \to 1}(\*y, \*z) \propto \exp\tp{\*h_\star^\intercal \*z}$ becomes a product distribution; this essentially recovers the Hubbard--Stratonovich transformation, see e.g. \cite{AKV24}.
Therefore, under this setting the proximal sampler takes a pleasant form.
From $\*x \in \{\pm 1\}^n$ it moves to $\*z \in \{\pm 1\}^n$ by:
\begin{enumerate}
    \item (Noising Step) Sample $\*y \in \R^r$ from $\+{N}(L\*x, I_r)$;
    \item (Denoising Step) Sample a new vector $\*z \in \{\pm 1\}^n$ from the product distribution $\nu$ defined as:
        \begin{align*}
            \nu(\*z) \propto \exp\tp{(\*h + L^\intercal \*y)^\intercal \*z}, \quad \forall \*z \in \{\pm 1\}^n.
        \end{align*}
\end{enumerate}

\subsubsection{Approximate conservation of entropy}
\label{subsubsec:app-conserv-ent}

The framework of localization scheme provides a rich toolbox for analyzing the mixing time of Markov chains.
In this section, we will introduce these tools aligning with the context of information-theoretic view for localization scheme.

\paragraph{Relative Entropy.}
Let $f: \XX \to \R_{\ge 0}$ be a non-negative function. 
Then, for a convex function $\phi:\mathbb{R} \to \mathbb{R}$,  the (relative) \emph{$\phi$-entropy} functional is defined as:
\begin{align*}
    \PhiEnt[\mu]{f} := \textstyle \E[\mu]{ \phi(f) } - \phi(\E[\mu]{f}).
\end{align*}

\begin{remark}
    Typical $\phi$-entropy functionals are variance and relative entropy.
    Specifically:
    \begin{itemize}
        \item When $\phi(x) = x^2$, the $\phi$-entropy corresponds to the variance functional $$\PhiEnt[\mu]{f}=\Var[\mu]{f} := \E[\mu]{f^2} - \E[\mu]{f}^2.$$
        \item When $\phi(x) = x\log x$, the $\phi$-entropy  corresponds to relative entropy functional $$\PhiEnt[\mu]{f}=\Ent[\mu]{f} := \E[\mu]{f\log f} - \E[\mu]{f}\log \E[\mu]{f}.$$
    \end{itemize}
\end{remark}

For $\theta \in [0,1]$, define $f_\theta: \XX \to \R_{\ge 0}$ as $f_\theta := Q_{\theta \to 1} f$ (so $f_1 = f$); that is, $f_\theta(x) = \E[Q_{\theta \to 1}(x,\cdot)]{f}$.
Observe that $\E[\mu_\theta]{f_\theta} = \E[\mu]{f}$. 
More generally, for all $0 \le \eta \le \theta \le 1$, it holds that 
\begin{align*}
    Q_{\eta \to \theta} f_\theta = f_\eta
    \qandq
    P_{\theta \to \eta} f_\eta = f_\theta.
\end{align*}

The following law of total entropy for 
$\PhiEnt[\mu_\theta]{f_\theta}=\textstyle \E[\mu_\theta]{ \phi(f_\theta) } - \phi(\E[\mu_\theta]{f_\theta})$
plays a key role.

\begin{lemma}[Law of Total Entropy]
\label{lem:total-ent}
    For any $0 \le \eta \le \theta \le 1$, it holds
    \begin{align*}
        \PhiEnt[\mu_\theta]{f_\theta} = \PhiEnt[\mu_\eta]{f_\eta} + \E[x \sim \mu_\eta] { \PhiEnt[Q_{\eta \to \theta}(x,\cdot)]{f_\theta} }.
    \end{align*}
\end{lemma}


\paragraph{Entropy Conservation and Tensorization.}
A key aspect of the theory of localization scheme is its ability to establish an approximate conservation of entropy.

\begin{definition}[Approximate Conservation of Entropy] \label{def:entropy-conservation}
    Given a time $\theta\in[0,1]$ and a parameter $K>0$, 
    a noising process $(X_t)_{t\in [0,1]}$ with $X_0\sim \mu$ is said to satisfy
    \emph{$K$-approximate conservation of $\phi$-entropy from time $\theta$} (or equivalently, the corresponding denoising process $(Y_\theta)_{\theta\in [0,1]}$ satisfies \emph{$K$-approximate conservation of $\phi$-entropy up to time $\theta$}) 
    if, for all $f: \+X \to \=R_{\geq 0}$, it holds
    \begin{align*}
        \PhiEnt[\mu]{f} \leq K\cdot \E[x\sim \mu_\theta]{\PhiEnt[Q_{\theta \to 1}(x,\cdot)]{f}},
    \end{align*}
    where 
    $Q_{\cdot\to\cdot}$ represents the Markov kernel of the corresponding denoising process $(Y_\theta)_{\theta \in [0,1]}$.
\end{definition}

The approximate conservation of $\phi$-entropy, when established for suitable $\phi$,
can imply the Poincar\'{e} or modified log-Sobolev inequalities for the Markov chains 
$K_{1 \leftrightarrow \theta} = P_{1 \to \theta} Q_{\theta \to 1}$,
thereby leading to rapid mixing of these chains.
For the specific noising and denoising processes we consider,
$K_{1 \leftrightarrow \theta}$  may be either the field dynamics or the proximal sampler.
Therefore, the corresponding approximate conservation of $\phi$-entropy may imply the rapid mixing of these samplers.

\begin{lemma} \label{lem:entropy-decay}
  Let $\phi: \mathbb{R} \to \mathbb{R}$ be a convex function.
  Suppose there exists a noising process $(X_t)_{t \in [0,1]}$ such that $X_0 \sim \mu$, and for a time $\theta \in [0, 1]$, the process $(X_t)_{t\in [0,1]}$ satisfies $(1/\kappa)$-approximate conservation of $\phi$-entropy from time $\theta$.
  Then, the corresponding Markov kernel $K_{1\leftrightarrow\theta} = P_{1\to\theta}Q_{\theta\to 1}$ admits $\phi$-entropy decay with rate $\kappa$.
  That is, for every function $f:\XX \to \mathbb{R}_{\geq 0}$, we have
  \begin{align}
    \PhiEnt[\mu]{K_{1\leftrightarrow\theta} f} \leq (1 - \kappa) \PhiEnt[\mu]{f}.
  \end{align}
\end{lemma}
\Cref{lem:entropy-decay} is a simple consequence of the law of total entropy.
To prove \Cref{lem:entropy-decay}, we will consider $\phi$-entropy with respect to random variables.
Let $F, G$ be two random variables on the same probability space such that $F$ is $R_{\geq 0}$-valued.
We can define the $\phi$-entropy and conditional $\phi$-entropy on $F$ and $G$ as follows:
\begin{align}
  \label{eq:def-phi-ent-variable}
    \PhiEnt{F} &:= \E{\phi(F)} - \phi(\E{F}), \\
  \label{eq:def-phi-ent-variable-cond}
    \PhiEnt{F\mid G} &:= \E{\phi(F) \mid G} -  \phi(\E{F\mid G}).
\end{align}
\begin{proof}[Proof of \cref{lem:entropy-decay}]
  Let $X_0 \sim \mu$, $Y_\theta \sim P_{1\to \theta}(X_0, \cdot)$ and $Y_1 \sim Q_{\theta \to 1}(Y_\theta, \cdot)$ be random variables in the same probability space.
  Let $F = f(Y_1)$, and by the conservation of entropy, we have
  \begin{align*}
    \PhiEnt{F}
    &\leq \frac{1}{\kappa} \cdot \E{\PhiEnt{F \mid Y_\theta}} \\
    &= \frac{1}{\kappa} \tp{\PhiEnt{F} - \PhiEnt{\E{F \mid Y_\theta}}}, \tag{law of total entropy}
  \end{align*}
  which, after rearranging terms, implies
  \begin{align*}
    \PhiEnt{\E{F \mid Y_\theta}} \leq (1 - \kappa) \PhiEnt{F}.
  \end{align*}
  By another application of the law of total entropy, we have
  \begin{align*}
    \PhiEnt{\E{F\mid Y_\theta}}
    &= \E{\PhiEnt{\E{F\mid Y_\theta} \mid X_0}} + \PhiEnt{\E{\E{F\mid Y_\theta}\mid X_0}} \\
    &\geq \PhiEnt{\E{\E{F\mid Y_\theta}\mid X_0}} = \PhiEnt[\mu]{K_{1\leftrightarrow\theta} f}, \tag{non-negativity of $\phi$-entropy}
  \end{align*}
  where the last inequality is also known as the data processing inequality in information theory.
\end{proof}

When applied to the analysis of Glauber dynamics, 
another important consequence of the approximate conservation of entropy is a ``lifting'' of the mixing properties in a sub-critical regime up to the critical threshold. 
In particular, these mixing properties are captured by the abstraction of approximate tensorization of $\phi$-entropy,
which, with appropriate choices of $\phi$, corresponds to the Poincar\'{e} or modified log-Sobolev inequalities for Glauber dynamics.

\begin{definition}[Approximate Tensorization of Entropy] \label{def:approximate-tensorization}
    Given a parameter $C > 0$, a distribution $\mu$, supported on $\+X \subseteq \=R^n$, is said to satisfy \emph{$C$-approximate tensorization of $\phi$-entropy} if, for all $f: \+X \to \=R_{\geq 0}$, it holds
    \begin{align*}
        \PhiEnt[\mu]{f} \leq C \sum_{i=1}^n \E[x\sim \mu]{\PhiEnt[\mu^{x_{- i}}]{f}},
    \end{align*}
    where $x_{-i}$ denotes the partial configuration $x([n]\setminus \set{i})$.
\end{definition}

The notion of approximate tensorization of $\phi$-entropy  is closely related to the mixing time of Glauber dynamics.
\begin{itemize}
    \item \textbf{Variance case}:
    When $\phi(x) = x^2$, 
    the $\chi^2$-entropy corresponds to variance. 
    In this context, the approximate tensorization of variance with parameter $C$ is equivalent to the well-known Poincar{\'e} inequality with constant $C$, which leads to a $1/(Cn)$ lower bound on the spectral gap of the Glauber dynamics.
\item \textbf{Entropy case}:
When $\phi(x) = x\log x$, where the $\phi$-entropy  corresponds to the relative entropy.
Here, the approximate tensorization of entropy indicates a decay of relative entropy at a rate of $1/(Cn)$.
This decay of entropy further implies an upper bound on the mixing time of the Glauber dynamics of the form $C n (\log\log 1/\mu_{\min})$, where $\mu_{\min} = \min_{x\in\XX} \mu(x)$. 
\end{itemize}

The following lemma formalizes a ``lifting'' of the approximate tensorization of $\phi$-entropy within the context of the noising-denoising framework  we adopt. 

\begin{lemma} \label{lem:annealing}
Let $\phi:\=R \to \=R$ be a convex function.
    Suppose there exists a noising process $(X_t)_{t\in [0,1]}$  such that $X_0\sim\mu$,
    and for a time  $\theta \in [0, 1]$, the following conditions hold: 
    \begin{enumerate}
        \item \label{item:conservation} The noising process $(X_t)_{t\in[0,1]}$ satisfies $K$-approximate conservation of $\phi$-entropy from time~$\theta$.
        \item  \label{item:tensorization} The distribution $Q_{\theta \to 1}(X_{1-\theta}, \cdot)$ always satisfies $C$-approximate tensorization of $\phi$-entropy.
    \end{enumerate}
    Then, the distribution $\mu$ satisfies $CK$-approximate tensorization of $\phi$-entropy.
\end{lemma}
This ``lifting'' phenomenon has been observed in various contexts in prior work.
Notably, it includes the approximate tensorization of entropy through uniform block factorization of entropy \cite[Lemma 2.3]{chen2021optimal}, the boosting theorem of field dynamics~\cite[Theorem 1.9]{chen2021rapid}, and annealing via stochastic localization~\cite[Section 4]{CE22}.
All these  results can be unified within the abstract framework of localization schemes (see~\cite{CE22}) and are thereby  encompassed by \Cref{lem:annealing}.

\begin{proof}[Proof of \Cref{lem:annealing}]
We re-express the conditions of \Cref{lem:annealing} using the $\phi$-entropy for random variables defined in~\eqref{eq:def-phi-ent-variable} and~\eqref{eq:def-phi-ent-variable-cond}.

Let $f:\XX \to \mathbb{R}_{\geq 0}$ be a function, and let $(Y_\theta)_{\theta \in [0,1]}$ be the corresponding denoising process, which is the time-reversed process of the noising process $(X_t)_{t\in[0,1]}$.
Note that $Y_1 \sim \mu$.
We define $F := f(Y_1)$ and  $\+P_i := Y_1([n]\setminus\set{i})$ for each $i \in [n]$.
Thus, $F, \+P_i, Y_\theta$ are random variables on the same probability space.
In particular, we have $Y_\theta \sim \mu_\theta$ and $(Y_1 \mid Y_\theta) \sim Q_{\theta\to 1}(Y_\theta, \cdot)$.

Then, the approximate conservation of $\phi$-entropy, as required in \Cref{item:conservation} of \Cref{lem:annealing}, can be expressed as 
\begin{align} \label{eq:conservation-of-ent-new}
    \PhiEnt{F} &\leq K\cdot \E{\PhiEnt{F \mid Y_\theta}}.
\end{align}
Similarly, the approximate tensorization of $\phi$-entropy, as required in \Cref{item:tensorization} of \Cref{lem:annealing},  means that the following always holds:
\begin{align} \label{eq:tensorization-of-ent-new}
 \PhiEnt{F \mid Y_\theta} \leq C \sum_{i=1}^n \E{\PhiEnt{F \mid \+P_i, Y_\theta} | Y_\theta}.
\end{align}
    Plugging \eqref{eq:tensorization-of-ent-new} into \eqref{eq:conservation-of-ent-new}, we have
    \begin{align}
       \PhiEnt{F}
        \nonumber &\leq CK  \sum_{i=1}^n \E{\E{\PhiEnt{F \mid \+P_i, Y_\theta} | Y_\theta}} \\
\label{eq:AT-raw} &= CK \sum_{i=1}^n \E{\E{\PhiEnt{F \mid \+P_i, Y_\theta} \mid \+P_i}}, 
    \end{align}
    where the equality follows from the law of total expectation.
   
    We apply the law of total entropy (\Cref{lem:total-ent}) to obtain:
    \begin{align} \label{eq:AT-refined} 
        \E{\PhiEnt{F \mid Y_{\theta}, \+P_i} \mid \+P_i} 
         &= \PhiEnt{F \mid \+P_i} - \PhiEnt{\E{F \mid Y_{\theta}, \+P_i} \mid \+P_i}
        \leq \PhiEnt{F \mid \+P_i},
    \end{align}
    where the last inequality follows from the non-negativity of $\phi$-entropy.
    
    Finally, we conclude the proof by combining \eqref{eq:AT-raw} and \eqref{eq:AT-refined}.
\end{proof}

\subsubsection{Entropic stability}
\label{subsubsec:ent-stab}
The following concept of entropic stability quantifies the rate of local conservation of $\phi$-entropy in a noising/denoising process.
\begin{definition}[Entropic Stability]
Given a function  $C: [0,1] \to \R_{\ge 0}$,
we say the target distribution $\mu$ is \emph{$\phi$-entropically stable with rate $C$} with respect to a noising process $(X_t)_{t\in [0,1]}$ where $X_0 \sim \mu$ if, for any $\eta \in [0,1]$, $x \in \XX$ and $f: \XX \to \R_{\ge 0}$, 
\begin{align} \label{eq:def-ent-stability}
    \lim_{\theta \to \eta^+} \frac{1}{\theta-\eta} \PhiEnt[Q_{\eta \to \theta}(x,\cdot)]{f_\theta}
    \le \frac{C(\eta)}{1-\eta} \cdot \PhiEnt[Q_{\eta \to 1}(x,\cdot)]{f},
\end{align}
where $Q_{\cdot\to\cdot}$ represents the Markov kernel of the corresponding denoising process $(Y_\theta)_{\theta\in [0,1]}$, which is the time-reversed process of $(X_t)_{t\in [0,1]}$.
\end{definition}

\begin{remark}
    When $\phi = x^2$, $\mu$ is said to be \emph{spectrally stable} with rate $C$; when $\phi = x \log x$, $\mu$ is said to be \emph{entropically stable} with rate $C$.
\end{remark}

The notion of the entropic stability is closely related to the convexity of the $\phi$-entropy functional.
Let $\+E(\theta) = \PhiEnt[\mu_\theta]{f_\theta}$ for $\theta \in [0,1]$.
If we take expectation for $x\sim \mu_\eta$ on both side of \eqref{eq:def-ent-stability}, then, by the law of total entropy, we have
\begin{align*}
  \+E'(\eta) \leq C(\eta) \cdot \frac{\+E(1) - \+E(\eta)}{1 - \eta},
\end{align*}
which is equivalent to the convexity of function $\+E$ when $C(\eta) = 1$.
In this perspective, \eqref{eq:def-ent-stability} can be seen as a relaxed version of convexity.

\begin{theorem}
\label{thm:CTEI-FD}
If the target distribution $\mu$ is $\phi$-entropically stable with rate $C$ with respect to a noising process $(X_t)_{t\in [0,1]}$ where $X_0\sim \mu$,
then for every $\theta \in [0,1]$, 
it satisfies $K$-approximate conservation of $\phi$-entropy from time $\theta$,
where
$$K=\exp\left(\int_0^\theta \frac{C(\eta)}{1-\eta} \-d \eta\right).$$
That is,
for all $f: \XX \to \R_{\ge 0}$, it holds 
\begin{align}\label{eq:CTEI-FD}
\PhiEnt[\mu]{f} \le  
\exp\left( \int_0^\theta \frac{C(\eta)}{1-\eta} \dif \eta \right) 
\cdot\E[x \sim \mu_\theta]{ \PhiEnt[Q_{\theta \to 1}(x,\cdot)]{f} }.
\end{align}
where 
$Q_{\cdot\to\cdot}$ represents the Markov kernel of the corresponding denoising process $(Y_\theta)_{\theta\in [0,1]}$.
\end{theorem}

\begin{proof}
Fix $f: \XX \to \R_{\ge 0}$ with $\E[\mu]{f} = 1$.
Let us define for $\theta \in [0,1]$ that 
\begin{align*}
    \EE(\theta) = \PhiEnt[\mu_\theta]{f_\theta}.
\end{align*}
In particular, $\EE(1) = \PhiEnt[\mu]{f}$ and $\EE(0) = 0$. 
Also, $\EE$ is monotone increasing by the data processing inequality (this can be easily derived from \cref{lem:total-ent}). 
For any $\eta \in [0,1)$ we deduce that
\begin{align*}
    \EE'(\eta)
    &= \lim_{\theta \to \eta^+} \frac{1}{\theta-\eta} \left( \PhiEnt[\mu_\theta]{f_\theta} - \PhiEnt[\mu_\eta]{f_\eta} \right) \\
    &= \lim_{\theta \to \eta^+} \frac{1}{\theta-\eta} \,\E[x \sim \mu_\eta]{ \PhiEnt[Q_{\eta \to \theta}(x,\cdot)]{f_\theta} } \tag{\cref{lem:total-ent} (Law of total entropy)} \\
    &= \E[x \sim \mu_\eta]{ \lim_{\theta \to \eta^+} \frac{1}{\theta-\eta} \PhiEnt[Q_{\eta \to \theta}(x,\cdot)]{f_\theta} } \\
    &\le \E[x \sim \mu_\eta]{ \frac{C(\eta)}{1-\eta} \PhiEnt[Q_{\eta \to 1}(x,\cdot)]{f} } \tag{$(C,\phi)$-Entropic Stability} \\
    &= \frac{C(\eta)}{1-\eta} \left( \PhiEnt[\mu_1]{f_1} - \PhiEnt[\mu_\eta]{f_\eta} \right) \tag{\cref{lem:total-ent} (Law of total entropy)} \\
    &= \frac{C(\eta)}{1-\eta} \left( \EE(1) - \EE(\eta) \right).
\end{align*}
We can apply Gr\"{o}nwall's inequality to obtain the theorem, or more directly
\begin{align*}
    \log\left( \frac{\EE(1)}{\EE(1)-\EE(\theta)} \right) 
    = \int_0^\theta \left( \frac{\dif}{\dif \eta} \log\left( \frac{1}{\EE(1)-\EE(\eta)} \right) \right) \dif \eta 
    = \int_0^\theta \frac{\EE'(\eta)}{\EE(1) - \EE(\eta)} \dif \eta 
    \le \int_0^\theta \frac{C(\eta)}{1-\eta} \dif \eta.
\end{align*}
Observe $\EE(1)-\EE(\theta) = \E[x \sim \mu_\theta]{ \PhiEnt[Q_{\theta \to 1}(x,\cdot)]{f} }$ by the law of total entropy. The theorem  follows.
\end{proof}

\subsection{Hardcore model and field dynamics}
\label{subsec:hardcore-upper}
We now apply the abstract framework developed in the previous sections to analyze the Glauber dynamics for the Gibbs distribution of the critical hardcore model.
The noising process $(X_t)_{t \in [0,1]}$ and the denoising process $(Y_\theta)_{\theta\in[0,1]}$ considered in this context is the \emph{continuous-time down/up walk} (see \Cref{definition-decreasing-process}).
In particular, for $\eta\in[0,1]$, recall the distribution $(1-\eta)*\mu$ as defined in \eqref{eq:def-mu-external-field}.

The following lemma establishes a connection between $\chi^2$-entropic (i.e., spectral) stability with respect to the continuous-time down walk and spectral independence. We refer the readers to a recent work \cite{CCCYZ25+} for more discussion on this topic.
\begin{lemma}\label{lem:variance-stability-SI}
	Let $\mu$  be a distribution over $2^{[n]}$.
    If there exists a function $C:[0,1]\to \mathbb{R}_{\ge 0}$ such that for every $\eta \in [0,1]$, the distribution $(1-\eta) * \mu$ is $C(\eta)$-spectrally independent for all pinning, 
    then $\mu$ is spectrally stable with rate $C$ with respect to the continuous-time down walk $(X_t)_{t \in [0,1]}$ where $X_0 \sim \mu$.
\end{lemma}
The proof of~\Cref{lem:variance-stability-SI} is deferred to~\Cref{sec:append-variance}.

The following bounds on spectral independence and the approximate tensorization of variance (i.e., a Poincar\'{e} inequality) have been  established for the Gibbs distribution of the hardcore model.

\begin{lemma}[\cite{chen2023rapid}]\label{lem:sub-SI-hardcore}
    For any $\delta\in(0,1)$ and $\Delta \ge 3$, 
    the Gibbs distribution $\mu$ of the hardcore model on any graph with a maximum degree at most $\Delta$ and fugacity $\lambda \le (1-\delta) \lambda_c(\Delta)$ is $\frac{4\e (1+\frac{1}{\Delta-2})}{\delta}$-spectrally independent for all pinning. Formally, for any pinning $\tau$, the influence matrix satisfies
    \begin{align*}
    \lambda_{\max}\tp{\Psi_{\mu^\tau}} \le \frac{4\e(1+\frac{1}{\Delta-2})}{\delta}.
    \end{align*}
\end{lemma}
We remark that a crude bound ($32/\delta$-spectral independence) was established in~\cite{chen2023rapid}.
Meanwhile, it is direct to get the bound stated in~\Cref{lem:sub-SI-hardcore} by going through the same proof with a more careful analysis.
We leave important details into~\Cref{sec:SI-hardcore}.

\begin{lemma}[\cite{chen1998trilogy}]\label{lem:sub-PI-hardcore}
For any $\Delta \ge 3$, 
the Gibbs distribution $\mu$ of the hardcore model on any graph $G=(V,E)$ with a maximum degree at most $\Delta$ and fugacity $\lambda \le \frac{1}{2\Delta}$ satisfies $C$-approximate tensorization of variance with constant $C=2$ (see \Cref{def:approximate-tensorization} for a formal definition).
\end{lemma}
\begin{proof}[Proof of the upper bound in~\Cref{thm:hardcore}]
    Let $G=(V,E)$ be a graph with sufficiently many vertices, specifically $n=|V| \ge 1000$, and maximum degree  $\Delta \ge 3$. 
    Let $\mu$ be the Gibbs distribution of the hardcore model on $G$ with fugacity $\lambda=\lambda_c(\Delta)$.
    By~\Cref{thm:CTEI-FD}, for all $\theta \in [0,1]$, the distribution $\mu_\theta$  satisfies 
    \begin{align*}
    \forall f: \+I(G) \to \mathbb{R}_{\ge 0}, \quad \Var[\mu]{f} \le \exp\tp{\int_{0}^\theta \frac{C(\eta)}{1-\eta} \dif \eta} \cdot \E[x \sim \mu_\theta]{\Var[Q_{\theta \to 1}(x,\cdot)]{f}}.
    \end{align*}
    Note that for any independent set $I$ in $G$, the distribution $Q_{\eta \to 1}(I,\cdot)$ is the Gibbs distribution of hardcore model with fugacity $\lambda_\eta = (1-\eta) \lambda_c(\Delta)$. 
    Therefore, by~\Cref{lem:variance-stability-SI}, 
	$\mu$ is spectrally stable with rate $C$ with respect to the continuous-time down walk $(X_t)_{t \in [0,1]}$ where $X_0\sim\mu$, 
    where $C:[0,1] \to \mathbb{R}_{\ge 0}$ satisfies
    \begin{align*}
        \forall \eta \in [0,1],\quad C(\eta) = \min\tp{\frac{c^\star}{\eta},n},
    \end{align*}
    with $c^\star = 4\e (1+\frac{1}{\Delta-2})$. The bound $C(\eta) \le n$ follows from the crude upper bound $\lambda_{\max}(\Psi_{\mu}) \le n$ for spectral independence. Hence, when $\theta > c^\star / n$, the approximate conservation of variance holds with constant
    \begin{align*}
        \exp\tp{\int_{0}^\theta \frac{C(\eta)}{1-\eta} \dif \eta} &= \exp\tp{n \cdot \int_{0}^{c^\star/n} \frac{1}{1-\eta} \dif \eta + \int_{c^\star/n}^\theta \frac{c^\star}{\eta (1-\eta)} \dif \eta} \\
        &\le \exp\tp{2c^\star+  c^\star \log \frac{\theta}{1-\theta} + c^\star \log \frac{n}{c^\star}}.
    \end{align*}
    Fix $\theta = \frac{9}{10}$. It satisfies $\exp\tp{\int_0^{\theta} \frac{C(\eta)}{1-\eta} \dif \eta} = O\tp{n^{c^\star}}$. Combining~\Cref{lem:sub-PI-hardcore} and~\Cref{lem:annealing}, the distribution $\mu$ satisfies the approximate tensorization of variance with constant $O\tp{n^{c^\star}}$.
    
    According to~\Cref{lem:tensorization-implies-mixing}, this result gives an $\Omega(1/n^{1+c^\star})$ lower bound on the spectral gap of the Glauber dynamics, which implies an  $O\tp{n^{2+c^\star}}$ upper bound on the mixing time.
\end{proof}

\begin{remark}
    Although our proof of the mixing time for the critical hardcore model (\Cref{thm:hardcore}) is constructed within the general framework of the (information-theoretical view of the) localization scheme, 
    it is noteworthy that  a similar (though cruder) polynomial upper bound can also be obtained by applying the original machinery in~\cite{chen2021rapid} without resorting to localization schemes.
    
    Given a graph $G$ with maximum degree $\Delta$, let $\mu_{\lambda}$ be the Gibbs distribution of the hardcore model on $G$ with fugacity $\lambda$.
    Let $\-{gap}(\mu_\lambda)$ denote the spectral gap of the Glauber dynamics for $\mu_\lambda$.
    Applied to the hardcore model, \cite[Theorem 1.9]{chen2021rapid} guarantees that if $\lambda \leq (1-\delta)\lambda_c(\Delta)$, then for every $\theta \in (0,1)$, the following inequality holds\footnote{The base of the exponential in \cite[Theorem 1.9]{chen2021rapid} is $\theta/2$. But it can be improved to $\theta$ by going through the proof carefully.
      The extra factor $1/2$ arises primarily from \cite[Theorem 2.5]{chen2021rapid}, which provides a crude lower bound of $(\frac{\ell}{2n})^{2\ctp{\eta}+1}$ for the quantity $(\frac{\ell-\ctp{\eta}}{n - \ctp{\eta}})^{2\ctp{\eta} +1}$.
    From the setting of \cite[Theorem 1.9]{chen2021rapid}, we have $\ell/n = \theta$, and both $\ell$ and $n$ are made sufficiently large.
    Therefore, this bound can be improved.
    }
    \begin{align} \label{eq:boosting}
        \-{gap}(\mu_\lambda) \geq \theta^{C(\delta)} \min_\tau \-{gap}(\mu^\tau_{\theta\lambda}),
    \end{align}
    where $\mu_{\theta\lambda}$ is the Gibbs distribution of the hardcore model on $G$ with fugacity $\theta\lambda$; the minimum is taken over all possible pinnings $\tau$; and $C(\delta)$ is a function of $\delta$ (roughly $\min\set{n, c^\circ/\delta}$ for some universal constant $c^\circ > c^\star$).
    This implies that we can define a function $f(t) := \min_\tau \-{gap}(\mu^\tau_{(1-t)\lambda_c})$, and from \eqref{eq:boosting}, we have
    \begin{align*}
        \limsup_{h \to 0^+} \frac{\log f(t+h) - \log f(t)}{h} \leq \limsup_{h\to 0^+} \frac{C(t) \log \frac{1-t}{1-t-h}}{h} = \frac{C(t)}{1-t}.
    \end{align*}
    Let $\theta = 9/10$, by \Cref{lem:sub-PI-hardcore}, we have $f(\theta) \geq (2n)^{-1}$.
    By ``integrating'' w.r.t. $\log f$, we have
    \begin{align*}
        \-{gap}(\mu_{\lambda_c}) \geq f(0) \geq f(\theta) \exp\tp{- \int_0^\theta \frac{C(t)}{1-t}\-d t} = \frac{1}{\-{poly}(n)}.
    \end{align*}
\end{remark}

\subsection{Ising model and proximal sampler}
\label{subsec:Ising-upper}
We now move on to the critical Ising model.
The noising/denoising process considered in this context is the \emph{Brownian bridge process} and \emph{F{\"o}llmer process}, which are defined in \Cref{definition-brownian-bridge}.

The following theorem establishes the approximate tensorization of entropy for the Gibbs distribution of the Ising model via spectral independence.
\begin{theorem}[\text{\cite[Theorem 49]{CE22}}]
\label{thm:Ising-integral}
    Let $J \in \mathbb{R}^{n \times n}$ be a positive semi-definite matrix.
    If there exists a function $C:[0,1] \to \mathbb{R}$ such that for all external fields $\*h^\star \in \mathbb{R}^n$ and $t \in [0,1]$, 
    the distribution $\mu_{(1-t)J,\*h^\star}$ is $C(t)$-spectrally independent,
    then for any external fields $\*h \in \mathbb{R}^n$, 
    the Gibbs distribution $\mu = \mu_{J,\*h}$ satisfies the approximate tensorization of entropy with constant $\exp\tp{\norm{J}_{2} \int_0^1 C(t) \dif t}$.\footnote{We remark that the integral in~\cite[Theorem 49]{CE22} differs slightly from the one presented here. The distinction arises from the different definitions of the Ising model. In their work, they define the Ising model specified by an interaction matrix $J$ and external fields $\*h$ as
    $$\forall \*x \in \{-1,+1\}^n,\quad \mu(\*x) \propto \tp{\*x^\intercal J \*x + \*h^\intercal \*x}.$$} 
\end{theorem}

\Cref{thm:Ising-integral} was originally proved in \cite{CE22} using the framework of stochastic localization, 
with the proof relying extensively on stochastic calculus. 
Here, we re-establish the result through our noising/denoising interpretation of localization schemes, 
providing a more elementary proof.
To this end, we introduce the following lemma, which connects entropic stability with spectral independence.

\begin{lemma}\label{lem:EI-entropic-stability}
    Let $J = L^\intercal L \in \mathbb{R}^{n \times n}$ be a positive semi-definite matrix. 
    If there exists a function $C:[0,1) \to \mathbb{R}_{\ge 0}$ such that for all external fields $\*h \in \mathbb{R}^{n}$ and $\eta \in [0,1)$,
    the distribution $\mu_{J_\eta,\*h}$, where $J_\eta = (1-\frac{\eta}{1-\eta}) J$,  is $C(\eta)$-spectrally independent, 
    then 
    for any $\*h \in \mathbb{R}^n$ the distribution $\mu_{J,\*h}$ is entropically stable with rate $\frac{\norm{J}_{2}}{1-\eta} \cdot C $ with respect to the Brownian bridge process $(X_t)_{t \in [0,1]}$ with driving matrix $L$ where $X_0 \sim \mu_{J,\*h}$.
\end{lemma}
The proof of~\Cref{lem:EI-entropic-stability} is deferred to~\Cref{sec:append-entropy}.

\begin{proof}[Proof of~\Cref{thm:Ising-integral}]
Let $(X_t)_{t \in [0,1]}$ be the Brownian bridge process starting from the Gibbs distribution $\mu=\mu_{J,\*h}$ with driving matrix $L \in \R^{r \times n}$, where $J = L^\intercal L$. 
By~\Cref{lem:EI-entropic-stability}, $\mu$ is entropically stable with rate $\frac{\norm{J}_{2}}{1-\eta} \cdot C^\star$ with respect to $(X_t)_{t \in [0,1]}$, where the function $C^\star:[0,1/2] \to \mathbb{R}_{\ge 0}$ is given by:
\begin{align*}
C^\star(\eta) = C\tp{\frac{\eta}{1-\eta}}.
\end{align*}
Let $\theta = 1/2$. By~\Cref{thm:CTEI-FD}, we have the following conservation of entropy:
    \begin{align*}
        \forall f :\{-1,+1\}^n \to \mathbb{R}_{\ge 0}, \quad \Ent[\mu]{f} &\le \exp\tp{\norm{J}_2 \int_{0}^{\theta} \frac{C^\star(\eta)}{1-\eta} \dif \eta} \cdot \E[x \sim \mu_{\theta}]{\Ent[Q_{\theta \to 1}(x,\cdot)]{f}}\\
        &=\exp\tp{\norm{J}_{2}\int_{0}^{\theta} \frac{C\tp{\eta/(1-\eta)}}{(1-\eta)^2} \dif \eta} \cdot \E[x \sim \mu_{\theta}]{\Ent[Q_{\theta \to 1}(x,\cdot)]{f}}\\
        &= \exp\tp{\norm{J}_{2} \int_0^1 C(t) \dif t}\cdot \E[x \sim \mu_{\theta}]{\Ent[Q_{\theta \to 1}(x,\cdot)]{f}}.
    \end{align*}
Recall that for all $x \in \mathbb{R}^n$, the distribution $Q_{\theta \to 1}(x,\cdot)$ is a product distribution over $\{-1,+1\}^n$ when $\theta = 1/2$ (see the discussion after \cref{def:proximal-sampler}), which exhibits the approximate tensorization of entropy with constant $1$. Together with~\Cref{lem:annealing}, we conclude the proof of~\Cref{thm:Ising-integral}.
\end{proof}

\subsubsection{Mixing at the critical temperature}

\fixed{To prove the mixing time upper bound in~\Cref{thm:Ising-graphical}, we need the following spectral independence result for sub-critical Ising models.}

\begin{lemma}[\cite{chen2023rapid}]
	\label{lem:Ising-graphical-SI}
Let $G$ be a graph with $n$ vertices and maximum degree $\Delta \ge 3$, 
and let $\mu$ be the Gibbs distribution of the Ising model on $G$ with inverse temperature $\beta \in \mathbb{R}$ and external fields $\*h \in \mathbb{R}^V$. 
\fixed{If $(\Delta - 1) \tanh \abs{\beta} \le 1-\delta$ for some $\delta > 0$, then  $\mu$ is $\frac{\Delta}{\Delta-1}\frac{1}{\delta}$-spectrally independent.}
\end{lemma}


\fixed{We now proceed to prove the upper bound in \Cref{thm:Ising-graphical}.}

\fixed{
	\begin{proof}[Proof of the upper bound in~\Cref{thm:Ising-graphical}]
		 By symmetry, we only consider the $\beta > 0$ case.
  When $\beta < 0$, the theorem can be proved by essentially the same proof.
  Let $G=(V,E)$ be a graph with maximum degree $\Delta$ and $n$ vertices.
  Let $\mu$ be the Gibbs distribution of the Ising model on $G$ with inverse temperature satisfying $\beta = \beta_c$ and with arbitrary fields $\*h \in \mathbb{R}^n$. By definition, the interaction matrix $J$ is given by $J = \beta_c \tp{\Delta I + A_G} \succeq 0$, where $A_G$ is the adjacency matrix of $G$.
  By~\Cref{thm:Ising-integral} and~\Cref{lem:Ising-graphical-SI},  and applying a change of variables, the distribution $\mu$ satisfies the approximate tensorization of entropy with constant $\exp\tp{\norm{J}_{2} \int_0^{1} C(t) \dif t}$, where $C:[0,1] \to \mathbb{R}_{\ge 0}$ is given by
    \begin{align*}
        \forall t \in [0,1],\quad C(t) = 
            \min\tp{\frac{\Delta}{\Delta-1} \cdot \frac{1}{1-(\Delta-1) \tanh \tp{\beta_c t}}, n}.
    \end{align*}
    By a straightforward calculations, we have
    \begin{align}\label{eq:integral-partial}
        \nonumber \frac{1}{1 - (\Delta-1) \tanh \tp{\beta_c t}} \cdot \frac{\Delta}{\Delta-1} 
        &= \frac{\exp\tp{2 \beta_c t} + 1}{\Delta - (\Delta-2) \exp\tp{2\beta_c t } } \cdot \frac{\Delta}{\Delta-1} \\
        \nonumber &= \frac{1}{\Delta-1} + \frac{2 \exp\tp{2 \beta_c t }}{\Delta - (\Delta-2) \exp\tp{2 \beta_c t}}\\
        &= \frac{1}{\Delta-1} + \frac{2}{\Delta-2} \cdot \frac{1}{\exp\tp{2\beta_c (1-t)} -1},
    \end{align}
    where in the last equation, we use the fact that $\exp(2\beta_c) = \frac{\Delta}{\Delta-2}$ from the definition of the critical inverse temperature.
    Thus, the integral $\norm{J}_{2} \int_0^{1} C(t) \dif t $ can be bounded by
    \begin{align*}
    \norm{J}_{2} \int_0^{1} C(t) \dif t & \le 2\beta_c \Delta \int_0^{1} C(t) \dif t \le 2\beta_c \Delta + 2\beta_c \Delta \int_0^{1-1/n} C(t) \dif t\\
    \tp{\text{change of variable and~\eqref{eq:integral-partial}}} \quad &\le 2\beta_c \Delta + \frac{2\beta_c \Delta}{\Delta-1} + \frac{4\beta_c \Delta}{\Delta-2} \cdot \int_{1/n}^{1} \frac{1}{\exp \tp{2 \beta_c t} - 1} \dif t\\
    &\overset{(\star)}{\le} 2\beta_c \Delta + \frac{2\beta_c \Delta}{\Delta-1}\\
    &+ \frac{4 \beta_c \Delta}{\Delta-2} \cdot \frac{\log \tp{\exp(2\beta_c) - 1} - \log \tp{\exp\tp{2\beta_c/n} - 1}}{2\beta_c}\\
    &\overset{(*)}{\le} 2\beta_c \Delta+ \frac{2\beta_c \Delta}{\Delta-2} + \frac{4\Delta}{\Delta-2} + \frac{2 \Delta \log n}{\Delta-2}\\
    (\beta_c \Delta \le 3) \quad &\le 24 + \frac{2 \Delta \log n}{\Delta-2}.
    \end{align*}
    Here, the inequality $(\star)$ follows from $\int \frac{1}{\exp\tp{2\beta t}-1} \dif t = \frac{\log \tp{\exp\tp{2\beta t}-1}}{2\beta} - t + C$ and $(*)$ follows from the inequalities $\log \tp{\exp(2\beta_c)-1} \le \log(2 \e^2 \beta_c)$ where we use $0 \le \beta_c \le 1$, and $- \log \tp{\exp\tp{2\beta_c/n} - 1} \le -\log \tp{2\beta_c/n} = \log n - \log (2\beta_c)$.
    Therefore, the Gibbs distribution $\mu$ satisfies the approximate tensorization of entropy with constant
    \begin{align*}
        \exp\tp{\norm{J}_{2} \int_0^{1} C(t) \dif t} \le \exp\tp{24 +\frac{2\Delta \log n}{\Delta-2}} = O\tp{n^{2 + \frac{4}{\Delta-2}}}.
    \end{align*}
    This indicates that the Glauber dynamics on the Gibbs distribution of the critical Ising model mixes within time $O\tp{n^{3 + \frac{4}{\Delta-2}} \log n}$.
\end{proof}}

\begin{remark}\label{rmk:FD-Ising}
Despite of the success in analyzing Glauber dynamics for the critical hardcore model, 
the current approach based on
field dynamics may not be suitable for the critical Ising model. 
Specifically, the Ising model with the critical inverse temperature $\beta_c(\Delta)$ and a slightly biased external field $\exp\tp{h} = 1-\delta$ exhibits a spectral independence bound of $\Theta\tp{\frac{1}{\delta^2}}$, in contrast to the $O\tp{\frac{1}{\delta}}$-spectral independence as in the hardcore model. This results in a sub-exponential mixing time upper bound by simply applying a similar integration 
but in the context of field dynamics.
\end{remark}


\subsubsection{Mixing at the critical interaction norm}

We now consider the Gibbs distribution $\mu$ of the Ising model specified by an interaction matrix $J \in \mathbb{R}^{n \times n}$, 
where $J$ is symmetric and positive semi-definite.
Let $\*h \in \mathbb{R}^n$ denote the external fields. 
The Gibbs distribution $\mu=\mu_{J,\*h}$ is defined as follows:
\begin{align}\label{eq:Ising-model-interaction-matrix-2}
\forall \*x \in \{-1,+1\}^n, \quad \mu(\*x) \propto \exp\tp{\frac{1}{2} \*x^\intercal J \*x + \*h^\intercal \*x}.
\end{align}

\fixed{We present a stronger upper bound than that stated in \Cref{thm:Ising-interaction}, extending the same $O\tp{n^{3/2} \log n}$ mixing time bound beyond criticality by a margin of  $O\tp{1/\sqrt{n}}$.}

\begin{theorem}\label{thm:upper-Ising-interaction-strong}
For any  symmetric matrix $J \in \mathbb{R}^{n \times n}$ such that $0 \preceq J \preceq \tp{1+\frac{\alpha}{\sqrt{n}}} I$ for some constant $\alpha \ge 0$, and for any $\*h \in \mathbb{R}^n$,
the mixing time of the Glauber dynamics for  $\mu_{J,\*h}$
is $O(n^{3/2} \log n)$.
\end{theorem}

\fixed{In the mean-field Ising model, it has been shown in~\cite{ding2009meanfield} that the Glauber dynamics mixes in time $\Theta(n^{3/2})$ when the inverse temperature $\beta$ lies in a $\Theta(1/\sqrt{n})$-length window centered at its critical threshold. \Cref{thm:upper-Ising-interaction-strong} indicates a similar phenomenon for Ising models with general interaction matrices.}

To prove \cref{thm:upper-Ising-interaction-strong}, it suffices to consider the critical Ising model specified by a rank-one interaction matrix,
and the general case follows from the needle decomposition~\cite[Corollary 30]{anari2021entropic} and \cite{EKZ22} which decomposes the Ising model of interest into a mixture of rank-one Ising models satisfying nice properties.
\fixed{First}, we introduce the following lemma that establishes spectral independence for the critical rank-one Ising model.

\begin{lemma}\label{lem:Ising-interaction-SI}
    Let $J = \*u \*u^\intercal \in \mathbb{R}^{n \times n}$ be a rank-one matrix, and let $\*h \in \mathbb{R}^n$.
    Consider the Gibbs distribution $\mu = \mu_{J,\*h}$ of the Ising model with interaction matrix $J$ and external fields $\*h$.
    \begin{enumerate}
        \item \label{item:lem:Ising-interaction-SI-1}
        \textnormal{(\cite[Proposition 31]{anari2021entropic})} If $\norm{u}_2^2 \le 1-\delta$ for some $\delta > 0$, then $\mu$ is $\frac{1}{\delta}$-spectrally independent.
        \item \label{item:lem:Ising-interaction-SI-2}
        If $\norm{u}_2^2 \le 1+\frac{\alpha}{\sqrt{n}}$ for some constant $\alpha \ge 0$, then  $\mu$ is $O_\alpha(\sqrt{n})$-spectrally independent.
    \end{enumerate}
\end{lemma}

The sub-critical case (\Cref{item:lem:Ising-interaction-SI-1}) has already been established in \cite[Proposition 31]{anari2021entropic}.
We provide an alternative proof for this case, which is also generalizable to the critical point (\Cref{item:lem:Ising-interaction-SI-2}). 
This proof of \Cref{lem:Ising-interaction-SI} is detailed in~\Cref{sec:SI-criticality}.

We are now ready to prove~\Cref{thm:upper-Ising-interaction-strong}. The proof follows a structure analogous to that of \fixed{\Cref{thm:Ising-graphical}}, so we will omit certain calculations that are similar to those previously discussed.
\begin{proof}[Proof of~\Cref{thm:upper-Ising-interaction-strong}]
    Let $n$ be sufficiently large, specifically $n \ge 10 \alpha^2$. Let $\mu$ be the Gibbs distribution of the Ising model specified by the rank-one interaction matrix $J=\*u \*u^\intercal \in \mathbb{R}^{n \times n}$ with $\norm{u}_2^2 \le 1+\frac{\alpha}{\sqrt{n}}$ and external fields $\*h \in \mathbb{R}^n$. By~\Cref{thm:Ising-integral} and \Cref{lem:Ising-interaction-SI}, and applying a change of variables, the distribution $\mu$ satisfies the approximate tensorization of entropy with constant $\exp\tp{\int_0^{1+\alpha/\sqrt{n}} C(t) \dif t}$, where $C:[0,1+\alpha/\sqrt{n}] \to \mathbb{R}_{\ge 0}$ is given by
    \begin{align*}
        C(t) = \begin{cases}
            \min\tp{\frac{1}{1-t}, K \sqrt{n}} & t \in [0,1],\\
            K \sqrt{n} & t \in [1,1+3\alpha/\sqrt{n}].
        \end{cases}
    \end{align*}
    Here, $K \ge 0$ is a sufficiently large constant. By straightforward calculations, we have
    \begin{align*}
        \exp\tp{\int_0^{1+\alpha/\sqrt{n}} C(t) \dif t} \le \exp\tp{(\alpha+1)K + \frac{1}{2}\log n} = O\tp{n^{1/2}}.
    \end{align*}
    This implies that the Glauber dynamics of any rank-one Ising model with $\norm{u}^2_2 \le 1+\frac{\alpha}{\sqrt{n}}$ has an $\Omega(n^{-1/2})$ modified log-Sobolev constant. Together with the needle decomposition~\cite[Corollary 30]{anari2021entropic}\footnote{We remark that $\norm{u}_2^2 \le \norm{J}_2$ should be the correct criteria in the statement of the needle decomposition.}, the Glauber dynamics of any Ising model with $\norm{J}_{2} \le 1 + \frac{\alpha}{\sqrt{n}}$ has an $\Omega(n^{-1/2})$ modified log-Sobolev constant, thereby implying an $O\tp{n^{3/2} \log n}$ mixing time. 
\end{proof}

\section{\texorpdfstring{$O(\sqrt{n})$-Spectral Independence at Criticality}{O(sqrt(n))-Spectral Independence at Criticality}}\label{sec:SI-criticality}
\fixed{In this section, we consider the Ising model where criticality is characterized by the norm of the rank-one interaction matrix $J = \*u \*u^\intercal \in \mathbb{R}^{n \times n}$. 
In this context, we prove \Cref{lem:Ising-interaction-SI}.}

\begin{proof}[Proof of~\Cref{lem:Ising-interaction-SI}]
    Without loss of generality, we assume $u_i \neq 0$ for all $1 \le i \le n$.
    Let $D = \-{diag}\set{u_i}_{i=1}^n$ be the diagonal matrix with $D_{ii} = u_i$.
    Since $u_i \neq 0$ for all $1 \leq i \leq n$, it holds that $\lambda_{\max}(\Psi_\mu) = \lambda_{\max}(D^{-1}\Psi_\mu D)$; and we will bound $\lambda_{\max}(\Psi_\mu) = \lambda_{\max}(D^{-1}\Psi_\mu D)$ by
    \begin{align} \label{eq:row-sum-inf-mat}
      \norm{D^{-1}\Psi_\mu D}_\infty &= \max_{1 \le i \le n} \abs{u_i}^{-1} \sum_{1 \le j \le n} \abs{u_j} \abs{\Psi_{\mu}(i,j)}.
    \end{align}
    Let $\+T^{\mathrm{SAW}}_i$ be the self-avoiding walk tree (recall its definition at~\Cref{sec:def-SAW-tree}) rooted at element $1 \le i \le n$, and let $\mu_{\+T_i}$ be the corresponding Gibbs distribution on $\+T^{\mathrm{SAW}}_i$.
    By~\eqref{eq:row-sum-inf-mat} and the preservation of total influence in the self-avoiding walk tree \cite[Lemma 8]{chen2023rapid}, we have
    \begin{align}\label{eq:Ising-interaction-SI-bound}
        \lambda_{\max}(\Psi_{\mu}) \le \max_{1 \le i \le n} \abs{u_i}^{-1} \sum_{1 \le j \le n} \abs{u_j} \abs{\Psi_{\mu}(i,j)} \le \max_{1 \le i \le n} \abs{u_i}^{-1} \sum_{v \in V\tp{\+T^{\mathrm{SAW}}_i}} \abs{u_v} \cdot \abs{\Psi_{\mu_{\+T_i}}(i,v)}.
    \end{align}
    Fixing the starting element $1 \le i \le n$, the vertices $v$ in the self-avoiding walk tree $\+T_i^{\mathrm{SAW}}$ correspond one-to-one to tuples $(e_0,e_1,\ldots,e_\ell)$ with $e_0,e_1,\ldots,e_\ell$ being pairwise distinct and $e_0 = i$. Furthermore, by the conditional independence of influence \cite[Lemma 15]{chen2023rapid}, we have
    \begin{align*}
    \abs{\Psi_{T_i^{\mathrm{SAW}}}(i,v)} = \prod_{i=1}^\ell \abs{\Psi_{T_i^{\mathrm{SAW}}}(e_{i-1},e_i)} \le \prod_{i=1}^\ell \abs{\tanh (u_{e_{i-1}} u_{e_i})} \le \prod_{i=1}^{\ell} \abs{u_{e_{i-1}}} \abs{u_{e_i}}.
    \end{align*}
    Together with~\eqref{eq:Ising-interaction-SI-bound}, we can bound the maximum eigenvalue of the influence matrix as follows:
    \begin{align}\label{eq:pre-interaction}
        \lambda_{\max}(\Psi_{\mu}) 
 \le \sum_{\substack{(e_0,e_1,\ldots,e_\ell) \in V^\ell\\ e_0,e_1,\ldots,e_\ell \text{ are pairwise distinct, }e_0 = i}} \prod_{i=1}^\ell \abs{u_{e_i}}^2. 
    \end{align}
    Let $x_0,x_1,\ldots,x_n$ be i.i.d.~random variables on $[n]$ such that $x_i$ takes value $k$ with probability proportional to $u_{k}^2$, and $T$ is the first index with $e_T = e_j$ for some $0 \le j < T$. 
    Then, by \eqref{eq:pre-interaction},
    \begin{align} \label{eq:interaction}
      \lambda_{\max}(\Psi_{\mu})
      \leq \sum_{\ell = 0}^{+\infty} \Pr[]{T \ge \ell} \cdot \norm{u}_2^{2\ell}.
    \end{align}
    Now, we prove  \Cref{item:lem:Ising-interaction-SI-1} of \Cref{lem:Ising-interaction-SI}. Assume $\norm{u}_2^2 \le 1-\delta$. From~\eqref{eq:interaction}, we have
    \begin{align*}
    \lambda_{\max}(\Psi_{\mu}) \le \sum_{\ell=0}^{+\infty} \norm{u}_2^{2\ell} \le \frac{1}{\delta}.
    \end{align*}
    Next, we prove \Cref{item:lem:Ising-interaction-SI-2} of \Cref{lem:Ising-interaction-SI}. Assume $\norm{u}_2^2 \le 1+\frac{\alpha}{\sqrt{n}}$ for some constant $\alpha \ge 0$.
    By~\Cref{lem:birthday-paradox}, the right-hand side of~\eqref{eq:interaction} satisfies:
    \begin{align}\label{eq:integration-approx}
    \nonumber\sum_{\ell=0}^{+\infty} \Pr[]{T \ge \ell}\cdot \norm{u}_2^{2\ell} &\le 1 + \sum_{\ell=1}^{+\infty} \exp\tp{-\frac{(\ell-1)^2}{2n}} \cdot \tp{1+\frac{\alpha}{\sqrt{n}}}^\ell\\
    &\le 1+2\sqrt{n} \cdot \underbrace{\sum_{k\cdot\sqrt{n} \in \mathbb{N}} \frac{1}{\sqrt{n}} \cdot \exp\tp{-\frac{k^2}{2}} \cdot \exp\tp{\alpha k}}_{I_n}.
    \end{align}
    Note that $I_n$ converges to $\int_{0}^{+\infty} \exp\tp{\alpha t} \exp\tp{-\frac{t^2}{2}}\dif t < +\infty$ as $n\to\infty$. Thus, from~\eqref{eq:interaction} and~\eqref{eq:integration-approx}, the maximum eigenvalue of the influence matrix $\Psi_{\mu}$ is $O(\sqrt{n})$. This concludes the proof.
\end{proof}


\section{Lower Bounds}
\label{sec:LB}

In this section, we prove the lower bounds on the mixing time of Glauber dynamics for both the critical hardcore and Ising models, as stated in \Cref{thm:hardcore} and \Cref{thm:Ising-graphical}.

These lower bounds on mixing time are derived from the lower bounds on spectral independence. 
First, we present such a lower bound for the critical hardcore model.

\begin{theorem}\label{thm:lower-bound-hardcore}
  Let $\Delta \ge 3$ be a constant.
  For all  $n\ge 1$, there exists a bipartite graph $G=(L,R,E)$ with $\abs{L}=\abs{R}=2n$ and maximum degree at most $\Delta$ 
  such that the Gibbs distribution $\mu$ of the hardcore model on $G$ with critical fugacity $\lambda=\lambda_c(\Delta)= \frac{(\Delta-1)^{\Delta-1}}{(\Delta-2)^\Delta}$ satisfies 
  \begin{align*}
    \lambda_{\max} \tp{\Psi_\mu} = \Omega(n^{1/3}),
  \end{align*}
  where $\Psi_\mu$ is the influence matrix of  $\mu$.
\end{theorem}

According to the the universality of spectral independence~\cite[Theorem 3.1]{anari2024universality},
this lower bound of $\Omega(n^{1/3})$ on spectral independence implies a corresponding lower bound of $\Omega(n^{4/3})$ on the relaxation time of Glauber dynamics on the same instances. 
Consequently, this proves the lower bound of $\Omega(n^{4/3})$  on the mixing time of Glauber dynamics in \Cref{thm:hardcore}.

Next, we present a lower bound for the critical anti-ferromagnetic Ising model. 

\begin{theorem}\label{thm:lower-bound-Ising}
  Let $\Delta \ge 3$ be a constant.
For all $n \ge 1$, there exists a bipartite graph $G=(L,R,E)$ with $\abs{L} = \abs{R} = n$ and maximum degree at most $\Delta$ 
such that the Gibbs distribution $\mu$ of the anti-ferromagnetic Ising model on graph $G$ with zero external field and  critical inverse temperature $\beta=-\beta_c(\Delta)=-\frac{1}{2} \log \frac{\Delta}{\Delta-2}$  satisfies
  \begin{align}
    \lambda_{\max} \tp{\Psi_\mu} \ge \frac{\*s^\intercal \Psi_\mu \*s}{\*s^\intercal \*s} = \frac{\*s^\intercal \Psi_{\mu} \*s}{2n} = \Omega\tp{\sqrt{n}},\label{eq:thm:lower-bound-Ising}
  \end{align}
  where $\Psi_\mu$ is the influence matrix of  $\mu$, and $\*s \in \mathbb{R}^{L \cup R}$ is defined as $\*s_{u} = \begin{cases}
    1 & u \in L;\\
    -1 & u \in R.
  \end{cases}$
\end{theorem}

For each bipartite graph $G=(L,R,E)$ constructed in \Cref{thm:lower-bound-Ising}, flipping the spins on one side (say $R$) transforms the model into a ferromagnetic Ising model with zero external field. 
Formally, let $\mu$ be the  Gibbs distribution of the anti-ferromagnetic Ising model on $G$ defined in~\Cref{thm:lower-bound-Ising}, and let $\widetilde{\mu}$ be the Gibbs distribution of the ferromagnetic Ising model on the same graph $G$ with opposite inverse temperature $\beta = \beta_c(\Delta)$ and zero external field. 
It is straightforward to verify that 
    \begin{align}
        \forall u,v \in L \cup R, \quad \Psi_{\widetilde{\mu}}(u,v) = \begin{cases}
            \Psi_{\mu}(u,v) & \text{$u,v$ are on the same side};\\
            -\Psi_{\mu}(u,v) & \text{otherwise.}           
        \end{cases}
    \end{align}
Therefore, we have $\*1^\intercal \Psi_{\widetilde{\mu}} \*1=\*s^\intercal \Psi_{\mu} \*s$, which leads to the following lower bound for the critical ferromagnetic Ising model.

\begin{corollary}\label{cor:lower-bound-Ising}
Let $\Delta \ge 3$ be a constant.
For all $n \ge 1$, there exists a bipartite graph $G=(L,R,E)$ with $\abs{L} = \abs{R} = n$ and maximum degree at most $\Delta$ 
such that the Gibbs distribution $\mu$ of the ferromagnetic Ising model on graph $G$ with zero external field and critical inverse temperature $\beta=\beta_c(\Delta)=\frac{1}{2} \log \frac{\Delta}{\Delta-2}$ satisfies
  \begin{align*}
    \lambda_{\max} \tp{\Psi_\mu} \ge \frac{\*1^\intercal \Psi_{\mu} \*1}{\*1^\intercal \*1} = \frac{\*1^\intercal \Psi_\mu \*1}{2n} = \Omega\tp{\sqrt{n}},
  \end{align*}
  where $\Psi_\mu$ is the influence matrix of $\mu$, and $\*1 \in \mathbb{R}_{L \cup R}$ denotes the all-one vector.
\end{corollary}

Again, by the universality of spectral independence,
the lower bounds of $\Omega(\sqrt{n})$ on spectral independence, as established in \Cref{thm:lower-bound-Ising} and \Cref{cor:lower-bound-Ising}, imply corresponding lower bounds of $\Omega(n^{3/2})$ on the relaxation time (and consequently, the mixing time) of Glauber dynamics on the same instances.
This proves the lower bound in \Cref{thm:Ising-graphical}.

    This section is organized as follows. In~\Cref{sec:lower-proof-overview}, we outline the proof of mixing time lower bounds by introducing the anti-concentration phenomenon in the critical models. In the remaining sections (\Cref{sec:lower-hardcore} and \Cref{sec:lower-Ising}), we prove~\Cref{thm:lower-bound-hardcore} and~\Cref{thm:lower-bound-Ising} respectively.

\subsection{Proof overview for the mixing lower bounds}\label{sec:lower-proof-overview}
The hard instances used to prove the lower bounds in \Cref{thm:lower-bound-hardcore,thm:lower-bound-Ising} are well studied in the literature: 
random regular bipartite graphs.
The torpid mixing of local Markov chains on these instances in the super-critical (non-uniqueness) regime is well known~\cite{mossel2009hardness},
and such random graphs have been widely utilized as gadgets in complexity reductions to establish computational hardness in this regime~\cite{sly2010computational,sly2012computational,galanis2015inapproximability,GSV16}.

The slow mixing and computational hardness arise from the fact that,
in the super-critical regime, the probability mass of the magnetization (represented by the difference in average sign of spins between the two sides of the bipartition) becomes bimodal.
This bimodal property  creates bottlenecks in the state space and leads to long-range correlations in the Gibbs measure, 
resulting in slow mixing of local Markov chains and supporting gadget reductions for computational hardness.

This situation changes significantly at criticality: 
the bimodality of the Gibbs measure no longer holds due to the uniqueness of Gibbs measure,
yet the Gibbs measure lacks the strong concentration observed in the sub-critical (uniqueness) regime.
Instead, at criticality, the Gibbs measure exhibits an \emph{anti-concentration} phenomenon,
which prevents Glauber dynamics from  mixing ideally within $\widetilde{O}(n)$  steps. 
Specifically, the anti-concentration is characterized by the property
that the magnetization of a sample from the Gibbs distribution is $\Omega(n^{\alpha})$ for some $\alpha > -1/2$ 
with large probability, 
rather than concentrating around $\Theta(n^{-1/2})$ as in a product distribution. 
This critical behavior, akin to that observed in the mean-field Ising model~\cite{LLP10,ding2009meanfield}, is illustrated in \Cref{fig:Ising-large-deviation}.

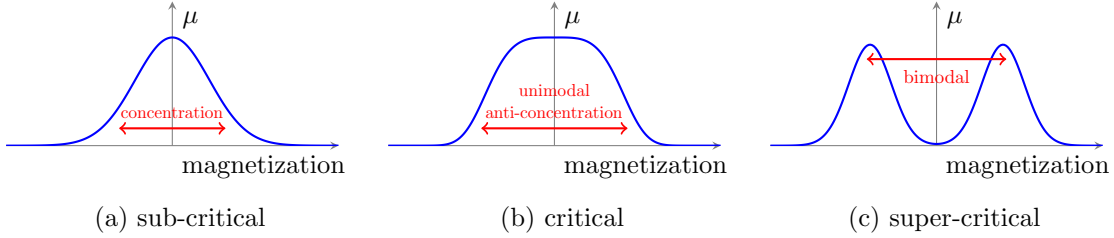
\begin{figure}
    \centering
    \begin{minipage}{0.3\textwidth}
    \centering
    \begin{tikzpicture}
        \begin{axis}[
            axis lines=middle,
            axis line style={gray},
            xlabel=magnetization,
            ylabel=$\mu$,
            xlabel style={at={(ticklabel* cs:1)}, anchor=north,font=\small,xshift=-1cm},
            ylabel style={font=\small},
            samples=200,
            domain=-5:5,
            ymin=0,
            ymax=1.0,
            xtick=\empty,
            ytick=\empty,
            height=3.5cm,
            width=6cm
        ]
        
        \addplot[blue, thick] {0.75*exp(-0.4*(x)^2)};
        \addplot[<->, red, thick] coordinates{(-1.6,0.12) (1.6,0.12)} node[midway, above] {\scalebox{0.6}{\centering concentration}};
        \end{axis}
    \end{tikzpicture}
    \subcaption{sub-critical}
    \end{minipage}
    \begin{minipage}{0.3\textwidth}
    \centering
    \begin{tikzpicture}
        \begin{axis}[
            axis lines=middle,
            axis line style={gray},
            xlabel=magnetization,
            ylabel=$\mu$,
            xlabel style={at={(ticklabel* cs:1)}, anchor=north,font=\small,xshift=-1cm},
            ylabel style={font=\small},
            samples=200,
            domain=-5:5,
            ymin=0,
            ymax=1.0,
            xtick=\empty,
            ytick=\empty,
            height=3.5cm,
            width=6cm
        ]
        
        \addplot[blue, thick] {0.75*exp(-0.04*(x)^(4))};
        \addplot[<->, red, thick] coordinates{(-2.2,0.12) (2.2,0.12)} node[midway, above] {\scalebox{0.6}{\parbox{4cm}{\centering unimodal\\anti-concentration}}};
        \end{axis}
    \end{tikzpicture}
    \subcaption{critical}
    \end{minipage}
    \begin{minipage}{0.3\textwidth}
    \centering
    \begin{tikzpicture}
        \begin{axis}[
            axis lines=middle,
            axis line style={gray},
            xlabel=magnetization,
            ylabel=$\mu$,
            xlabel style={at={(ticklabel* cs:1)}, anchor=north,font=\small,            xshift=-1cm},
            ylabel style={font=\small},
            samples=200,
            domain=-5:5,
            ymin=0,
            ymax=1.0,
            xtick=\empty,
            ytick=\empty,
            height=3.5cm,
            width=6cm,
        ]
        \addplot[blue, thick] {0.7*exp(-1.2*(x-2)^2) + 0.7*exp(-1.2*(x+2)^2)};
        \addplot[<->, red, thick] coordinates{(-2.1,0.6) (2.1,0.6)} node[midway, below] {\scalebox{0.65}{bimodal}};
        \end{axis}
    \end{tikzpicture}
    \subcaption{super-critical}
    \end{minipage}
    \caption{Density of magnetization in hardcore/Ising models}\label{fig:Ising-large-deviation}
\end{figure}

To leverage the anti-concentration in the Gibbs measure $\mu$ to prove a lower bound on the mixing time, 
we examine the quadratic form $\*s^\intercal\Psi_\mu\*s$ of the influence matrix $\Psi_\mu$ using a signed vector $\*s$ with opposite signs on the two sides of bipartition. 
This quadratic form, up to a constant factor, is approximately  $n^2$ times the variance of the magnetization.
By the anti-concentration property, this variance is lower bounded as $\Omega\tp{n^{2\alpha}}$ for some $\alpha > -1/2$.
From this, we can derive a lower bound on spectral independence, as the largest eigenvalue of the influence matrix is bounded by
$\lambda_{\max}(\Psi_\mu)\ge \frac{\*s^\intercal\Psi_\mu\*s}{\*s^\intercal\*s}=\Omega\tp{n^{1+2\alpha}}$.
By the universality of spectral independence, as established in \cite{anari2024universality},
this translates into a lower bound of $\Omega\tp{n^{2+2\alpha}}$ on the mixing time of the Glauber dynamics.

The spectral independence approach for deriving a mixing time lower bound offers a notable advantage over the traditional conductance-based method, such as that in~\cite{mossel2009hardness}, by generalizing the notion of cuts in the state space.
From a functional analysis perspective, the spectral independence method allows for the use of more general testing functions, beyond the Boolean functions typically employed in conductance-based bounds.
This flexibility enables us to identify a lower bound on the mixing time that is polynomially bounded.
Consider the Gibbs distribution $\mu$,  supported on $\Omega \subseteq \{-1,+1\}^{L \cup R}$,
induced by the hardcore model on a random $\Delta$-regular bipartite graph $G=(L,R,E)$ with critical fugacity $\lambda=\lambda_c(\Delta)$. 
The approximate tensorization of variance with constant $C$, which is equivalent to the Poincar\'{e} inequality for Glauber dynamics with constant $\frac{1}{Cn}$, is given by:
\begin{align}\label{eq:AT-variance}
    \Var[\mu]{f} \le C \sum_{v \in V} \E[]{\Var[v]{f}},
\end{align}
where $f \in \mathbb{R}^{\Omega}$ is a testing function. 
In contrast to conductance-based lower bounds, which are derived using the Boolean testing function $\*X \mapsto \*1[\sum_{u \in L} X_u > \sum_{v \in R} X_v]$ to identify a cut in the state space $\Omega$, 
the spectral independence approach uses a linear testing function $\*X \mapsto \*s^\intercal \*X$, where $\*s\in\mathbb{R}^{L\cup R}$ is a signed vector defined by letting $s_u = 1$ for $u \in L$; and $s_v = -1$ for $v \in R$.
This generalization of cuts allows for a more refined analysis of the Poincar\'{e} constant and nicely describes the anti-concentration phenomenon discussed above.
In particular, the variance of this testing function $f$ with $f(\*X) = \*s^\intercal \*X$, is given by:
    \begin{align*}
        \Var[\mu]{f} = \Var[\mu]{\*s^\intercal \*X} = \*s^\intercal \mathrm{Cov}(\mu) \*s.
    \end{align*}
Meanwhile, the sum of local variances of $f$ is bounded as:
    \begin{align*}
        \sum_{v \in V} \E[]{\Var[v]{f}} = \sum_{v \in V} \E[]{\Var[v]{\*s^\intercal \*X}} = \sum_{v \in V} \E[]{\Var[v]{X_v}} \cdot s_v^2 \le \*s^\intercal \mathrm{diag}\set{\Var[\mu]{X_v}}_{v \in L \cup R} \*s.
    \end{align*}
By leveraging the anti-concentration in the critical Gibbs measure, we obtain:
    \begin{align*}
        \*s^\intercal \mathrm{Cov}(\mu) \*s = \Omega\tp{n^{1+2\alpha}} \*s^\intercal \mathrm{diag}\set{\Var[\mu]{X_v}}_{v \in L \cup R} \*s.
    \end{align*}
Thus, combining these observations, the Poincar\'{e} constant $C$ in the approximate tensorization of variance~\eqref{eq:AT-variance} is lower bounded by $\Omega\tp{n^{1+2\alpha}}$ for some $\alpha>-\frac{1}{2}$. 
This rules out the possibility of achieving a near-optimal  $\tilde{O}(n)$  relaxation time for the Glauber dynamics.

Finally, to establish the anti-concentration in the critical Gibbs measure
which is essential for deriving the spectral independence lower bound, 
we face a significant technical challenge: 
the absence of bimodality in the critical regime. 
This lack of bimodality prevents us from leveraging the optimization-based techniques for approximating the log-partition function, 
which have been primary tools for analyzing the density of magnetization, especially in the context of torpid mixing and computational hardness in the super-critical regime.
To overcome this obstacle, we shift our focus from the log-partition function to the partition function itself. 
By employing a local limit theorem for large deviations, we can analyze the partition function's behavior in the critical regime. 
This approach not only helps us establish the necessary anti-concentration property but also offers a novel perspective on the analysis of critical Gibbs measures, which may be of independent interest.

\subsection{Critical hardcore model on  random symmetric bipartite graphs}\label{sec:lower-hardcore}
In this section, we will prove~\Cref{thm:lower-bound-hardcore}. We first show the spectral independence lower bound via anti-concentration in~\Cref{sec:hardcore-spectral-anticoncentrate}, and then establish the anti-concentration in~\Cref{sec:anti-concentration-hardcore}.

The hard instances we consider are the hardcore model on random symmetric bipartite graphs with a maximum degree $\Delta$ and critical fugacity $\lambda=\lambda_c(\Delta) = \frac{(\Delta-1)^{\Delta-1}}{(\Delta-2)^\Delta}$.
The random symmetric bipartite graph $G=(L,R,E)$ is constructed as follows:

\begin{itemize}
\item Fix a constant $\Delta \ge 3$. Let $L=\{\ell_1,\ell_2,\ldots,\ell_{2n}\}$ and $R=\{r_1,r_2,\ldots,r_{2n}\}$ be the vertex sets.

\item Sample $\Delta$ perfect matchings $M_1,M_2,\ldots,M_\Delta$ in the complete graph $K_{2n}$ on  vertices $[2n]$ uniformly and independently at random. 
For each $(u,v) \in M_k$ for some $1 \le k \le \Delta$, add edges $(\ell_u,r_v)$ and $(\ell_v,r_u)$ to  the bipartite graph $G=(L,R,E)$.
\end{itemize}
We use $\+G(n,\Delta)$ to denote the law of the random symmetric bipartite graph $G=(L,R,E)$. 
This construction of symmetric bipartite graphs has been used as a gadget for establishing computational hardness for super-critical anti-ferromagentic two-spin systems (such as in~\cite{sly2012computational}).

\subsubsection{Spectral correlation via anti-concentration}\label{sec:hardcore-spectral-anticoncentrate}
The lower bound in \Cref{thm:lower-bound-hardcore} establishes a spectral correlation property for the critical hardcore Gibbs measure.
This lower bound can be  implied by the anti-concentration of the Gibbs measure on the random instances $G=(L,R,E)\sim\+G(n,\Delta)$.

Recall $L=\{\ell_1,\ell_2,\ldots,\ell_{2n}\}$ and $R=\{r_1,r_2,\ldots,r_{2n}\}$.
For $S,T \subseteq [2n]$, let $\*1_{S,T}\subseteq L\cup R$ be the vertex set define as $\*1_{S,T}:=\{\ell_i,r_j\mid i\in S,j\in T\}$.
For integers $0 \le A,B,C \le 2n$, 
we define:
\begin{align*}
    \alpha_{A,B,C} := \sum_{\substack{S, T  \subseteq [2n]\text{ with}\\ |S \setminus T| = A, |T \setminus S| = B, |S \cap T| = C}} \E[G \sim \+G(n,\Delta)]{w_G(\*1_{S,T})},
\end{align*}
where $w_G(I)$ represents the weight of the vertex set $I$ in the hardcore model on $G$. 
Specifically, $w_G(I)=\lambda^{|I|}$ with $\lambda=\lambda_c(\Delta)$ if $I$ is an independent set in $G$, and $w_G(I)=0$ otherwise.
For a given bipartite graph $G = (L\cup R, E)$, let $\mu_G$ be a distribution over $2^{L\cup R}$ be the Gibbs distribution of the hardcore model defined as
\begin{align*}
  \forall \text{ independent set } I, \quad \mu_G(I) \propto w_G(I) = \lambda^{\abs{I}}.
\end{align*}
We may also consider $\mu_G$ as a distribution over $\set{0,1}^{L\cup R}$ (i.e., use $1$ to indicate a vertex is in the independent set; and use $0$ to indicate that it is not) for technical convenience.

The following lemma establishes an anti-concentration property, which provides a key to proving the spectral correlation described in \Cref{thm:lower-bound-hardcore}.
\begin{lemma}\label{lem:distant-vs-all-hardcore} 
 Fix any constant $\Delta \ge 3$.
    There exist constants $\epsilon,\eta \in (0,1)$ such that
    \begin{align}\label{eq:large-deviation-hardcore}
      \sum_{\substack{|A-B|>\eta\cdot  n^{2/3}}} \alpha_{A,B,C} > \epsilon\cdot \sum_{0 \le A,B,C \le 2n} \alpha_{A,B,C}
    \end{align}
    holds for all sufficiently large $n>0$ and all integers $0 \le A,B,C \le 2n$. 
\end{lemma}

\begin{remark}\label{remark:hardcore-anticoncentrate}
\Cref{lem:distant-vs-all-hardcore}  captures an anti-concentration of the Gibbs measure on the random instance $\+G(n,\Delta)$.
Suppose a random vertex set $I \subseteq L \cup R$ is drawn with a probability proportional to the expected weight $\E[G \sim \+G(n,\Delta)]{w_G(I)}$, representing a ``typical'' sample of  configurations in a random instance $G \sim \+G(n,\Delta)$.
Then, the inequality \eqref{eq:large-deviation-hardcore} guarantees that, with a non-negligible probability, the difference between the numbers of occupied vertices in $L$ and $R$ will be substantial, specifically bounded as $\Omega(n^{2/3})$. This phenomenon is as illustrated in \Cref{fig:Ising-large-deviation}.
\end{remark}

\Cref{thm:lower-bound-hardcore} follows from \Cref{lem:distant-vs-all-hardcore}. To see this, we first need to prove the following lemma.

\begin{lemma}\label{lem:hardcore-SI-partition-function}
For any bipartite graph $G=(L,R,E)$ with positive measure in $\+G(n,\Delta)$, it holds 
\[
\lambda_{\max}(\Psi_{\mu_G})  
\ge \frac{4\*s^\intercal \mathrm{Cov}(\mu_G) \*s}{\*s^\intercal \*s} =
\frac{2}{n} \cdot \sum_{1 \le i,j \le 2n} \tp{\Pr[I \sim \mu_G]{\ell_i, \ell_j \in I} - \Pr[I \sim \mu_G]{\ell_i, r_j \in I}}.
\]
where $\mathrm{Cov}(\mu_G)$\footnote{Here, we consider $\mu_G$ as a distribution over $\set{0,1}^{L\cup R}$, so that its covariance matrix can be defined. We refer to~\Cref{rem:SI-domain} for a detailed discussion.} denotes the covariance matrix for $\mu_G$ and the vector $\*s\in \mathbb{R}^{L\cup R}$ is defined as:
$$\*s_u = 
\begin{cases}
    1 & u \in L,\\
    -1 & u \in R.
\end{cases}$$
\end{lemma}
\begin{proof}
Consider the hardcore Gibbs distribution $\mu_G$ on graph $G=(L,R,E)$.
The influence matrix $\Psi_{\mu_G}$ is related to the covariance matrix $\mathrm{Cov}(\mu_G)$ through the equation \eqref{eq:SI-ratio}:
\[
\Psi_{\mu_G} = \Pi^{-1} \mathrm{Cov}(\mu_G),
\]
where $\Pi = \mathrm{diag} \set{\Pr[]{u}(1-\Pr[]{u})}_{u \in L \cup R}$ is the diagonal matrix of variances.
Here, $\Pr[]{u}$ represents the marginal probability that vertex $u$ is occupied in an independent set sampled from $\mu_G$.
Since $\Pi \preceq \frac{1}{4} I$, it follows that 
\begin{align}\label{eq:psi-cov}
    \lambda_{\max}(\Psi_{\mu_G}) \ge 4 \lambda_{\max}(\mathrm{Cov}(\mu_G)).
\end{align}

Let $(X,Y)\in\{0,1\}^{L\cup R}$ be the indicator vector for the vertex set $I\sim\mu_G$, 
where $X$ corresponds to the vertices in $L$ and $Y$ corresponds to those in $R$.
By the Courant-Fischer theorem, we have:
\begin{align}\label{eq:psi-max-two-spin}
    \lambda_{\max}(\mathrm{Cov}(\mu_G)) \ge \frac{\*s^\intercal \mathrm{Cov}(\mu_G) \*s}{\*s^\intercal \*s} = \frac{1}{4n} \Var[(X,Y) \sim \mu_G]{\sum_{i=1}^{2n} X_i - \sum_{i=1}^{2n} Y_i}.
\end{align}

By the symmetry of the bipartite graph $G$ with positive measure in $\+G(n,\Delta)$, the marginal probability that vertex $\ell_i\in L$ are occupied is equal to that of the corresponding vertex $r_i\in R$, for each $1 \leq i \leq 2n$. 
Therefore, the variance in~\eqref{eq:psi-max-two-spin} can be expressed as:
\begin{align}
    \Var[(X,Y) \sim \mu_G]{\sum_{i=1}^{2n} X_i - \sum_{i=1}^{2n} Y_i} &=  2\sum_{1 \le i,j \le 2n} \tp{\E[(X,Y) \sim \mu_G]{X_i X_j} -  \E[(X,Y) \sim \mu_G]{X_i Y_j}}\notag\\
    &= 
    2\sum_{1 \le i,j \le 2n} \tp{\Pr[I \sim \mu_G]{\ell_i, \ell_j \in I} - \Pr[I \sim \mu_G]{\ell_i, r_j \in I}}
    \label{eq:var-two-spin}
\end{align}
The lemma follows by combining \eqref{eq:psi-cov}, \eqref{eq:psi-max-two-spin}, and \eqref{eq:var-two-spin}.
\end{proof}

\begin{proof}[Proof of~\Cref{thm:lower-bound-hardcore} (via~\Cref{lem:distant-vs-all-hardcore})]
Let $Z_G=\sum_{I\subseteq L\cup R}w_G(I)$ denote the partition function, and for $u,v\in L\cup R$, define
\[
Z_G^{u \leftarrow +, v \leftarrow +}:=\sum_{\substack{I\subseteq L\cup R\\ I\supseteq \{u,v\}}}w_G(I),
\]
which represents the portion of the partition function $Z_G$ contributed by the independent sets $I$ that contain both $u$ and $v$.

We observe that the anti-concentration inequality~\eqref{eq:large-deviation-hardcore} stated in \Cref{lem:distant-vs-all-hardcore} translates to the following bound:
  \begin{align}\label{eq:target-hardcore-lb}
    \frac{\E[G \sim \+G(n,\Delta)]{\sum_{1 \le i, j \le 2n}\left(Z_G^{\ell_i \leftarrow +,\ell_j \leftarrow +}-Z_G^{\ell_i \leftarrow +,r_j \leftarrow +}\right)} }{\E[G \sim \+G(n,\Delta)]{Z_G}} = \Omega(n^{4/3}).
  \end{align}
%
    To prove this, we first express the expected value $\sum_{1 \le i,j \le 2n} \E[G \sim \+G(n,\Delta)]{Z_G^{\ell_i \leftarrow +, \ell_j \leftarrow +}}$ as follows:
    \begin{align*}
      \sum_{1 \le i,j \le 2n} \E[G \sim \+G(n,\Delta)]{Z_G^{\ell_i \leftarrow +, \ell_j \leftarrow +}}
      &= \sum_{1 \le i,j \le 2n} \sum_{\substack{S \subseteq [2n], T \subseteq [2n]\\ i,j \in S}} \E[G \sim \+G(n,\Delta)]{w_G(\*1_{S,T})}\\ 
      &= \sum_{S \subseteq [2n], T \subseteq [2n]} |S|^2 \E[G \sim \+G(n,\Delta)]{w_G(\*1_{S,T})}\\
    (\text{by symmetry of }\+G(n,\Delta))\qquad  &= \sum_{S \subseteq [2n], T \subseteq [2n]} \frac{|S|^2+|T|^2}{2} \E[G \sim \+G(n,\Delta)]{w_G(\*1_{S,T})}.
    \end{align*}
    where recall $\*1_{S,T}:=\{\ell_i,r_j\mid i\in S,j\in T\}$ for $S,T \subseteq [2n]$. 
    Similarly, we have
    \begin{align*}
      \sum_{1 \le i,j \le 2n} \E[G \sim \+G(n,\Delta)]{Z_G^{\ell_i\leftarrow +,r_j\leftarrow +}} = \sum_{S \subseteq [2n], T \subseteq [2n]} |S||T| \E[G \sim \+G(n,\Delta)]{w_G(\*1_{S,T})}.
    \end{align*}
    Thus, the numerator of ~\eqref{eq:target-hardcore-lb} can be expressed as
     \begin{align*}
     &\sum_{1 \le i, j \le 2n}\E[G \sim \+G(n,\Delta)]{Z_G^{\ell_i \leftarrow +,\ell_j \leftarrow +}} - \sum_{1 \le i,j \le 2n}\E[G \sim \+G(n,\Delta)]{Z_G^{\ell_i \leftarrow +,r_j \leftarrow +}} \\
     = & \sum_{S \subseteq [2n], T \subseteq [2n]} \frac{|S|^2+|T|^2}{2} \E[G \sim \+G(n,\Delta)]{w_G(\*1_{S,T})}-\sum_{S \subseteq [2n], T \subseteq [2n]} |S||T| \E[G \sim \+G(n,\Delta)]{w_G(\*1_{S,T})}\\
    =  & \frac{1}{2} \sum_{S \subseteq [2n], T \subseteq [2n]} (|S|-|T|)^2 \E[G \sim \+G(n,\Delta)]{w_G(\*1_{S,T})} \\
      = &\frac{1}{2} \sum_{0\le A,B,C \le 2n} (A-B)^2 \alpha_{A,B,C}.
    \end{align*}
    Now, applying the anti-concentration inequality \eqref{eq:large-deviation-hardcore} from \Cref{lem:distant-vs-all-hardcore}, we obtain:
    \begin{align*}
        \sum_{0\le A,B,C \le 2n} (A-B)^2 \alpha_{A,B,C}  \ge n^{4/3} \eta^2 \epsilon \sum_{0 \le A,B,C \le 2n} \alpha_{A,B,C} = n^{4/3} \eta^2 \epsilon \E[G \sim \+G(n,\Delta)]{Z_G}.
    \end{align*}
    Thus, inequality~\eqref{eq:target-hardcore-lb} holds by assuming \Cref{lem:distant-vs-all-hardcore}.

Next, applying the probabilistic method as described in \Cref{lem:probabilistic-method} to inequality \eqref{eq:target-hardcore-lb},
it ensures the existence of a graph $G=(L,R,E)$ such that:
\[
    \sum_{1 \le i,j \le 2n} \tp{\Pr[I \sim \mu_G]{\ell_i, \ell_j \in I} - \Pr[I \sim \mu_G]{\ell_i, r_j \in I}}
    = \frac{\sum_{1 \le i , j \le 2n} \tp{Z_G^{\ell_i \leftarrow +,\ell_j \leftarrow +} - Z_G^{\ell_i \leftarrow +,r_j \leftarrow +}}}{Z_G}=\Omega(n^{4/3}).
\]
Finally, by \Cref{lem:hardcore-SI-partition-function}, this implies $\lambda_{\max}(\Psi_{\mu_G})=\Omega(n^{1/3})$, completing the proof of the theorem.
\end{proof}

\subsubsection{A local limit theorem for large deviation}\label{sec:anti-concentration-hardcore}

It remains to prove the anti-concentration claimed in \Cref{lem:distant-vs-all-hardcore}, 
which is the main technical lemma for establishing the lower bound on the mixing time. 
The key idea of the proof is as follows:

\begin{itemize}
    \item \textbf{Near the center}: We express $\alpha_{A,B,C}$ as a ratio of probabilities that a random walk reaches the lattice point $(A,B,C)$ at some point in time. 
    This allows us to leverage a local limit theorem to obtain a tight approximation of $\alpha_{A,B,C}$ around the ``center of mass'' at $$(A,B,C) \approx \*m := \tp{\frac{2(\Delta-1)n}{\Delta^2},\frac{2(\Delta-1)n}{\Delta^2},\frac{2n}{\Delta^2}}.$$
\item \textbf{Far from the center}: For values of $(A,B,C)$ that deviate significantly from this center, we show that $\alpha_{A,B,C}$ decays exponentially.
This phenomenon is well studied and is critical in proving computational hardness in the super-critical regime~\cite{sly2012computational,GSV16}.
\end{itemize}
The following lemmas formalize  these intuitions.

\begin{lemma}\label{lem:center-alpha-hardcore}
Fix any constant $\Delta\ge 3$. There exist a finite $\gamma=\gamma(\Delta)>1$ and a positive-definite quadratic form $Q(\cdot,\cdot) = Q_\Delta(\cdot,\cdot)$ such that the following holds for all sufficiently large $n$.
There exists a factor $Z=Z(\Delta,n)$, depending only on $\Delta$ and $n$, 
such that for any integers $0 \le A,B,C \le 2n$ satisfying 
\[
\left\|(A,B,C) - \*m\right\|_{\infty} \le 2n^{2/3}, 
\]
where $\*m := \tp{\frac{2(\Delta-1)n}{\Delta^2}, \frac{2(\Delta-1)n}{\Delta^2}, \frac{2n}{\Delta^2}}$, the following holds:
\begin{align*}
    Z \le \frac{\alpha_{A,B,C}}{\exp\tp{-\frac{1}{n} Q(\theta_A+\theta_B,\theta_C) }} \le \gamma \cdot Z,
  \end{align*}
where $\*\theta=(\theta_A,\theta_B,\theta_C):=(A,B,C) - \*m$ represents the deviation from the ``center of mass'' $\*m$.
\end{lemma}

\begin{lemma}\label{lem:distant-alpha-hardcore}
Fix any constant $\Delta\ge 3$.
    There exist constants $\epsilon,\eta > 0$ such that the following holds for all sufficiently large~$n$.
    For integers $0 \le A,B,C \le 2n$ satisfying both the following conditions:
    \begin{enumerate}
        \item $\abs{A+B-(m_1+m_2)} > n^{2/3}$ or $\abs{C-m_3} > n^{2/3}$;
        \item $\abs{A-B} \le \eta\cdot  n^{2/3}$,
    \end{enumerate}
   where $\*m =(m_1,m_2,m_3)= \tp{\frac{2(\Delta-1)n}{\Delta^2}, \frac{2(\Delta-1)n}{\Delta^2}, \frac{2n}{\Delta^2}}$,
    it holds that
    \begin{align*}
        \alpha_{A,B,C} \le \exp\tp{-\epsilon\cdot n^{1/3}}\cdot \max_{0 \le a,b,c \le 2n} \alpha_{a,b,c}.
    \end{align*}
\end{lemma}

With  \Cref{lem:center-alpha-hardcore,lem:distant-alpha-hardcore}, we can now prove the anti-concentration claimed in \Cref{lem:distant-vs-all-hardcore}.

\begin{proof}[Proof of \Cref{lem:distant-vs-all-hardcore} (via \Cref{lem:center-alpha-hardcore,lem:distant-alpha-hardcore})]
It is sufficient to establish the anti-concentration property as described in~\eqref{eq:large-deviation-hardcore}.

    Note that if~\Cref{lem:distant-alpha-hardcore} holds with parameters $\epsilon,\eta >0$, then it automatically holds for all smaller positive parameters $\epsilon' < \epsilon$ and $\eta'<\eta$. 
    Therefore, we can assume without loss of generality that \Cref{lem:distant-alpha-hardcore} holds with parameters $\epsilon, \eta \in \tp{0,0.5}$. 

    Assume $n$ is sufficiently large.
    By~\Cref{lem:center-alpha-hardcore}, for any integer tuples $(A,B,C)$ and $(A',B',C)$ with $A+B=A'+B'$,
    if $\norm{(A,B,C)-\*m}_{\infty} \le 2n^{2/3}$ and $\norm{(A',B',C)-\*m}_{\infty} \le 2n^{2/3}$,
    then it holds
    \begin{align}\label{eq:alpha-same-hardcore}
    \alpha_{A,B,C} \le {\gamma}\cdot \alpha_{A',B',C},
    \end{align}
    where $\gamma=\gamma(\Delta)$ is the constant factor fixed in~\Cref{lem:center-alpha-hardcore}.
    
    Thus, for any integers $k$ and $C$ with $\abs{k - (\*m_1+\*m_2)} \le n^{2/3}$ and $\abs{C-\*m_3} \le n^{2/3}$, we have
    \begin{align}\label{eq:compare-1-hardcore}
        \sum_{\substack{A+B=k\\ \abs{A-B} \le \eta\cdot n^{2/3}}} \alpha_{A,B,C} \overset{(\star)}{\le} \gamma\cdot \sum_{\substack{A+B=k\\  \eta\cdot n^{2/3} < \abs{A-B} \le 2\eta\cdot n^{2/3}}} \alpha_{A,B,C} \le
        \gamma\cdot\sum_{\substack{A+B=k \\ \abs{A-B} > \eta\cdot n^{2/3}}} \alpha_{A,B,C},
    \end{align}
    where inequality $(\star)$ follows from~\eqref{eq:alpha-same-hardcore} and the fact that both $A$ and $B$ are $2n^{2/3}$-close to the center $\*m_1=\*m_2$ when $\abs{A+B-(\*m_1+\*m_2)} \le n^{2/3}$ and $\abs{A-B} \le 2\eta n^{2/3}$ are both satisfied.
    
    Therefore, by~\eqref{eq:compare-1-hardcore} and~\Cref{lem:distant-alpha-hardcore}, we have 
    \begin{align*}
    \sum_{\abs{A-B} \le \eta \cdot n^{2/3}} \alpha_{A,B,C} &\le \sum_{\substack{(A,B,C) \in \Omega\\ \abs{A-B} \le \eta\cdot n^{2/3}}} \alpha_{A,B,C} + \sum_{\substack{(A,B,C) \not\in \Omega\\ \abs{A-B} \le \eta\cdot n^{2/3}}} \alpha_{A,B,C}\\
    &\le \tp{\frac{\gamma}{\gamma+1} + (2n)^3 \exp\tp{-\epsilon n^{1/3}}} \sum_{0 \le A,B,C \le 2n} \alpha_{A,B,C},
    \end{align*}
    where $\Omega$ denotes the set of tuples $(A,B,C)$ with $\abs{A+B-(\*m_1+\*m_2)} \le n^{2/3}$ and $\abs{C-\*m_3} \le n^{2/3}$.
    
    Therefore,~\eqref{eq:large-deviation-hardcore} holds with parameters $\epsilon' = \frac{1}{2(\gamma+1)}$ and $\eta'=\eta$, when $n$ is sufficiently large.
\end{proof}

To complete the lower bound proof, we need to prove \Cref{lem:center-alpha-hardcore,lem:distant-alpha-hardcore}.
As discussed earlier, \Cref{lem:distant-alpha-hardcore} handles the case where the parameters $(A,B,C)$ are far from the center.
Specifically, it roughly says that $\alpha_{A,B,C}$ achieves its maximum approximately at $\*m=\tp{\frac{2(\Delta-1)n}{\Delta^2},\frac{2(\Delta-1)n}{\Delta^2},\frac{2n}{\Delta^2}}$, and decays exponentially as  the parameters move away from this center.
This phenomenon of exponential decay away from the center has been well established  in prior work on computational hardness in the super-critical regime, such as in~\cite{GSV16}.
The proof of \Cref{lem:distant-alpha-hardcore} follows a similar approach, and is thereby deferred to~\Cref{sec:append-hardcore}.

Next, it remains to prove~\Cref{lem:center-alpha-hardcore}, which addresses the case where the parameters $(A,B,C)$ are close to the center. 
This can be established by 
the local limit theorem for large deviation.

\begin{proof}[Proof of~\Cref{lem:center-alpha-hardcore}]
By definition of $\alpha_{A,B,C}$ and $\+G(n,\Delta)$, it holds that
\begin{align*}
\alpha_{A,B,C} &= \sum_{\substack{S, T  \subseteq [2n]\\ |S \setminus T| = A, |T \setminus S| = B, |S \cap T| = C}} \E[G \sim \+G(n,\Delta)]{w_G(\*1_{S,T})}\\
&= \sum_{\substack{S, T  \subseteq [2n]\\ |S \setminus T| = A, |T \setminus S| = B, |S \cap T| = C}} \lambda^{A+B+2C} \Pr[G \sim \+G(n,\Delta)]{\*1_{S,T} \in \+I(G)}\\
&=\sum_{\substack{S, T  \subseteq [2n]\\ |S \setminus T| = A, |T \setminus S| = B, |S \cap T| = C}} \lambda^{A+B+2C} \tp{\Pr[G \sim \+G(n,1)]{\*1_{S,T} \in \+I(G)}}^\Delta,
\end{align*}
where $\+I(G)$ denotes the collection of independent sets in $G$.

Let $\+R_{A,B,C}$ be the uniform distribution over subsets $S,T \subseteq [2n]$ such that $\abs{S \setminus T}=A$, $\abs{T \setminus S}=B$, and $\abs{S \cap T}=C$. 
Given the construction of $\+G(2n,1)$, the probability that $\*1_{S,T}$ forms an independent set in $G \sim \+G(2n,1)$ depends only on the sizes of  $S \setminus T$, $T \setminus S$ and $S \cap T$. 
Therefore, 
\begin{align}\label{eq:expected-weight-hardcore}
    \Pr[G \sim \+G(2n,1)]{\*1_{S,T} \in \+I(G)} = \Pr[\substack{G \sim \+G(2n,1)\\(S',T') \sim \+R_{A,B,C}}]{\*1_{S',T'} \in \+I(G)}.
\end{align}
Note that for any $G\in\+G(2n,1)$, 
the probability $\Pr[(S',T') \sim \+R_{A,B,C}]{\*1_{S',T'} \in \+I(G)}$ remains the same, since $\+R_{A,B,C}$ is a uniform distribution over possible subsets. Thus, we can reformulate \eqref{eq:expected-weight-hardcore} as:
 \begin{align}\label{eq:hardcore-alpha-formula}
 \Pr[G \sim \+G(2n,1)]{\*1_{S,T} \in \+I(G)} 
 &=\Pr[(S',T') \sim \+R_{A,B,C}]{\*1_{S',T'} \in \+I(G_0)} \notag\\
 &=
\frac{1}{\binom{2n}{A,B,C,2n-A-B-C}} \sum_{\substack{S,T \subseteq [2n]\\ \abs{S \setminus T} = A,\abs{T \setminus S} = B,\abs{S \cap T} = C}} \*1[\*1_{S,T} \in \+I(G_0)], 
 \end{align}
where $G_0=(L,R,E) \in \+G(2n,1)$ with edges defined such that $(\ell_{2i-1},r_{2i}), (\ell_{2i},r_{2i-1}) \in E$ for all $1 \le i \le n$, and the multinomial coefficient $\binom{2n}{A,B,C,2n-A-B-C} = \frac{(2n)!}{A!B!C!(2n-A-B-C)!}$ counts the number of possible subsets $S,T \subseteq [2n]$ in $\+R_{A,B,C}$.

Next, we express both the numerator and denominator in~\eqref{eq:hardcore-alpha-formula} as coefficients of generating polynomials.
Specifically, the denominator, the multinomial coefficient $\binom{2n}{A,B,C,2n-A-B-C}$, corresponds to the coefficient of the term $x^A y^B z^C$ in the generating polynomial $D = (1+x+y+z)^{2n}$:
\begin{align*}
    \binom{2n}{A,B,C,2n-A-B-C} = [x^Ay^Bz^C] \underbrace{\tp{1+x+y+z}^{2n}}_{=:D}.
\end{align*}
Similarly, the numerator in~\eqref{eq:hardcore-alpha-formula} can also be expressed as the coefficient of the term $x^A y^B z^C$  in another generating polynomial $N=\tp{1+2x+2y+2z+x^2+y^2}^n$:
\begin{align}\label{eq:num-hardcore}
\sum_{\substack{S,T \subseteq [2n]\\ \abs{S \setminus T} = A,\abs{T \setminus S} = B,\abs{S \cap T} = C}} \*1[\*1_{S,T} \in \+I(G_0)] = [x^A y^B z^C] \underbrace{\tp{1+2x+2y+2z+x^2+y^2}^n}_{=:N}.
\end{align}
To see this, note that the left-hand-side of~\eqref{eq:num-hardcore} counts the number of independent sets $\*1_{S,T}$ in $G_0$ with $\abs{S \setminus T} = A$, $\abs{T \setminus S} = B$ and $\abs{S \cap T} = C$. For all $1 \le i \le n$, the number of independent sets $\*1_{S,T} \subseteq \{\ell_{2i-1},\ell_{2i},r_{2i-1},r_{2i}\}$ satisfying $\abs{S \setminus T} = a$, $\abs{T \setminus S} = b$ and $\abs{S \cap T} = c$ is given by the coefficient of the term $x^a y^b z^c$ in the generating polynomial $N^{1/n} := 1+2x+2y+2z+x^2+y^2$. 
Then, equation~\eqref{eq:num-hardcore} follows from the product rule for generating polynomials.

Thus, $\alpha_{A,B,C}$ can be expressed as:
\begin{align}\label{eq:hardcore-alpha-coefficients}
    \alpha_{A,B,C} = \frac{\lambda^{A+B+2C}}{\tp{\left[x^A y^B z^C\right] D}^{\Delta-1}} \tp{\left[x^A y^B z^C\right] N}^{\Delta} = \frac{\tp{\left[x^A y^B z^C\right] N^\star}^{\Delta}}{\tp{\left[x^A y^B z^C\right] D^\star }^{\Delta-1}},
\end{align}
where $D^\star(x,y,z) = D \tp{\frac{x}{\Delta-1},\frac{y}{\Delta-1}, \frac{z}{(\Delta-1)^2}}$ and $N^\star(x,y,z) = N \tp{\frac{x}{\Delta-2},\frac{y}{\Delta-2},\frac{z}{(\Delta-2)^2}}$ are the re-weighted polynomials,
and the last equality follows from critical point $\lambda = \lambda_c(\Delta) = \frac{(\Delta-1)^{\Delta-1}}{(\Delta-2)^\Delta}$.

For the denominator, the coefficient $\left[x^A y^B z^C\right] D^\star$ can be interpreted as:
\[
\left[x^A y^B z^C\right] D^\star=F_D\cdot \Pr[]{\sum_{i=1}^{2n} \*{X}_i = (A,B,C)},
\]
where  $F_D$ is some factor depending only on $\Delta$ and $n$ but independent of $(A,B,C)$, 
and  $\*{X}_1,\ldots,\*{X}_{2n}$  are independent and identically distributed random variables defined as:
\begin{align*}
  \*X_i =
  \begin{cases}
    (0,0,0) & \text{with probability } \frac{(\Delta-1)^2}{\Delta^2}, \\
    (1,0,0) & \text{with probability } \frac{\Delta-1}{\Delta^2}, \\
    (0,1,0) & \text{with probability } \frac{\Delta-1}{\Delta^2}, \\
    (0,0,1) & \text{with probability } \frac{1}{\Delta^2}.
  \end{cases}
\end{align*}
By a straightforward calculation, we have:
\begin{align*}
  \E[]{\*X_i} &= \tp{\frac{\Delta-1}{\Delta^2}, \frac{\Delta-1}{\Delta^2}, \frac{1}{\Delta^2}}, \\
  \Sigma_{D} := \mathrm{Cov}[\*X_i] &= 
  \begin{bmatrix}
    u - u^2 & -u^2 & -uv\\
    -u^2 & u-u^2 & -uv\\
    -uv & -uv & v-v^2
  \end{bmatrix},
  \quad \text{where $u = \frac{\Delta-1}{\Delta^2}$ and $v = \frac{1}{\Delta^2}$.}
\end{align*}
By~\Cref{lem:llt} with $d = 3$, there exist constants $C_{D,\mathrm{LB}}, C_{D,\mathrm{UB}} > 0$ such that
\begin{align}\label{eq:gaussian-den-hardcore}
  C_{D,\mathrm{LB}} \le \frac{\Pr[]{\sum_{i=1}^{2n} \*X_i = (A,B,C)}}{n^{-3/2}\exp\tp{-\frac{1}{4n} \*\theta^T \Sigma_{D}^{-1} \*\theta}} \le C_{D,\mathrm{UB}},
\end{align}
where $\*\theta = (A,B,C) - 2n \E[]{\*X_1} =(A,B,C)-\*m$.

Similarly, for the numerator,  the coefficient $\left[x^A y^B z^C\right] N^\star$ can be interpreted as:
\[
\left[x^A y^B z^C\right] N^\star= F_N\cdot \Pr[]{\sum_{i=1}^{n} \*Y_i = (A,B,C)},
\]
where $F_N$ is some factor depending only on $\Delta$ and $n$ but independent of $(A,B,C)$, 
and $\*Y_1,\ldots,\*Y_n$ are independent and identically distributed random variables defined as:
\begin{align*}
  \*Y_i =
  \begin{cases}
    (0,0,0) & \text{with probability } \frac{(\Delta-2)^2}{\Delta^2},\\
    (1,0,0) & \text{with probability } \frac{2(\Delta-2)}{\Delta^2},\\
    (0,1,0) & \text{with probability } \frac{2(\Delta-2)}{\Delta^2},\\
    (0,0,1) & \text{with probability } \frac{2}{\Delta^2},\\
    (2,0,0) & \text{with probability } \frac{1}{\Delta^2},\\
    (0,2,0) & \text{with probability } \frac{1}{\Delta^2},
  \end{cases}
\end{align*}
whose expectation and covariance are given by:
\begin{align*}
  \E[]{\*Y_i} &= \tp{\frac{2(\Delta-1)}{\Delta^2}, \frac{2(\Delta-1)}{\Delta^2}, \frac{2}{\Delta^2}}, \\
  \Sigma_{N} := \mathrm{Cov}[\*Y_i] 
  &= 
    \begin{bmatrix}
      s - u^2 & -u^2 & -uv\\
      -u^2 & s-u^2 & -uv\\
      -uv & -uv & v-v^2
    \end{bmatrix},\\
    &\qquad \text{where $s = \frac{2}{\Delta}$, $u = \frac{2(\Delta - 1)}{\Delta^2}$, and $v = \frac{2}{\Delta^2}$.}
\end{align*}
By~\Cref{lem:llt}, there exist constants $C_{N,\mathrm{LB}}, C_{N,\mathrm{UB}} > 0$ such that
\begin{align}\label{eq:Gaussian-num-hardcore}
  C_{N,\mathrm{LB}} \le \frac{\Pr[]{\sum_{i=1}^n \*Y_i = (A,B,C)}}{n^{-3/2} \exp\tp{-\frac{1}{2n} \*\theta^T \Sigma_{N}^{-1} \*\theta}} \le C_{N,\mathrm{UB}},
\end{align}
where $\*\theta  = (A,B,C) - n \E[]{\*Y_1}=(A,B,C)-\*m$. 

Let $\Sigma := 2 \Delta \Sigma_{N}^{-1} - (\Delta-1) \Sigma_{D}^{-1}$. 
Combining~\eqref{eq:hardcore-alpha-coefficients}, \eqref{eq:gaussian-den-hardcore}, and~\eqref{eq:Gaussian-num-hardcore}, we have:
\begin{align}\label{eq:approx-hardcore}
   \frac{C_{N,\mathrm{LB}}^\Delta}{C_{D,\mathrm{UB}}^{\Delta-1}} \frac{F_{N}^\Delta}{F_D^{\Delta-1}} \le \frac{\alpha_{A,B,C}}{\exp\tp{-\frac{1}{4n} \*\theta^T \Sigma \*\theta}} \le \frac{C_{N,\mathrm{UB}}^\Delta }{C_{D,\mathrm{LB}}^{\Delta-1}} \frac{F_N^\Delta}{F_D^{\Delta-1}},
\end{align}

Therefore, we can choose parameters $Z =  \frac{C_{N,\mathrm{LB}}^\Delta}{C_{D,\mathrm{UB}}^{\Delta-1}} \frac{F_{N}^\Delta}{F_D^{\Delta-1}}$ and $\gamma = \frac{C^\Delta_{N,\mathrm{UB}} C^{\Delta-1}_{D,\mathrm{UB}}}{C^\Delta_{N,\mathrm{LB}} C^{\Delta-1}_{D,\mathrm{LB}}}$. 
To conclude the proof of~\Cref{lem:center-alpha-hardcore}, it remains to show that $\*\theta^{T} \Sigma \*\theta=Q(\theta_A+\theta_B,\theta_C)$, where $\*\theta=(\theta_A,\theta_B,\theta_C)$, for a positive-definite quadratic form $Q(\cdot,\cdot)$. 
It is  straightforward to verify:
\begin{align*}
\nonumber & \Sigma = 2\Delta \Sigma_{N}^{-1} - (\Delta-1)\Sigma_{D}^{-1}\\
&= 
\frac{\Delta^2}{(\Delta-1)(\Delta^3 - 4 \Delta^2 + 6 \Delta - 4)}
\begin{bmatrix}
   2 - 2 \Delta^2 + \Delta^3 &  2 - 2 \Delta^2 + \Delta^3 &  4-4\Delta +\Delta^3\\
   2 - 2 \Delta^2 + \Delta^3 &  2 - 2 \Delta^2 + \Delta^3 &  4-4\Delta +\Delta^3\\
  4-4\Delta +\Delta^3 &  4-4\Delta +\Delta^3 &  8 - 16  \Delta + 12 \Delta^2 - 4 \Delta^3 + \Delta^4
\end{bmatrix}.
\end{align*}
Let  $\kappa_1 =(2 - 2 \Delta^2 + \Delta^3)\kappa$, $\kappa_2 = (4 - 4 \Delta + \Delta^3)\kappa$, and $\kappa_3 =(8 - 16  \Delta + 12 \Delta^2 - 4 \Delta^3 + \Delta^4)\kappa$, where $\kappa=\frac{\Delta^2}{(\Delta-1)(\Delta^3 - 4 \Delta^2 + 6 \Delta - 4)}$. 
Therefore, the quadratic form $\*\theta^T \Sigma \*\theta$ can be expressed as
\begin{align*}
    \*\theta^T \Sigma \*\theta = \underbrace{\kappa_1 \cdot (\theta_A+\theta_B)^2 + 2\kappa_2\cdot (\theta_A+\theta_B) \theta_C + \kappa_3 \theta_C^2}_{=:Q(\theta_A+\theta_B,\theta_C)}.
\end{align*}
The positive-definiteness of $Q(\cdot,\cdot)$ follows from:
\begin{align*}
\kappa_1\cdot\kappa_3-\kappa_2^2
&=\left[(2-2\Delta^2+\Delta^3)(8 - 16  \Delta + 12 \Delta^2 - 4 \Delta^3 + \Delta^4) - (4-4\Delta+\Delta^3)^2\right]\cdot \kappa^2 \\
&=\Delta^2 (\Delta-1) (\Delta-2)^2 \tp{\Delta^2 - 2\Delta + 2}\cdot \frac{\Delta^4}{(\Delta-1)^2(\Delta^3 - 4 \Delta^2 + 6 \Delta - 4)^2}
> 0,
\end{align*}
for  $\Delta \ge 3$.
\end{proof}

\subsection{Critical Ising model on random regular bipartite graphs}\label{sec:lower-Ising}

    In this section, we will prove~\Cref{thm:lower-bound-Ising}. The structure of this section is similar to~\Cref{sec:lower-hardcore}. We first show the spectral independence lower bound via anti-concentration in~\Cref{sec:Ising-spectral-anticoncentrate}, and then establish the anti-concentration in~\Cref{sec:anti-concentration-Ising,sec:multi-edges}.

We focus on  critical instances of the anti-ferromagnetic Ising model on random bipartite graph with maximum degree at most $\Delta$ at the critical inverse temperature $\beta = -\beta_c(\Delta) = -\frac{1}{2} \log \frac{\Delta}{\Delta-2}$, with zero external field.
The random bipartite graph $G=(L,R,E)$ is constructed as follows:

\begin{itemize}
\item Fix a constant $\Delta \ge 3$. Let $L=\{\ell_1,\ell_2,\ldots,\ell_n\}$ and $R=\{r_1,r_2,\ldots,r_n\}$ be the vertex sets.
 \item Sample $\Delta$ perfect matchings $M_1,M_2,\ldots,M_\Delta$ between $L$ and $R$ uniformly and independently at random, and construct the graph $G=(L,R,E)$ where $E=\cup_{i=1}^\Delta M_i$.
\end{itemize}
We use $\+G(n,\Delta)$ to denote the law of the random bipartite graph $G=(L,R,E)$.

\subsubsection{Spectral correlation via anti-concentration}\label{sec:Ising-spectral-anticoncentrate}
Similar to the proof of~\Cref{thm:lower-bound-hardcore}, the lower bound in~\Cref{thm:lower-bound-Ising} can be  implied by the anti-concentration of the Gibbs measure on the random instances $G=(L,R,E)\sim\+G(n,\Delta)$.

Recall $L=\{\ell_1,\ell_2,\ldots,\ell_{n}\}$ and $R=\{r_1,r_2,\ldots,r_{n}\}$. For any subsets of indices $S,T \subseteq [n]$, the configuration $\*1_{S,T}$ takes $+1$ spins on $\{\ell_i,r_j \mid i \in S, j\in T\}$, and $-1$ otherwise. Formally, the configuration $\*1_{S,T}$ satisfies
\begin{align*}
\forall u \in L \cup R,\quad \*1_{S,T}(u)=
\begin{cases}
+1,& u = \ell_i \text{ for some } i \in S, \text{ or } r_j \text{ for some } j \in T\\
-1, & \text{otherwise.}
\end{cases}
\end{align*}
For integers $0 \le s,t \le n$, we define:
\begin{align*}
\alpha_{s,t}
&:=
\sum_{\sigma\in \Omega_{s,t}}\E[G \sim \+G(n,\Delta)]{w_G(\sigma)},
\end{align*}
where $\Omega_{s,t}$ is the collection of configurations $\*1_{S,T} \in \{-1,+1\}^{L \cup R}$ with subsets $S,T \subseteq [n]$ of indices satisfying $\abs{S}=s$ and $\abs{T} = t$. Here, $w_G(\sigma)$ represents the weight of configuration $\sigma$ in the Ising model on graph $G=(V,E)$. Specifically: $w_G(\sigma) = \exp\tp{2\beta \cdot m_G(\sigma)}$ with critical inverse temperature $\beta = -\beta_c(\Delta) = -\frac{1}{2} \log \frac{\Delta}{\Delta-2}$, where $m_G(\sigma)$ is the number of monochromatic edges, i.e. the number of edges $uv \in E$ with $\sigma_u = \sigma_v$.

The following lemma establishes an anti-concentration property, which provides a key to proving the spectral correlation described in \Cref{thm:lower-bound-Ising}.
\begin{lemma}\label{lem:distant-vs-all-Ising} 
 Fix any constant $\Delta \ge 3$.
    There exist constants $\epsilon,\eta \in (0,1)$ such that
    \begin{align}\label{eq:large-deviation-Ising}
      \sum_{\substack{|s-t|>\eta\cdot  n^{3/4}}} \alpha_{s,t} > \epsilon\cdot \sum_{0 \le s,t \le n} \alpha_{s,t}
    \end{align}
    holds for all sufficiently large $n$ and all integers $0 \le s,t \le n$. 
\end{lemma}

Similar to the argument in~\Cref{remark:hardcore-anticoncentrate},~\Cref{lem:distant-vs-all-Ising} captures an anti-concentration of the Ising Gibbs measure on the random instance $\+G(n,\Delta)$.

\Cref{thm:lower-bound-Ising} follows from \Cref{lem:distant-vs-all-Ising}. To see this, we first need to prove the following lemma.

\begin{lemma}\label{lem:Ising-SI-partition-function}
For any bipartite graph $G=(L,R,E)$ with positive measure in $\+G(n,\Delta)$, it holds 
\begin{align*}
\lambda_{\max}(\Psi_{\mu_G})  
&\ge 
\frac{2}{n} \cdot \sum_{1 \le i,j \le n} \tp{\Pr[\sigma \sim \mu_G]{\sigma_{\ell_i}= \sigma_{\ell_j} = +1} - \Pr[\sigma \sim \mu_G]{\sigma_{r_i} = \sigma_{r_j} = +1}}\\
& - \frac{2}{n} \sum_{1 \le i,j \le n} \tp{\Pr[\sigma \sim \mu_G]{\sigma_{\ell_i} = \sigma_{r_j} = +1 } + \Pr[\sigma \sim \mu_G]{\sigma_{r_i} = \sigma_{\ell_j} = +1}},
\end{align*}
where the vector $\*s\in \mathbb{R}^{L\cup R}$ is defined as:
$$\*s_u = 
\begin{cases}
    1 & u \in L,\\
    -1 & u \in R.
\end{cases}$$
\end{lemma}
\begin{proof}
Let $\mu=\mu_G$ denote the Gibbs distribution of the critical anti-ferromagnetic Ising model on $G=(L,R,E)$.
By the symmetry of $\pm 1$ spins, we have $\E[\sigma \sim \mu]{\sigma_u} = 0$ for all $u \in L \cup R$. Hence,
  \begin{align*}
  \forall u,v \in L \cup R, \quad \Psi_{\mu}(u,v) = 4\Pr[\sigma \sim \mu]{\sigma_u = \sigma_v = +1} - 1. 
  \end{align*}
  The quadratic form $\*s^\intercal \Psi_{\mu} \*s$ is  thus  expressed as:
  \begin{align}\label{eq:quadratic-form-Ising-2}
    \nonumber \*s^\intercal \Psi_{\mu} \*s &= 4 \sum_{1 \le i,j \le n} \tp{\Pr[\sigma \sim \mu]{\sigma_{\ell_i} = \sigma_{\ell_j} = +1} + \Pr[\sigma \sim \mu]{\sigma_{r_i} = \sigma_{r_j} = +1}}\\ 
    &- 4\sum_{1 \le i,j \le n}\tp{\Pr[\sigma \sim \mu]{\sigma_{\ell_i} = \sigma_{r_j} = +1} + \Pr[\sigma \sim \mu]{\sigma_{r_i} = \sigma_{\ell_j} = +1}}.
  \end{align}
  By Courant-Fischer theorem, the maximum eigenvalue of the influence matrix can be bounded by
  \begin{align}\label{eq:courant-fischer-Ising}
  \lambda_{\max}\tp{\Psi_{\mu}} \ge \frac{\*s^\intercal \Psi_{\mu} \*s}{\*s^\intercal \*s} = \frac{1}{2n} \*s^\intercal \Psi_{\mu} \*s.
  \end{align}
  \Cref{lem:Ising-SI-partition-function} then follows from~\eqref{eq:quadratic-form-Ising-2} and~\eqref{eq:courant-fischer-Ising}.
\end{proof}

\begin{proof}[Proof of~\Cref{thm:lower-bound-Ising} (via~\Cref{lem:distant-vs-all-Ising})]
Let $Z_G=\sum_{\sigma \in \{-1,+1\}^{L \cup R} }w_G(\sigma)$ denote the partition function, and for vertices $u,v\in L\cup R$, define
\[
Z_G^{u \leftarrow +, v \leftarrow +}:=\sum_{\substack{\sigma \in \{-1,+1\}^{L \cup R} \\ \sigma_u = \sigma_v = +1}}w_G(\sigma),
\]
which represents the portion of the partition function $Z_G$ contributed by the configuration $\sigma$ that both $\sigma_u$ and $\sigma_v$ are assigned with $+1$ spins.

First, we observe that the anti-concentration inequality~\eqref{eq:large-deviation-Ising} stated in \Cref{lem:distant-vs-all-Ising} translates to the following bound:
  \begin{align}\label{eq:target-Ising-lb}
    \frac{\sum_{1 \le i, j \le n} \E[G \sim \+G(n,\Delta)]{Z_G^{\ell_i \gets +,\ell_j \gets +} + Z_G^{r_i \gets +,r_j \gets +}- Z_G^{\ell_i \gets +,r_j \gets +} - Z_G^{r_i \gets +,\ell_j \gets +}}}{\E[G \sim \+G(n,\Delta)]{Z_G}} = \Omega(n^{3/2}).
  \end{align}
%
    To prove this, we first express the expected value $\sum_{1 \le i,j \le n} \E[G \sim \+G(n,\Delta)]{Z_G^{\ell_i \leftarrow +, \ell_j \leftarrow +}}$ as follows:
    \begin{align*}
      \sum_{1 \le i,j \le n} \E[G \sim \+G(n,\Delta)]{Z_G^{\ell_i \leftarrow +, \ell_j \leftarrow +}}
      &= \sum_{1 \le i,j \le n} \sum_{\substack{S \subseteq [n], T \subseteq [n]\\ i,j \in S}} \E[G \sim \+G(n,\Delta)]{w_G(\*1_{S,T})}\\ 
      &= \sum_{S \subseteq [n], T \subseteq [n]} |S|^2 \E[G \sim \+G(n,\Delta)]{w_G(\*1_{S,T})},
    \end{align*}
    where $\*1_{S,T}$ is a configuration that takes $+1$ spin on $\{\ell_i,r_j\mid i\in S,j\in T\}$ for $S,T \subseteq [2n]$. 
    Similarly, 
    \begin{align*}
      \sum_{1 \le i,j \le n} \E[G \sim \+G(n,\Delta)]{Z_G^{r_i\leftarrow +,r_j\leftarrow +}} = \sum_{S \subseteq [n], T \subseteq [n]} |T|^2 \E[G \sim \+G(n,\Delta)]{w_G(\*1_{S,T})},\\
      \sum_{1 \le i,j \le n} \E[G \sim \+G(n,\Delta)]{Z_G^{\ell_i\leftarrow +,r_j\leftarrow +}} = \sum_{S \subseteq [n], T \subseteq [n]} |S||T| \E[G \sim \+G(n,\Delta)]{w_G(\*1_{S,T})},\\
    \sum_{1 \le i,j \le n} \E[G \sim \+G(n,\Delta)]{Z_G^{r_i\leftarrow +,\ell_j\leftarrow +}} = \sum_{S \subseteq [n], T \subseteq [n]} |S||T| \E[G \sim \+G(n,\Delta)]{w_G(\*1_{S,T})}.
    \end{align*}
    Thus, the numerator of~\eqref{eq:target-Ising-lb} can be expressed as
     \begin{align*}
     &\sum_{1 \le i, j \le n} \E[G \sim \+G(n,\Delta)]{Z_G^{\ell_i \gets +,\ell_j \gets +} + Z_G^{r_i \gets +,r_j \gets +}- Z_G^{\ell_i \gets +,r_j \gets +} - Z_G^{r_i \gets +,\ell_j \gets +}} \\
     = & \sum_{S \subseteq [n], T \subseteq [n]} \tp{\abs{S}-\abs{T}}^2 \E[G \sim \+G(n,\Delta)]{w_G(\*1_{S,T})}\\
      = &\sum_{0\le s,t \le n} (s-t)^2 \alpha_{s,t}.
    \end{align*}
    Now, applying the anti-concentration inequality \eqref{eq:large-deviation-Ising} from \Cref{lem:distant-vs-all-Ising}, we obtain:
    \begin{align*}
        \sum_{0\le s,t \le n} (s-t)^2 \alpha_{s,t}  \ge n^{3/2} \eta^2 \epsilon \sum_{0 \le s,t \le n} \alpha_{s,t} = n^{3/2} \eta^2 \epsilon \E[G \sim \+G(n,\Delta)]{Z_G}.
    \end{align*}
    Thus, inequality~\eqref{eq:target-Ising-lb} holds by assuming \Cref{lem:distant-vs-all-Ising}.

Next, applying the probabilistic method as described in \Cref{lem:probabilistic-method} to inequality \eqref{eq:target-Ising-lb},
it ensures the existence of a graph $G=(L,R,E)$ such that:
\begin{align*}
&\sum_{1 \le i,j \le n} \tp{\Pr[\sigma \sim \mu_G]{\sigma_{\ell_i}= \sigma_{\ell_j} = +1} - \Pr[\sigma \sim \mu_G]{\sigma_{r_i} = \sigma_{r_j} = +1}}\\
 - &\sum_{1 \le i,j \le n} \tp{\Pr[\sigma \sim \mu_G]{\sigma_{\ell_i} = \sigma_{r_j} = +1 } + \Pr[\sigma \sim \mu_G]{\sigma_{r_i} = \sigma_{\ell_j} = +1}}\\
=  &\frac{\sum_{1 \le i, j \le n} \E[G \sim \+G(n,\Delta)]{Z_G^{\ell_i \gets +,\ell_j \gets +} + Z_G^{r_i \gets +,r_j \gets +}- Z_G^{\ell_i \gets +,r_j \gets +} - Z_G^{r_i \gets +,\ell_j \gets +}}}{\E[G \sim \+G(n,\Delta)]{Z_G}}
= \Omega(n^{3/2}).
\end{align*}
Finally, by \Cref{lem:Ising-SI-partition-function}, this implies $\lambda_{\max}(\Psi_{\mu_G})=\Omega(n^{1/2})$.
\end{proof}

\subsubsection{A local limit theorem for large deviation}\label{sec:anti-concentration-Ising}

It remains to prove the anti-concentration claimed in \Cref{lem:distant-vs-all-Ising}, 
which is the main technical lemma for establishing the lower bound on the mixing time. Instead of establishing~\Cref{lem:distant-vs-all-Ising}, we first prove an anti-concentration property for a variant of $\alpha_{s,t}$.

Alternatively, we consider the distribution $\widetilde{\+G}(n,\Delta)$ over $\Delta$-regular bipartite (multi-)graphs $G=(L,R,E)$, 
which differs slightly from $\+G(n,\Delta)$ by allowing parallel edges. 
Here, the set of edges $E=\uplus_{i=1}^\Delta M_i$ is the multiset union of the $\Delta$ perfect matchings $M_1,M_2,\ldots,M_\Delta$ between $L$ and $R$, sampled uniformly and independently at random.

We also similarly define $\widetilde{\alpha}_{s,t}$ for the distribution $\widetilde{\+G}(n,\Delta)$ of multigraph. Formally,
\begin{align*}
\widetilde{\alpha}_{s,t} = \sum_{\sigma \in \Omega_{s,t}} \E[G \sim \widetilde{\+G}(n,\Delta)]{w_G(\sigma)},
\end{align*}
where $\Omega_{s,t}$ is the collection of configurations $\*1_{S,T} \in \{-1,+1\}^{L \cup R}$ with subsets $S,T \subseteq [n]$ of indices satisfying $\abs{S}=s$ and $\abs{T} = t$. 

The following lemma describes the anti-concentration property for $\widetilde{\alpha}_{s,t}$.

\begin{lemma}\label{lem:distant-vs-all-Ising-multi}
 Fix any constant $\Delta \ge 3$.
    There exist constants $\epsilon,\eta \in (0,1)$ such that
    \begin{align}\label{eq:large-deviation-Ising-multi}
      \sum_{\substack{|s-t|>\eta\cdot  n^{3/4}}} \widetilde{\alpha}_{s,t} > \epsilon\cdot \sum_{0 \le s,t \le n} \widetilde{\alpha}_{s,t}
    \end{align}
    holds for all sufficiently large $n$ and all integers $0 \le s,t \le n$. 

\end{lemma}
The key idea of the proof is similar to the critical hardcore model, which is as follows: 

\begin{itemize}
    \item \textbf{Near the center}: We express $\widetilde{\alpha}_{s,t}$ as a ratio of probabilities that a random walk reaches the lattice point $(s,t)$ at some point in time. 
    This allows us to leverage a local limit theorem to obtain a tight approximation of $\widetilde{\alpha}_{s,t}$ around the ``center of mass'' at $$(s,t) \approx \*m := \tp{\frac{n}{2},\frac{n}{2}}.$$
\item \textbf{Far from the center}: For values of $(s,t)$ that deviate significantly from this center, we show that $\widetilde{\alpha}_{s,t}$ decays exponentially.
This phenomenon is well studied and is critical in proving computational hardness in the super-critical regime~\cite{sly2012computational,GSV16}.
\end{itemize}
The following lemmas formalize  these intuitions.

\begin{lemma}\label{lem:center-alpha-Ising}
Fix any constant $\Delta\ge 3$. There exist a finite $\gamma=\gamma(\Delta)>1$ and a constant $Q=Q(\Delta) > 0$ such that the following holds for all sufficiently large $n$.
There exists a factor $Z=Z(\Delta,n)$, depending only on $\Delta$ and $n$, 
such that for any integers $0 \le s,t \le n$ satisfying 
\[
\left\|(s,t) - \*m\right\|_{\infty} \le 2n^{3/4}, 
\]
where $\*m := \tp{\frac{n}{2}, \frac{n}{2}}$, the following holds:
\begin{align*}
    Z \le \frac{\widetilde{\alpha}_{s,t}}{\exp\tp{-\frac{Q}{n} \cdot (\theta_s+\theta_t)^2 }} \le \gamma \cdot Z,
  \end{align*}
where $\*\theta=(\theta_s,\theta_t) := (s,t) - \*m$ represents the deviation from the ``center of mass'' $\*m$.
\end{lemma}

\begin{lemma}\label{lem:distant-alpha-Ising}
Fix any constant $\Delta\ge 3$.
    There exist constants $\epsilon,\eta > 0$ such that the following holds for all sufficiently large~$n$.
    For integers $0 \le s,t \le 2n$ satisfying both the following conditions:
    \begin{enumerate}
        \item $\abs{s+t-(m_1+m_2)} > n^{3/4}$;
        \item $\abs{s-t} \le \eta\cdot  n^{3/4}$,
    \end{enumerate}
   where $\*m =(m_1,m_2)= \tp{\frac{n}{2},\frac{n}{2}}$,
    it holds that
    \begin{align*}
        \widetilde{\alpha}_{s,t} \le \exp\tp{-\epsilon\cdot n^{1/2}}\cdot \max_{0 \le s',t' \le 2n} \widetilde{\alpha}_{s',t'}.
    \end{align*}
\end{lemma}

With  \Cref{lem:center-alpha-Ising,lem:distant-alpha-Ising}, we can now prove the anti-concentration claimed in \Cref{lem:distant-vs-all-Ising}.

\begin{proof}[Proof of \Cref{lem:distant-vs-all-Ising-multi} (via \Cref{lem:center-alpha-Ising,lem:distant-alpha-Ising})]
It is sufficient to show the anti-concentration property as described in~\eqref{eq:large-deviation-Ising-multi}.

    Note that if~\Cref{lem:distant-alpha-Ising} holds with parameters $\epsilon,\eta >0$, then it automatically holds for all smaller positive parameters $\epsilon' < \epsilon$ and $\eta'<\eta$. 
    Therefore, we can assume without loss of generality that \Cref{lem:distant-alpha-Ising} holds with parameters $\epsilon, \eta \in \tp{0,0.5}$. 

    Assume $n$ is sufficiently large.
    By~\Cref{lem:center-alpha-Ising}, for any integer tuples $(s,t)$ and $(s',t')$ with $s+t=s'+t'$,
    if $\norm{(s,t)-\*m}_{\infty} \le 2n^{3/4}$ and $\norm{(s',t')-\*m}_{\infty} \le 2n^{3/4}$,
    then it holds
    \begin{align}\label{eq:alpha-same-Ising}
    \widetilde{\alpha}_{s,t} \le {\gamma}\cdot \widetilde{\alpha}_{s',t'},
    \end{align}
    where $\gamma=\gamma(\Delta)$ is the constant factor fixed in~\Cref{lem:center-alpha-hardcore}.
    
    Thus, for any integer $k$ with $\abs{k - (m_1+m_2)} \le n^{3/4}$, we have
    \begin{align}\label{eq:compare-1-Ising}
        \sum_{\substack{s+t=k\\ \abs{s-t} \le \eta\cdot n^{3/4}}} \widetilde{\alpha}_{s,t} \overset{(\star)}{\le} \gamma\cdot \sum_{\substack{s+t=k\\  \eta\cdot n^{3/4} < \abs{s-t} \le 2\eta\cdot n^{3/4}}} \widetilde{\alpha}_{s,t} \le
        \gamma\cdot\sum_{\substack{s+t=k \\ \abs{s-t} > \eta\cdot n^{3/4}}} \widetilde{\alpha}_{s,t},
    \end{align}
    where inequality $(\star)$ follows from~\eqref{eq:alpha-same-Ising} and the fact that both $s$ and $t$ are $2n^{3/4}$-close to the center $m_1=m_2$ when $\abs{s+t-(m_1+m_2)} \le n^{3/4}$ and $\abs{s-t} \le 2\eta \cdot n^{3/4}$ are both satisfied.
    
    Therefore, by~\eqref{eq:compare-1-Ising} and~\Cref{lem:distant-alpha-Ising}, we have 
    \begin{align*}
    \sum_{\abs{s-t} \le \eta \cdot n^{3/4}} \widetilde{\alpha}_{s,t} &\le \sum_{\substack{(s,t) \in \Omega\\ \abs{s-t} \le \eta\cdot n^{3/4}}} \widetilde{\alpha}_{s,t} + \sum_{\substack{(s,t) \not\in \Omega\\ \abs{s-t} \le \eta\cdot n^{3/4}}} \widetilde{\alpha}_{s,t}\\
    &\le \tp{\frac{\gamma}{\gamma+1} + n^2 \exp\tp{-\epsilon n^{1/2}}} \sum_{0 \le s,t \le n} \widetilde{\alpha}_{s,t},
    \end{align*}
    where $\Omega$ denotes the set of pairs $(s,t)$ with $\abs{s+t-(m_1+m_2)} \le n^{2/3}$.
    
    Therefore,~\eqref{eq:large-deviation-Ising-multi} holds with parameters $\epsilon' = \frac{1}{2(\gamma+1)}$ and $\eta'=\eta$, when $n$ is sufficiently large.
\end{proof}

To complete the lower bound proof, we need to prove \Cref{lem:center-alpha-hardcore,lem:distant-alpha-hardcore}.
As discussed earlier, \Cref{lem:distant-alpha-hardcore} handles the case where the parameters $(s,t)$ are far from the center.
Specifically, it says that $\widetilde{\alpha}_{s,t}$ achieves its maximum approximately at its ``center of mass'' $\*m=\tp{\frac{n}{2},\frac{n}{2}}$, and decays exponentially as  the parameters move away from this center.
This phenomenon of exponential decay away from the center has been well established  in prior work on computational hardness in the super-critical regime, such as in~\cite{GSV16}.
The proof of \Cref{lem:distant-alpha-hardcore} follows a similar approach, and is thereby deferred to~\Cref{sec:append-Ising}.

Next, it remains to prove~\Cref{lem:center-alpha-Ising}, which addresses the case where the parameters $(s,t)$ are close to the center. 
This can be established by 
the local limit theorem for large deviation.

\begin{proof}[Proof of~\Cref{lem:center-alpha-Ising}]
By the definition of $\widetilde{\alpha}_{s,t}$, the expected contribution $\widetilde{\alpha}_{s,t}$ satisfies
\begin{align}\label{eq:alpha-Ising-formula}
\nonumber\widetilde{\alpha}_{s,t} = \sum_{\substack{S,T \subseteq[n] \\ \abs{S}=s, \abs{T}=t }} \E[G\sim \widetilde{\+G}(n,\Delta)]{w_G(\*1_{S,T})}
&= \sum_{\substack{S,T \subseteq[n] \\ \abs{S}=s, \abs{T}=t }} \E[G \sim \widetilde{\+G}(n,\Delta)]{\exp\tp{2\beta \cdot m_G(\*1_{S,T})}}\\
&=\sum_{\substack{S,T \subseteq[n] \\ \abs{S}=s, \abs{T}=t }} \tp{\E[G \sim \+G(n,1)]{\exp\tp{2\beta \cdot m_G(\*1_{S,T})} }}^{\Delta}
\end{align}
where $m_G(\*1_{S,T})$ is the number of monochromatic edges.

Let $\+R_{s,t}$ be the uniform distribution over pairs of subsets $S \in \binom{[n]}{s}$ and $T \in \binom{[n]}{t}$. By symmetry, the expected value of $\exp\tp{2\beta \cdot m_G(\*1_{S,T})}$ over $G \sim \+G(n,1)$ only relies on $s$, $t$. Hence,
\begin{align}\label{eq:expected-weight-Ising}
\E[G \sim \+G(n,1)]{\exp\tp{2\beta \cdot m_G\tp{\*1_{S,T}}}} &= \E[G \sim \+G(n,1), (S',T') \sim \+R_{s,t}]{\exp\tp{2\beta \cdot m_G\tp{\*1_{S',T'}}}}.
\end{align}
Note that for any fixed graph $G \sim \+G(n,1)$, the expected value of $\exp\tp{2\beta \cdot m_G(\*1_{S',T'})}$ over distribution $(S',T') \sim \+R_{s,t}$ is the same as $\+R_{s,t}$ is a uniform distribution over fixed-size subsets $S$ and $T$. Therefore,~\eqref{eq:expected-weight-Ising} can be further expressed as
 \begin{align}\label{eq:alpha-formula-Ising}
 \E[G \sim \+G(n,1)]{\exp\tp{2\beta \cdot m_G(\*1_{S,T})} } = \tp{\binom{n}{s} \binom{n}{t}}^{-1} \sum_{\substack{S,T \subseteq[n] \\ \abs{S}=s, \abs{T}=t }} \exp\tp{2\beta \cdot m_{G_0}(\*1_{S,T})}
 \end{align}
where $G_0=(L,R,E)$ is a perfect matching with $(\ell_i,r_i) \in E$ for all $1 \le i \le n$.

Next, we express both the numerator and denominator in~\eqref{eq:alpha-formula-Ising} as coefficients of generating polynomials.
Specifically, the denominator, the binomial coefficients $\binom{n}{s} \binom{n}{t}$, corresponds to the coefficient of the term $x^s y^t$ in the generating polynomial $D = (1+x)^{n}(1+y)^n$:
\begin{align}\label{eq:den-Ising}
    \binom{n}{s}  \binom{n}{t} = [x^sy^t] \underbrace{\tp{1+x}^{n}\tp{1+y}^n}_{=:D}.
\end{align}
Similarly, the numerator in~\eqref{eq:alpha-formula-Ising} can also be expressed as the coefficient of the term $x^s y^t$ in another generating polynomial $N=\tp{\exp\tp{2\beta} x y + x + y + \exp\tp{2\beta}}^n$:
\begin{align}\label{eq:num-Ising}
\sum_{\substack{S,T \subseteq [n]\\ \abs{S} = s,\abs{T}=t}} \exp\tp{2\beta \cdot m_{G_0}(\*1_{S,T})} = [x^s y^t] \underbrace{\tp{\exp\tp{2\beta} x y + x + y + \exp\tp{2\beta}}^n}_{=:N}.
\end{align}
This follows from a similar argument as in the proof of~\eqref{eq:num-hardcore}.
Thus, by~\eqref{eq:alpha-Ising-formula},~\eqref{eq:alpha-formula-Ising},~\eqref{eq:den-Ising} and~\eqref{eq:num-Ising}, $\widetilde{\alpha}_{s,t}$ can be expressed as:
\begin{align}\label{eq:alpha-express-Ising}
\widetilde{\alpha}_{s,t} = \frac{\tp{[x^sy^t] N}^{\Delta}}{\tp{[x^s y^t] D}^{\Delta-1}}.
\end{align}

For the denominator, the coefficient $[x^s y^t] D$ can be interpreted as: 
\begin{align*}
[x^s y^t] D = F_D \cdot \Pr[]{\sum_{i=1}^{n} \*{X}_i = (s,t)},
\end{align*} 
where $F_D$ is some factor depending on $\Delta$ and $n$ but independent of $(s,t)$, and $\*{X}_1,\*{X}_2,\ldots,\*{X}_{n}$ are independent and identically distributed random variables defined as:
\begin{align*}
  \*X_i =
  \begin{cases}
    (0,0) & \text{with probability } 1/4 \\
    (1,0) & \text{with probability } 1/4 \\
    (0,1) & \text{with probability } 1/4 \\
    (1,1) & \text{with probability } 1/4.
  \end{cases}
\end{align*}
By a straightforward calculation, we have
\begin{align*}
  \E[]{\*X_i} = \tp{1/2,1/2} \quad \text{and} \quad
  \Sigma_{D} := \mathrm{Cov}[\*X_i] = 
  \begin{bmatrix}
    1/4 & 0\\
    0 & 1/4
  \end{bmatrix}
  .
\end{align*}
By~\Cref{lem:llt} with $d = 2$, there exist constants $C_{D,\mathrm{LB}}, C_{D,\mathrm{UB}} > 0$ such that
\begin{align}\label{eq:gaussian-den-Ising}
  C_{D,\mathrm{LB}} \le \frac{\Pr[]{\sum_{i=1}^{n} \*X_i = (s,t)}}{n^{-1} \exp\tp{-\frac{1}{2n} \*\theta^T \Sigma_{D}^{-1} \*\theta}} \le C_{D,\mathrm{UB}},
\end{align}
where $\*\theta := (s,t) - n \E[]{\*X_1} = \tp{s-m_1,t-m_2}$.

Similarly, the coefficient $[x^s y^t] N$  can be interpreted as 
\begin{align*}
[x^s y^t] N = F_N \cdot \Pr[]{\sum_{i=1}^n Y_i = (s,t)},
\end{align*} 
where $F_N$ is some factor depending on $\Delta$ and $n$ but independent of $(s,t)$, and $\*Y_1,\*Y_2,\ldots,\*Y_n$ are independent and identically distributed random variables defined as
\begin{align*}
\*Y_i=
\begin{cases}
(0,0) & \text{with probability } \frac{\exp\tp{2\beta}}{2(1+\exp\tp{2\beta})}\\
(0,1) & \text{with probability } \frac{1}{2(1+\exp\tp{2\beta})}\\
(1,0) & \text{with probability } \frac{1}{2(1+\exp\tp{2\beta})}\\
(1,1) & \text{with probability } \frac{\exp\tp{2\beta}}{2(1+\exp\tp{2\beta})}.
\end{cases}
\end{align*}
Recall that $\beta = -\beta_c(\Delta) = - \frac{1}{2} \log \frac{\Delta}{\Delta-2}$. A straightforward calculation yields
\begin{align*}
\E[]{\*Y_i}=\tp{1/2,1/2} \text{ and } \Sigma_{N} = \mathrm{Cov}[\*Y_i] = 
\begin{bmatrix}
1/4 & - 1/(4\Delta-4)\\ 
-1/(4\Delta-4)& 1/4
\end{bmatrix}
.
\end{align*}
By~\Cref{lem:llt}, there exist constants $C_{N,\mathrm{LB}},C_{N,\mathrm{UB}} > 0$ satisfying
\begin{align}\label{eq:guassian-num-Ising}
    C_{N,\mathrm{LB}}\le \frac{\Pr[]{\sum_{i=1}^n \*Y_i = (s,t)}}{n^{-1} \exp \tp{-\frac{1}{2n} \*\theta^T \Sigma_{N}^{-1} \*\theta}} \le C_{N,\mathrm{UB}},
\end{align}
where $\*\theta=(s,t) - n \E[]{\*Y_1} = \tp{s-m_1,t-m_2}$.
Combining~\eqref{eq:alpha-express-Ising},~\eqref{eq:gaussian-den-Ising} and~\eqref{eq:guassian-num-Ising}, it holds that
\begin{align}\label{eq:approx-Ising}
   \frac{C_{N,\mathrm{LB}}^\Delta}{C_{D,\mathrm{UB}}^{\Delta-1}} \frac{F_{N}^\Delta}{F_D^{\Delta-1}} \le \frac{\widetilde{\alpha}_{s,t}}{\exp\tp{-\frac{1}{2n} \*\theta^T \Sigma \*\theta}} \le \frac{C_{N,\mathrm{UB}}^\Delta }{C_{D,\mathrm{LB}}^{\Delta-1}} \frac{F_N^\Delta}{F_D^{\Delta-1}},
\end{align}
where $\Sigma = \Delta \Sigma_{N}^{-1} - (\Delta-1) \Sigma_{D}^{-1} = \frac{4(\Delta-1)}{\Delta-2}  
\begin{bmatrix}
    1 & 1\\
    1 & 1
\end{bmatrix}$. Hence,~\Cref{lem:center-alpha-Ising} holds with $Q = \frac{2(\Delta-1)}{\Delta-2}$, $Z =  \frac{C_{N,\mathrm{LB}}^\Delta}{C_{D,\mathrm{UB}}^{\Delta-1}} \frac{F_{N}^\Delta}{F_D^{\Delta-1}}$ and $\gamma = \frac{C^\Delta_{N,\mathrm{UB}} C^{\Delta-1}_{D,\mathrm{UB}}}{C^\Delta_{N,\mathrm{LB}} C^{\Delta-1}_{D,\mathrm{LB}}}$.
\end{proof}

\subsubsection{Removing parallel edges}\label{sec:multi-edges}

Finally, we prove that the expected contribution $\alpha_{s,t}$ and $\widetilde{\alpha}_{s,t}$ are close in value.
\begin{lemma}\label{lem:comparison-Ising}
  For sufficiently large $n$ and all pairs of integers $0 \le s,t \le n$, it holds that
  \begin{align*}
    \exp(-10^4 \Delta^3) \alpha_{s,t} \le \widetilde{\alpha}_{s,t} \le \alpha_{s,t}.
  \end{align*}
\end{lemma}

With this comparison lemma, we can verify~\eqref{eq:target-Ising-lb} in~\Cref{lem:distant-vs-all-Ising} for $\alpha_{s,t}$.

\begin{proof}[Proof of~\Cref{lem:distant-vs-all-Ising}]
    Let $\eta$ and $\epsilon$ be the constants in~\Cref{lem:distant-vs-all-Ising-multi}. 
    By~\Cref{lem:distant-vs-all-Ising-multi,lem:comparison-Ising},
    \begin{align*}
        \sum_{\abs{s-t} > \eta n^{3/4}} \alpha_{s,t} \ge \sum_{\abs{s-t} > \eta n^{3/4}} \widetilde{\alpha}_{s,t} > \epsilon \sum_{0 \le s,t\le n} \widetilde{\alpha}_{s,t} \ge \epsilon \exp\tp{-10^4 \Delta^3} \sum_{0 \le s,t \le n} \alpha_{s,t}.
    \end{align*}
    Thus,~\Cref{lem:distant-vs-all-Ising} holds with $\epsilon' = \exp\tp{-10^4\Delta^3} \epsilon$ and $\eta' = \eta$.
\end{proof}

It only remains to prove~\Cref{lem:comparison-Ising}, which is obtained by a coupling argument.
\begin{proof}[Proof of~\Cref{lem:comparison-Ising}]
  For any (multi)graph $\widetilde{G}=(V,E)$ and configuration $\*1_{S,T}$, let $G$ be the graph obtained by removing parallel edges in $\widetilde{G}$. By the definition of $w_G(\*1_{S,T}) = \exp\tp{2\beta \cdot m_G(\*1_{S,T})}$ and $\beta = -\beta_c(\Delta) < 0$, it holds that $w_{\widetilde{G}}(\*1_{S,T}) \le w_{G}(\*1_{S,T})$, implying $\widetilde{\alpha}_{s,t} \le \alpha_{s,t}$ by the definition of $\widetilde{\alpha}_{s,t}$ and $\alpha_{s,t}$. We now focus on proving the lower bound of $\widetilde{\alpha}_{s,t}$. 
  
  Fix subsets $S, T \subseteq [n]$ with $\abs{S} = s$ and $\abs{T} = t$. It holds that
  \begin{align*}
    \frac{\widetilde{\alpha}_{s,t}}{\alpha_{s,t}} &= \frac{\E[G \sim \widetilde{\+G}(n,\Delta)]{w_G(\*1_{S,T})}}{\E[G \sim \+G(n,\Delta)]{w_G(\*1_{S,T} )}}
    = \frac{\tp{\E[G \sim \+G(n,1)]{w_G(\*1_{S,T})}}^{\Delta}}{\E[G \sim \+G(n,\Delta)]{w_G(\*1_{S,T})}} \\
    &= \prod_{i=1}^{\Delta} \frac{\E[G_{i-1} \sim \+G(n,i-1)]{w_{G_{i-1}}(\*1_{S,T})}\E[H \sim \+G(n,1)]{w_H(\*1_{S,T})}}{\E[G_i \sim \+G(n,i)]{w_{G_i}(\*1_{S,T})}},
  \end{align*}
  Therefore, it suffices to prove that \begin{align*}
    \frac{\E[G_{i-1} \sim \+G(n,i-1)]{w_{G_{i-1}}(\*1_{S,T})}\E[H \sim \+G(n,1)]{w_H(\*1_{S,T})}}{\E[G_i \sim \+G(n,i)]{w_{G_i}(\*1_{S,T})}} \ge \exp\tp{-10^4\Delta^2}
  \end{align*} for all $1 \le i \le \Delta$.
  Without loss of generality, we only consider the case where $i = \Delta$. Note that
  \begin{align*}
    \E[G \sim \+G(n,\Delta)]{w_{G}(\*1_{S,T})} = \E[\substack{H \sim \+G(n,1)\\ G \sim \+G(n,\Delta-1)}]{w_G(\*1_{S,T}) w_{H \setminus G}(\*1_{S,T})}. 
  \end{align*}
  Therefore, it suffices to prove that for any fixed perfect matching $H \sim \+G(n,1)$,
  \begin{align}\label{eq:target-2-multi}
    \E[G \sim \+G(n,\Delta-1)]{w_G(\*1_{S,T})} w_H(\*1_{S,T}) \ge \exp\tp{- 10^4 \Delta^2}\E[G \sim \+G(n,\Delta-1)]{w_G(\*1_{S,T}) w_{H \setminus G}(\*1_{S,T})}.
  \end{align}
  Without loss of generality, we assume that the matching $H=(L,R,\{(\ell_i,r_i)\}_{1 \le i \le n})$. 
  Recall that $G=(L,R,E) \sim \+G(n,\Delta-1)$ is the union of $\Delta-1$ independent random perfect matchings.
  Let $e_{(i-1)n+1},e_{(i-1)n+2},\ldots,e_{in}$ be the edges in the $i$-th matching shuffled uniformly at random. 
  Denote the graph $(L,R,E\setminus \{e_1,e_2,\ldots,e_i\})$ by $G^{-i}$ (in particular, we use the convention $G^{-0} = G$).
  We now claim that for all $1 \le i \le n(\Delta-1)$,
  \begin{align}\label{eq:target-multi}
    \frac{\E[G \sim \+G(n,\Delta-1)]{w_G(\*1_{S,T}) w_{H \setminus G^{-(i-1)}}(\*1_{S,T})}}{\E[G \sim \+G(n,\Delta-1)]{w_G(\*1_{S,T}) w_{H \setminus G^{-i}}(\*1_{S,T})}} \ge 1 - \frac{\exp\tp{-16\beta}}{n} \ge \exp\tp{-\frac{10^4}{n}},
  \end{align}
  where the last inequality follows from $\exp\tp{-2\beta} = \exp\tp{2\beta_c(\Delta)} \le 3$ and $n$ being sufficiently large.
  Assuming the correctness of~\eqref{eq:target-multi}, the inequality~\eqref{eq:target-2-multi} holds by telescoping.

  To simplify the notations, we will use
  \begin{align*}
    A_{i,j} &:= \E[G \sim \+G(n,\Delta-1)]{w_G(\*1_{S,T}) w_{H \setminus G^{-i}}(\*1_{S,T}) \mid e_j \in H} \\
    B_{i,j} &:= \E[G \sim \+G(n,\Delta-1)]{w_G(\*1_{S,T}) w_{H \setminus G^{-i}}(\*1_{S,T}) \mid e_j \not\in H}.
  \end{align*}
  Note that when $e_{i} \not\in H$, $H\setminus G^{-(i-1)} = H\setminus G^{-i}$ and we have $w_{H\setminus G^{-(i-1)}}(\*1_{S,T}) = w_{H\setminus G^{-i}}(\*1_{S,T})$.
  After taking the expectation, this implies that
  \begin{align} \label{eq:ABi-relationship}
    B_{i-1,i} = B_{i,i}.
  \end{align}
  Moreover, we claim that the following inequality holds:
  \begin{align} \label{eq:target-3-multi}
    A_{i,i} \leq \exp(-16\beta)B_{i,i},
  \end{align}
  which implies that
  \begin{align} \label{eq:cor-3-multi}
    \nonumber A_{i,i} &\leq \exp(-16\beta) \tp{\frac{1}{n} A_{i,i} + \frac{n-1}{n} B_{i,i}} \\
    &= \exp(-16\beta) \E[G \sim \+G(n,\Delta-1)]{w_G(\*1_{S,T}) w_{H \setminus G^{-i}}(\*1_{S,T})}.
  \end{align}
  In order to prove~\eqref{eq:target-multi}, by the law of total expectation, we have
  \begin{align*}
    &\E[G \sim \+G(n,\Delta-1)]{w_G(\*1_{S,T}) w_{H \setminus G^{-(i-1)}}(\*1_{S,T})} \\
    &\quad = \frac{1}{n} A_{i-1,i} + \frac{n-1}{n} B_{i-1,i} \\
    (\text{by \eqref{eq:ABi-relationship}}) &\quad \geq \frac{n-1}{n} B_{i,i} = \E[G \sim \+G(n,\Delta-1)]{w_G(\*1_{S,T}) w_{H \setminus G^{-i}}(\*1_{S,T})} - \frac{1}{n} A_{i,i} \\
    (\text{by \eqref{eq:cor-3-multi}}) &\quad \geq \tp{1 - \frac{\exp(-16\beta)}{n}} \E[G \sim \+G(n,\Delta-1)]{w_G(\*1_{S,T}) w_{H \setminus G^{-i}}(\*1_{S,T})}.
  \end{align*}

  Now, we only left to prove~\eqref{eq:target-3-multi}.
  We finish the proof by constructing a coupling $\+C$ between $\+G(n,\Delta-1)$ condition on the event $e_i \in H$ and the event $e_i \not \in H$, such that $(G_1=(L,R,E_1),G_2=(L,R,E_2)) \sim \+C$ satisfying $|E_1 \oplus E_2| \le 4$:
  we sample $G_2 = (L,R,E_2)$ from $\+G(n,\Delta-1)$ condition on the event $e_i \not\in H$. In this case, suppose $e_i = (\ell_u, r_v)$ (since $e_i \not\in H$) and $e_j=(\ell_w,r_u)$ is the edge that origins from the same matching, i.e. $\lfloor \frac{i-1}{n}\rfloor = \lfloor \frac{j-1}{n} \rfloor$. We construct $E_1$ by substituting $e_i$ with $(\ell_u,r_u)$ and $e_j$ with $(\ell_w,r_v)$.

  With this coupling $\+C$, by definition, $A_{i,i} - \exp(-16\beta)B_{i,i}$ equals to 
  \begin{align*}
    \E[(G_1,G_2) \sim \+C]{w_{G_1}(\*1_{S,T}) w_{H \setminus G_1^{-i}}(\*1_{S,T}) - \exp\tp{-16\beta} w_{G_2}(\*1_{S,T}) w_{H \setminus G_2^{-i}}(\*1_{S,T})} \le 0,
  \end{align*}
  as the size of symmetric difference between edge sets of $G_1$ and $G_2$ is at most $4$.
  This finishes the proof of \eqref{eq:target-3-multi} and hence the whole proof.
\end{proof}


\section{Deterministic Algorithm}
\label{sec:deterministic}
In this section, we show that the noising/denoising framework can be applied to deterministic counting.
This results in sub-exponential time deterministic algorithms for approximating the partition function for the critical hardcore and Ising models, 
thereby proving \Cref{thm:main-deterministic}.

Let $\-{wt}:2^{[n]} \to \=R_{\geq 0}$ be a non-negative weight function over all subsets of $[n]$.
We normalize $\-{wt}(\cdot)$ into a probability distribution $\mu$ as follow
\begin{align*}
    \forall T \subseteq [n], \quad \mu(T) = \frac{\-{wt}(T)}{Z},
\end{align*}
where $Z = \sum_T \-{wt}(T)$ is the partition function for $\mu$.
For $\theta \in (0,1)$ and $S \subseteq [n]$, we define 
\begin{align} \label{eq:def-Z-S-theta}
    Z_{S, \theta} := \sum_{T\supseteq S} \-{wt}(T) (1 - \theta)^{\abs{T}}.
\end{align}
This quantity $Z_{S,\theta}$ corresponds to the partition function for the distribution $((1-\theta)*\mu)(\cdot\mid S)$, which is defined as $(1-\theta)*\mu$ (as described in \eqref{eq:def-mu-external-field}) conditioned on $S$ being occupied.

\begin{lemma} \label{lem:subexp-counting}
Fix $\theta \in (0, 1)$ and $\epsilon_0 > 0$.
Suppose for every $S\subseteq [n]$, $\hat{Z}_{S,\theta}$ is an approximation of $Z_{S,\theta}$ within relative error $(1\pm \epsilon_0)$.
Then, for any $\epsilon > 0$, let $k = \ctp{\e^2 n \cdot \frac{\theta}{1 - \theta} + \log \frac{2}{\epsilon}}$. The quantity:
$$\hat{Z}=\sum_{S:\abs{S} < k}\tp{\frac{\theta}{1-\theta}}^{\abs{S}} \hat{Z}_{S,\theta}$$
approximates the partition function $Z$ within relative error $(1\pm (\epsilon_0 + \epsilon))$.
\end{lemma}

\begin{proof}
Let $(X_t)_{t\in [0,1]}$ be the continuous-time down walk described in \Cref{definition-decreasing-process}.
Fix any subset $S \subseteq [n]$. 
According to the law of total probability applied to $(X_t)_{t\in [0,1]}$, we have:
\begin{align*}
    \Pr{X_\theta = S}
    &= \sum_{T\subseteq[n]} \Pr{X_0 = T} \Pr{X_\theta = S \mid X_0 = T} \\
    &= \sum_{T\supseteq S} \mu(T) \theta^{\abs{S}} (1 - \theta)^{\abs{T}-\abs{S}}.
\end{align*}
As a well-defined probability space, we know that:
\begin{align*}
1 = \sum_{S\subseteq[n]} \Pr{X_\theta = S} = \sum_{S\subseteq[n]} \sum_{T\supseteq S} \mu(T) \theta^{\abs{S}} (1 - \theta)^{\abs{T} - \abs{S}}.
\end{align*}
This implies
\begin{align*}
    Z = \sum_{S\subseteq[n]} \sum_{T\supseteq S} \-{wt}(T) \theta^{\abs{S}} (1 - \theta)^{\abs{T} - \abs{S}} = \sum_{S\subseteq[n]} \tp{\frac{\theta}{1-\theta}}^{\abs{S}} Z_{S,\theta}.
\end{align*}
By definition, $Z_{S,\theta} \leq Z_{\emptyset, \theta}\le Z$. 
Now, let $k = \ctp{\e^2 n \cdot \frac{\theta}{1 - \theta} + \log 2/\epsilon}$. We consider the contribution of sets $S\subseteq[n]$ with $\abs{S}\ge k$ in the partition function:
\begin{align*}
    \frac{\sum_{S:\abs{S}\geq k}\tp{\frac{\theta}{1-\theta}}^{\abs{S}} Z_{S,\theta}}{Z_{\emptyset, \theta}} 
    &\leq \sum_{\ell=k}^n \binom{n}{\ell} \tp{\frac{\theta}{1 - \theta}}^\ell 
    \leq \sum_{\ell=k}^n \tp{\frac{\e n}{\ell} \cdot \frac{\theta}{1 - \theta}}^\ell \\
    &\le \tp{\frac{\e n}{k} \frac{\theta}{1 - \theta}}^k \cdot \frac{1}{1 - \frac{\e n}{k} \frac{\theta}{1 - \theta}} \leq \epsilon.
\end{align*}
This implies that $\sum_{S:\abs{S} < k}\tp{\frac{\theta}{1-\theta}}^{\abs{S}} Z_{S,\theta}$ approximates $Z$ as follows:
\begin{align*}
   Z\geq \sum_{S:\abs{S} < k}\tp{\frac{\theta}{1-\theta}}^{\abs{S}} Z_{S,\theta} \geq Z - \epsilon Z_{\emptyset, \theta} \geq (1 - \epsilon)Z.
\end{align*}
Now, suppose we have an approximation $\hat{Z}_{S,\theta}\in [(1-\epsilon_0)Z_{S,\theta},(1+\epsilon_0)Z_{S,\theta}]$ of $Z_{S,\theta}$. 
Then,
\[
\hat{Z}=\sum_{S:\abs{S} < k}\tp{\frac{\theta}{1-\theta}}^{\abs{S}} \hat{Z}_{S,\theta}\ge (1-\epsilon_0)\sum_{S:\abs{S} < k}\tp{\frac{\theta}{1-\theta}}^{\abs{S}} Z_{S,\theta}\ge (1-\epsilon_0)(1 - \epsilon)Z\ge (1-\epsilon_0-\epsilon)Z,
\]
and 
\[
\hat{Z}=\sum_{S:\abs{S} < k}\tp{\frac{\theta}{1-\theta}}^{\abs{S}} \hat{Z}_{S,\theta} \le (1+\epsilon_0)\sum_{S:\abs{S} < k}\tp{\frac{\theta}{1-\theta}}^{\abs{S}} Z_{S,\theta}\le (1+\epsilon_0)Z. \qedhere
\]
\end{proof}

According to \Cref{lem:subexp-counting},
to estimate $Z$, 
we can enumerate every subset $S\subseteq[n]$ with $\abs{S} < k$ and estimate $Z_{S,\theta}$, which is typically more sub-critical and thus  expected to be easier to estimate.
The overhead of this enumeration in the running time  is bounded by:
\begin{align*}
     O(n^k) & = \exp\tp{O\left(  \frac{\theta }{1-\theta}\cdot n\log n + \log\frac{1}{\epsilon}\cdot \log n\right)}. 
\end{align*}
This is formalized in the following corollary of \Cref{lem:subexp-counting}.

\begin{corollary} \label{cor:subexp-deterministic-counting}
    Suppose there exist $\theta \in (0, 1)$, $\epsilon_0 > 0$, and $T > 0$ such that for every $S\subseteq [n]$, there is a deterministic algorithm $\+A$ which can approximate $Z_{S,\theta}$ in time $T$ within relative error $(1\pm \epsilon_0)$.
    Then, for any $\epsilon > 0$, there exists a deterministic algorithm that can approximate the partition function $Z$ within relative error $(1\pm (\epsilon_0 + \epsilon))$ in time $T\cdot\exp\tp{O\left(  \frac{\theta }{1-\theta}\cdot n\log n + \log\frac{1}{\epsilon}\cdot \log n\right)}$.
\end{corollary}

\subsection{Applications to critical hardcore and Ising models}

It is well established that, within the $\delta$-uniqueness regime, the correlation decay method provides an approximation deterministically  for the partition function $Z_{S, \theta}$, with a running time of $(\frac{n}{\epsilon_0})^{O(\frac{\log \Delta}{\delta})}$.

\begin{theorem}[\text{\cite{li2012correlation}}] \label{lem:deterministic-easy-regime}
    Let $G = ([n], E)$ be a graph with $n$ vertices and maximum degree $\Delta$, and $\mu$ be the Gibbs distribution for either the hardcore model or the Ising model on $G$.
    Moreover, let $\tau_\Lambda$ be a pinning on $\Lambda \subseteq [n]$, and $u \in [n] \setminus \Lambda$ be a free vertex.
    Given any $\epsilon_1 > 0$, there exists a deterministic algorithm $\+M(\mu, \tau_\Lambda, u)$ that can approximate the marginal probability $\mu^{\tau_\Lambda}_u$ within relative error $(1 \pm \epsilon_1)$, with the following guarantees on the running time:
    \begin{itemize}
        \item \textbf{Hardcore model}: If the model is the hardcore model and the parameter satisfies $\lambda \leq (1 - \theta) \lambda_c(\Delta)$, the algorithm $\+M(\mu, \tau_\Lambda, u)$ halts in time $(1/\epsilon_1)^{O(\frac{\log \Delta}{\theta})}$.
        \item \textbf{Ising model}: If the model is the critical Ising model, the algorithm $\+M((1-\theta) *\mu, \tau_\Lambda, u)$ halts in time $(1/\epsilon_1)^{O(\frac{\log \Delta}{\theta^2})}$.
    \end{itemize}
\end{theorem}

Given $S \subseteq [n]$ and $\theta \in (0, 1)$, let $\nu = \theta * \mu$. It holds that
\begin{align*}
  \frac{(1-\theta)^{\abs{S}} \-{wt}(S)}{Z_{\emptyset, \theta}} = \nu(\*1_S \uplus \*0_{V\setminus S}) = \nu^{\*1_S}_{V\setminus S}(\*0) \cdot \nu_S(\*1) = \nu^{\*1_S}_{V\setminus S}(\*0) \cdot \frac{Z_{S,\theta}}{Z_{\emptyset, \theta}},
\end{align*}
which implies that $Z_{S,\theta} = {(1-\theta)^{\abs{S}} \-{wt}(S)}/{\nu^{\*1_S}_{V\setminus S}(\*0)}$.
Thus, to approximate $Z_{S,\theta}$ with relative error $(1 \pm \epsilon_0)$, we only need to approximate $\nu^{\*1_S}_{V\setminus S}(\*0)$ with relative error $(1 \pm \epsilon_0)$, which can be achieved by using \Cref{lem:deterministic-easy-regime} with $\epsilon_1 = \Theta(\epsilon_0/n)$.

This leads to the following consequences for the deterministic approximation of $Z_{S,\theta}$, given an arbitrary subset $S \subseteq [n]$, parameter $\theta \in (0, 1)$, and error bound $\epsilon_0 > 0$:
\begin{itemize}
    \item \textbf{Hardcore model}: For the critical hardcore model, the value of $Z_{S,\theta}$ can be approximated within relative error $(1 \pm \epsilon_0)$ in deterministic time $(n/\epsilon_0)^{O\left(\frac{\log \Delta}{\theta}\right)}$.
    \item \textbf{Ising model}: For the critical Ising model, the value of $Z_{S,\theta}$ can be approximated within relative error $(1 \pm \epsilon_0)$ in deterministic time $(n/\epsilon_0)^{O\left(\frac{\log \Delta}{\theta^2}\right)}$.
\end{itemize}

We can then apply \Cref{cor:subexp-deterministic-counting} 
to obtain sub-exponential time deterministic approximation algorithms for the partition function of the critical hardcore and Ising models.

By taking $\theta = \Theta(n^{-1/2})$ for the  hardcore model and $\theta = \Theta(n^{-1/3})$ for the  Ising model, we obtain the following results.

\begin{theorem}
  For any $\epsilon > 0$, there exists a deterministic algorithm that approximates  
  the partition function $Z$ of the critical hardcore model (where $\lambda = \lambda_c(\Delta)$) 
  on graphs with $n$ vertices and maximum degree $\Delta$ 
  within relative error $(1 \pm \epsilon)$. 
  This algorithm runs in time
  \begin{align*}
    \exp\tp{O(\sqrt{n} \log^2(n/\epsilon))}.
  \end{align*}
\end{theorem}

\begin{theorem}
  For any $\epsilon > 0$, there exists a deterministic algorithm that approximates
  the partition function $Z$ of the critical Ising model (where $(\Delta-1) \tanh \abs{\beta} = 1$, and $\lambda = 1$)
  on graphs with $n$ vertices and maximum degree $\Delta$ 
  within relative error $(1 \pm \epsilon)$. 
  This algorithm runs in time
  \begin{align*}
    \exp\tp{O(n^{2/3} \log^2(n/\epsilon))}.
  \end{align*}
\end{theorem}


\section{Conclusion}
\label{sec:conclusion}
This paper explores both upper and lower bounds on the mixing time of Glauber dynamics for the critical hardcore and Ising models. The upper bounds are derived from an information-theoretic reformulation of localization schemes. 
\fixed{To obtain sharper bounds for the critical Ising model with interaction norm $1$, we further establish an $O(\sqrt{n})$-spectral independence bound at criticality.}
The lower bounds are obtained by the universality of spectral independence and the \emph{anti-concentration} phenomenon of the Gibbs distribution of critical hardcore/Ising model on typical random symmetric regular bipartite graphs.

To conclude, we list several open problems that remain to be addressed:
\begin{itemize}
\item What is the mixing time for general anti-ferromagnetic two-spin systems at criticality? The current approach for proving upper bounds on the mixing time handles the hardcore model and the Ising model separately (see also the discussion in \cref{rmk:FD-Ising}). A potential direction is to design a new noising-denoising localization process that unifies the field dynamics and proximal sampler.
\item How to close the gap between the current upper and lower bounds on the mixing time for hardcore and graphical Ising models at criticality? 
\item Can we design an efficient deterministic algorithm for approximate counting at criticality, with or even without the bounded-degree assumption? 
\end{itemize}

\section*{Acknowledgments}
\fixed{The authors thank Tianhui Jiang for helpful discussions on the error that appeared in the previous version of the paper, as summarized in \Cref{rmk:bug}.}


\bibliographystyle{alpha}
\bibliography{critical-Dec-2025.bib}

@article{CCCYZ25+,
	title={Rapid Mixing on Random Regular Graphs beyond Uniqueness},
	author={Chen, Xiaoyu and Chen, Zejia and Chen, Zongchen and Yin, Yitong and Zhang, Xinyuan},
	journal={arXiv preprint arXiv:2504.03406},
	year={2025}
}

@article{BBD24,
	title={Stochastic dynamics and the {P}olchinski equation: an introduction},
	author={Bauerschmidt, Roland and Bodineau, Thierry and Dagallier, Benoit},
	journal={Probability Surveys},
	volume={21},
	pages={200--290},
	year={2024},
	publisher={The Institute of Mathematical Statistics and the Bernoulli Society}
}

@article{PS24+,
	title={Polynomial Mixing of the critical {G}lauber Dynamics for the {I}sing Model},
	author={Prodromidis, Kyprianos-Iason and Sly, Allan},
	journal={arXiv preprint arXiv:2411.10318},
	year={2024}
}

@inproceedings{CE22,
  author       = {Yuansi Chen and
                  Ronen Eldan},
  title        = {Localization Schemes: {A} Framework for Proving Mixing Bounds for
                  Markov Chains (extended abstract)},
  booktitle    = {63rd {IEEE} Annual Symposium on Foundations of Computer Science, {FOCS}
                  2022, Denver, CO, USA, October 31 - November 3, 2022},
  pages        = {110--122},
  publisher    = {{IEEE}},
  year         = {2022},
  doi          = {10.1109/FOCS54457.2022.00018},
}

@article {GSV16,
    AUTHOR = {Galanis, Andreas and {\v S}tefankovi{\v c}, Daniel and Vigoda, Eric},
     TITLE = {Inapproximability of the partition function for the
              antiferromagnetic {I}sing and hard-core models},
   JOURNAL = {Combin. Probab. Comput.},
  FJOURNAL = {Combinatorics, Probability and Computing},
    VOLUME = {25},
      YEAR = {2016},
    NUMBER = {4},
     PAGES = {500--559},
}

@inproceedings{anari2024universality,
  author       = {Nima Anari and
                  Vishesh Jain and
                  Frederic Koehler and
                  Huy Tuan Pham and
                  Thuy{-}Duong Vuong},
  title        = {Universality of Spectral Independence with Applications to Fast Mixing
                  in Spin Glasses},
  booktitle    = {Proceedings of the Annual {ACM-SIAM} Symposium on Discrete Algorithms (SODA)},
  pages        = {5029--5056},
  year         = {2024}
}

@inproceedings {li2012correlation,
    AUTHOR = {Li, Liang and Lu, Pinyan and Yin, Yitong},
     TITLE = {Correlation decay up to uniqueness in spin systems},
 BOOKTITLE = {Proceedings of the Annual {ACM-SIAM} Symposium on Discrete Algorithms (SODA)},
     PAGES = {67--84},
      YEAR = {2012},
}

@article{Mon23,
	title={Sampling, diffusions, and stochastic localization},
	author={Montanari, Andrea},
	journal={arXiv preprint arXiv:2305.10690},
	year={2023}
}

@article {EM22,
    AUTHOR = {El Alaoui, Ahmed and Montanari, Andrea},
     TITLE = {An information-theoretic view of stochastic localization},
   JOURNAL = {IEEE Trans. Inform. Theory},
  FJOURNAL = {Institute of Electrical and Electronics Engineers.
              Transactions on Information Theory},
    VOLUME = {68},
      YEAR = {2022},
    NUMBER = {11},
     PAGES = {7423--7426},
}

@article {chen2023rapid,
    AUTHOR = {Chen, Zongchen and Liu, Kuikui and Vigoda, Eric},
     TITLE = {Rapid mixing of {G}lauber dynamics up to uniqueness via
              contraction},
   JOURNAL = {SIAM J. Comput.},
  FJOURNAL = {SIAM Journal on Computing},
    VOLUME = {52},
      YEAR = {2023},
    NUMBER = {1},
     PAGES = {196--237},
}

@article {ding2009meanfield,
    AUTHOR = {Ding, Jian and Lubetzky, Eyal and Peres, Yuval},
     TITLE = {The mixing time evolution of {G}lauber dynamics for the
              mean-field {I}sing model},
   JOURNAL = {Comm. Math. Phys.},
  FJOURNAL = {Communications in Mathematical Physics},
    VOLUME = {289},
      YEAR = {2009},
    NUMBER = {2},
     PAGES = {725--764},
}

@inproceedings{chen2023near,
  author       = {Xiaoyu Chen and
                  Xinyuan Zhang},
  title        = {A Near-Linear Time Sampler for the {I}sing Model with External Field},
  booktitle    = {Proceedings of the Annual {ACM-SIAM} Symposium on Discrete Algorithms (SODA)},
  pages        = {4478--4503},
  year         = {2023},
}

@inproceedings{sly2010computational,
  author       = {Allan Sly},
  title        = {Computational Transition at the Uniqueness Threshold},
  booktitle    = {Proceedings of the Annual IEEE Symposium on Foundations of Computer Science (FOCS)},
  pages        = {287--296},
  year         = {2010},
}

@inproceedings{sly2012computational,
  author       = {Allan Sly and
                  Nike Sun},
  title        = {The Computational Hardness of Counting in Two-Spin Models on d-Regular
                  Graphs},
  booktitle    = {Proceedings of the Annual IEEE Symposium on Foundations of Computer Science (FOCS)},
  pages        = {361--369},
  year         = {2012},
}

@incollection {chen1998trilogy,
    AUTHOR = {Chen, Mu-Fa},
     TITLE = {Trilogy of couplings and general formulas for lower bound of
              spectral gap},
 BOOKTITLE = {Probability towards 2000 ({N}ew {Y}ork, 1995)},
    SERIES = {Lect. Notes Stat.},
    VOLUME = {128},
     PAGES = {123--136},
 PUBLISHER = {Springer, New York},
      YEAR = {1998},
}

@inproceedings{chen2021rapid,
  author       = {Xiaoyu Chen and
                  Weiming Feng and
                  Yitong Yin and
                  Xinyuan Zhang},
  title        = {Rapid mixing of {G}lauber dynamics via spectral independence for all
                  degrees},
  booktitle    = {Proceedings of the Annual IEEE Symposium on Foundations of Computer Science (FOCS)},
  pages        = {137--148},
  year         = {2021},
}

@article {mossel2009hardness,
    AUTHOR = {Mossel, Elchanan and Weitz, Dror and Wormald, Nicholas},
     TITLE = {On the hardness of sampling independent sets beyond the tree
              threshold},
   JOURNAL = {Probab. Theory Related Fields},
  FJOURNAL = {Probability Theory and Related Fields},
    VOLUME = {143},
      YEAR = {2009},
    NUMBER = {3-4},
     PAGES = {401--439},
}

@article {dyer2002counting,
    AUTHOR = {Dyer, Martin and Frieze, Alan and Jerrum, Mark},
     TITLE = {On counting independent sets in sparse graphs},
   JOURNAL = {SIAM J. Comput.},
  FJOURNAL = {SIAM Journal on Computing},
    VOLUME = {31},
      YEAR = {2002},
    NUMBER = {5},
     PAGES = {1527--1541},
}

@article {galanis2015inapproximability,
    AUTHOR = {Galanis, Andreas and {\v S}tefankovi{\v c}, Daniel and Vigoda, Eric},
     TITLE = {Inapproximability for antiferromagnetic spin systems in the
              tree nonuniqueness region},
   JOURNAL = {J. ACM},
  FJOURNAL = {Journal of the ACM},
    VOLUME = {62},
      YEAR = {2015},
    NUMBER = {6},
     PAGES = {Art. 50, 60},
}

@article{anari2021entropic,
  title={Entropic independence {I}: Modified log-{S}obolev inequalities for fractionally log-concave distributions and high-temperature {I}sing models},
  author={Anari, Nima and Jain, Vishesh and Koehler, Frederic and Pham, Huy Tuan and Vuong, Thuy-Duong},
  journal={arXiv preprint arXiv:2106.04105},
  year={2021}
}

@inproceedings{lee2021structured,
  author       = {Yin Tat Lee and
                  Ruoqi Shen and
                  Kevin Tian},
  title        = {Structured Logconcave Sampling with a Restricted Gaussian Oracle},
  booktitle    = {The Annual Conference on Learning Theory (COLT)},
  pages        = {2993--3050},
  year         = {2021},
}

@inproceedings {chen2021optimal,
    AUTHOR = {Chen, Zongchen and Liu, Kuikui and Vigoda, Eric},
     TITLE = {Optimal mixing of {G}lauber dynamics: entropy factorization
              via high-dimensional expansion},
 BOOKTITLE = {Proceedings of the Annual ACM Symposium on Theory of Computing (STOC)},
     PAGES = {1537--1550},
      YEAR = {2021},
}

@inproceedings {anari2020spectral,
    AUTHOR = {Anari, Nima and Liu, Kuikui and {Oveis Gharan}, Shayan},
     TITLE = {Spectral independence in high-dimensional expanders and
              applications to the hardcore model},
 BOOKTITLE = {Proceedings of the Annual IEEE Symposium on Foundations of Computer Science (FOCS)},
     PAGES = {1319--1330},
      YEAR = {2020},
}

@book{levin2017markov,
  title={Markov chains and mixing times},
  author={Levin, David A and Peres, Yuval},
  volume={107},
  year={2017},
  publisher={American Mathematical Soc.}
}

@article {bobkov2006modified,
    AUTHOR = {Bobkov, Sergey G. and Tetali, Prasad},
     TITLE = {Modified logarithmic {S}obolev inequalities in discrete
              settings},
   JOURNAL = {J. Theoret. Probab.},
  FJOURNAL = {Journal of Theoretical Probability},
    VOLUME = {19},
      YEAR = {2006},
    NUMBER = {2},
     PAGES = {289--336},
}

@article {caputo2015approximate,
    AUTHOR = {Caputo, Pietro and Menz, Georg and Tetali, Prasad},
     TITLE = {Approximate tensorization of entropy at high temperature},
   JOURNAL = {Ann. Fac. Sci. Toulouse Math. (6)},
  FJOURNAL = {Annales de la Facult\'{e} des Sciences de Toulouse. Math\'{e}matiques.
              S\'{e}rie 6},
    VOLUME = {24},
      YEAR = {2015},
    NUMBER = {4},
     PAGES = {691--716},
}

@article{stefankovic2023lecture,
      title={Lecture Notes on Spectral Independence and Bases of a Matroid: Local-to-Global and Trickle-Down from a Markov Chain Perspective}, 
      author={Daniel {\v S}tefankovi{\v c} and Eric Vigoda},
      journal={arXiv preprint arXiv:2307.13826},
      year={2023},
}

@inproceedings {weitz2006counting,
    AUTHOR = {Weitz, Dror},
     TITLE = {Counting independent sets up to the tree threshold},
 BOOKTITLE = {Proceedings of the Annual ACM Symposium on Theory of Computing (STOC)},
     PAGES = {140--149},
      YEAR = {2006},
}

@book {Bar06,
    AUTHOR = {Barvinok, Alexander},
     TITLE = {Combinatorics and complexity of partition functions},
    SERIES = {Algorithms and Combinatorics},
    VOLUME = {30},
 PUBLISHER = {Springer, Cham},
      YEAR = {2016},
     PAGES = {vi+303},
}

@article {peters2019conjecture,
    AUTHOR = {Peters, Han and Regts, Guus},
     TITLE = {On a conjecture of {S}okal concerning roots of the
              independence polynomial},
   JOURNAL = {Michigan Math. J.},
  FJOURNAL = {Michigan Mathematical Journal},
    VOLUME = {68},
      YEAR = {2019},
    NUMBER = {1},
     PAGES = {33--55},
}

@inproceedings{CFYZ22optimal,
  title={Optimal mixing for two-state anti-ferromagnetic spin systems},
  author={Chen, Xiaoyu and Feng, Weiming and Yin, Yitong and Zhang, Xinyuan},
  booktitle={Proceedings of the Annual IEEE Symposium on Foundations of Computer Science (FOCS)},
  pages={588--599},
  year={2022},
}

@article {LS12,
    AUTHOR = {Lubetzky, Eyal and Sly, Allan},
     TITLE = {Critical {I}sing on the square lattice mixes in polynomial
              time},
   JOURNAL = {Comm. Math. Phys.},
  FJOURNAL = {Communications in Mathematical Physics},
    VOLUME = {313},
      YEAR = {2012},
    NUMBER = {3},
     PAGES = {815--836},
}

@article {BD24,
    AUTHOR = {Bauerschmidt, Roland and Dagallier, Benoit},
     TITLE = {Log-{S}obolev inequality for near critical {I}sing models},
   JOURNAL = {Comm. Pure Appl. Math.},
  FJOURNAL = {Communications on Pure and Applied Mathematics},
    VOLUME = {77},
      YEAR = {2024},
    NUMBER = {4},
     PAGES = {2568--2576},
}

@article {blanca2019phase,
    AUTHOR = {Blanca, Antonio and Chen, Yuxuan and Galvin, David and
              Randall, Dana and Tetali, Prasad},
     TITLE = {Phase coexistence for the hard-core model on {$\mathbb{Z}^2$}},
   JOURNAL = {Combin. Probab. Comput.},
  FJOURNAL = {Combinatorics, Probability and Computing},
    VOLUME = {28},
      YEAR = {2019},
    NUMBER = {1},
     PAGES = {1--22},
}

@article{richter1958multi,
  title={Multi-dimensional local limit theorems for large deviations},
  author={Richter, V},
  journal={Theory of Probability \& Its Applications},
  volume={3},
  number={1},
  pages={100--106},
  year={1958},
}

@inproceedings{chen2024fast,
  author       = {Zongchen Chen and
                  Yuzhou Gu},
  title        = {Fast Sampling of \emph{b}-Matchings and \emph{b}-Edge Covers},
  booktitle    = {Proceedings of the Annual {ACM-SIAM} Symposium on Discrete Algorithms (SODA)},
  pages        = {4972--4987},
  year         = {2024},
}

@inproceedings{chen2023coloring,
  author       = {Zongchen Chen and
                  Kuikui Liu and
                  Nitya Mani and
                  Ankur Moitra},
  title        = {Strong Spatial Mixing for Colorings on Trees and its Algorithmic Applications},
  booktitle    = {Proceedings of the Annual IEEE Symposium on Foundations of Computer Science (FOCS)},
  pages        = {810--845},
  year         = {2023},
}

@article {chen2024stability,
    AUTHOR = {Chen, Zongchen and Liu, Kuikui and Vigoda, Eric},
     TITLE = {Spectral independence via stability and applications to
              {H}olant-type problems},
   JOURNAL = {TheoretiCS},
  FJOURNAL = {TheoretiCS},
    VOLUME = {3},
      YEAR = {2024},
     PAGES = {Art. 16, 49},
}

@inproceedings{anari2022entropic,
  author       = {Nima Anari and
                  Vishesh Jain and
                  Frederic Koehler and
                  Huy Tuan Pham and
                  Thuy{-}Duong Vuong},
  title        = {Entropic independence: optimal mixing of down-up random walks},
  booktitle    = {Proceedings of the Annual ACM Symposium on Theory of Computing (STOC)},
  pages        = {1418--1430},
  year         = {2022},
}

@inproceedings{liu2021coupling,
  author       = {Kuikui Liu},
  title        = {From Coupling to Spectral Independence and Blackbox Comparison with
                  the Down-Up Walk},
  booktitle    = {APPROX/RANDOM},
  pages        = {32:1--32:21},
  year         = {2021},
}

@inproceedings{blanca2022mixing,
  author       = {Antonio Blanca and
                  Pietro Caputo and
                  Zongchen Chen and
                  Daniel Parisi and
                  Daniel Stefankovic and
                  Eric Vigoda},
  title        = {On Mixing of Markov Chains: Coupling, Spectral Independence, and Entropy
                  Factorization},
  booktitle    = {Proceedings of the Annual {ACM-SIAM} Symposium on Discrete Algorithms (SODA)},
  pages        = {3670--3692},
  year         = {2022},
}

@article {EKZ22,
    AUTHOR = {Eldan, Ronen and Koehler, Frederic and Zeitouni, Ofer},
     TITLE = {A spectral condition for spectral gap: fast mixing in
              high-temperature {I}sing models},
   JOURNAL = {Probab. Theory Related Fields},
  FJOURNAL = {Probability Theory and Related Fields},
    VOLUME = {182},
      YEAR = {2022},
    NUMBER = {3-4},
     PAGES = {1035--1051},
}

@inproceedings{Kun24,
  title={Optimality of {G}lauber dynamics for general-purpose {I}sing model sampling and free energy approximation},
  author={Kunisky, Dmitriy},
  booktitle={Proceedings of the 2024 Annual ACM-SIAM Symposium on Discrete Algorithms (SODA)},
  pages={5013--5028},
  year={2024},
}

@article{GKK24,
  title={On sampling from {I}sing models with spectral constraints},
  author={Galanis, Andreas and Kalavasis, Alkis and Kandiros, Anthimos Vardis},
  journal={arXiv preprint arXiv:2407.07645},
  year={2024}
}

@article {song2019counting,
    AUTHOR = {Song, Renjie and Yin, Yitong and Zhao, Jinman},
     TITLE = {Counting hypergraph matchings up to uniqueness threshold},
   JOURNAL = {Inform. and Comput.},
  FJOURNAL = {Information and Computation},
    VOLUME = {266},
      YEAR = {2019},
     PAGES = {75--96},
}

@article {vigoda2001note,
    AUTHOR = {Vigoda, Eric},
     TITLE = {A note on the {G}lauber dynamics for sampling independent
              sets},
   JOURNAL = {Electron. J. Combin.},
  FJOURNAL = {Electronic Journal of Combinatorics},
    VOLUME = {8},
      YEAR = {2001},
    NUMBER = {1},
     PAGES = {Research Paper 8, 8},
}

@article {ALOV24,
    AUTHOR = {Anari, Nima and Liu, Kuikui and {Oveis Gharan}, Shayan and
              Vinzant, Cynthia},
     TITLE = {Log-concave polynomials {II}: {H}igh-dimensional walks and an
              {FPRAS} for counting bases of a matroid},
   JOURNAL = {Ann. of Math. (2)},
  FJOURNAL = {Annals of Mathematics. Second Series},
    VOLUME = {199},
      YEAR = {2024},
    NUMBER = {1},
     PAGES = {259--299},
}

@article {KO20,
    AUTHOR = {Kaufman, Tali and Oppenheim, Izhar},
     TITLE = {High order random walks: beyond spectral gap},
   JOURNAL = {Combinatorica},
  FJOURNAL = {Combinatorica. An International Journal on Combinatorics and
              the Theory of Computing},
    VOLUME = {40},
      YEAR = {2020},
    NUMBER = {2},
     PAGES = {245--281},
}

@article {CGM21,
    AUTHOR = {Cryan, Mary and Guo, Heng and Mousa, Giorgos},
     TITLE = {Modified log-{S}obolev inequalities for strongly log-concave
              distributions},
   JOURNAL = {Ann. Probab.},
  FJOURNAL = {The Annals of Probability},
    VOLUME = {49},
      YEAR = {2021},
    NUMBER = {1},
     PAGES = {506--525},
}

@inproceedings{AL20,
  title={Improved analysis of higher order random walks and applications},
  author={Alev, Vedat Levi and Lau, Lap Chi},
  booktitle={Proceedings of the Annual ACM Symposium on Theory of Computing (STOC)},
  pages={1198--1211},
  year={2020}
}

@incollection {follmer2005,
    AUTHOR = {F{\"o}llmer, Hans},
     TITLE = {An entropy approach to the time reversal of diffusion
              processes},
 BOOKTITLE = {Stochastic differential systems ({M}arseille-{L}uminy, 1984)},
    SERIES = {Lect. Notes Control Inf. Sci.},
    VOLUME = {69},
     PAGES = {156--163},
PUBLISHER = {Springer, Berlin},
      YEAR = {1985},
}

@incollection {follmer2006,
    AUTHOR = {F{\"o}llmer, Hans},
     TITLE = {Time reversal on {W}iener space},
 BOOKTITLE = {Stochastic processes---mathematics and physics ({B}ielefeld,
              1984)},
    SERIES = {Lecture Notes in Math.},
    VOLUME = {1158},
     PAGES = {119--129},
PUBLISHER = {Springer, Berlin},
      YEAR = {1986},
}

@article {Lehec13,
    AUTHOR = {Lehec, Joseph},
     TITLE = {Representation formula for the entropy and functional
              inequalities},
   JOURNAL = {Ann. Inst. Henri Poincar\'{e} Probab. Stat.},
  FJOURNAL = {Annales de l'Institut Henri Poincar\'{e} Probabilit\'{e}s et
              Statistiques},
    VOLUME = {49},
      YEAR = {2013},
    NUMBER = {3},
     PAGES = {885--899},
}

@article {ELS20,
    AUTHOR = {Eldan, Ronen and Lehec, Joseph and Shenfeld, Yair},
     TITLE = {Stability of the logarithmic {S}obolev inequality via the
              {F}\"{o}llmer process},
   JOURNAL = {Ann. Inst. Henri Poincar\'{e} Probab. Stat.},
  FJOURNAL = {Annales de l'Institut Henri Poincar\'{e} Probabilit\'{e}s et
              Statistiques},
    VOLUME = {56},
      YEAR = {2020},
    NUMBER = {3},
     PAGES = {2253--2269},
}

@article {EM20,
    AUTHOR = {Eldan, Ronen and Mikulincer, Dan},
     TITLE = {Stability of the {S}hannon-{S}tam inequality via the {F}\"{o}llmer
              process},
   JOURNAL = {Probab. Theory Related Fields},
  FJOURNAL = {Probability Theory and Related Fields},
    VOLUME = {177},
      YEAR = {2020},
    NUMBER = {3-4},
     PAGES = {891--922},
}

@article {Mik21,
    AUTHOR = {Mikulincer, Dan},
     TITLE = {Stability of {T}alagrand's {G}aussian transport-entropy
              inequality via the {F}\"{o}llmer process},
   JOURNAL = {Israel J. Math.},
  FJOURNAL = {Israel Journal of Mathematics},
    VOLUME = {242},
      YEAR = {2021},
    NUMBER = {1},
     PAGES = {215--241},
}

@article {KP23,
    AUTHOR = {Klartag, Bo'az and Putterman, Eli},
     TITLE = {Spectral monotonicity under {G}aussian convolution},
   JOURNAL = {Ann. Fac. Sci. Toulouse Math. (6)},
  FJOURNAL = {Annales de la Facult\'{e} des Sciences de Toulouse. Math\'{e}matiques.
              S\'{e}rie 6},
    VOLUME = {32},
      YEAR = {2023},
    NUMBER = {5},
     PAGES = {939--967},
}

@article {MS24,
    AUTHOR = {Mikulincer, Dan and Shenfeld, Yair},
     TITLE = {The {B}rownian transport map},
   JOURNAL = {Probab. Theory Related Fields},
  FJOURNAL = {Probability Theory and Related Fields},
    VOLUME = {190},
      YEAR = {2024},
    NUMBER = {1-2},
     PAGES = {379--444},
}

@inproceedings{AKV24,
  title={Trickle-Down in Localization Schemes and Applications},
  author={Anari, Nima and Koehler, Frederic and Vuong, Thuy-Duong},
  booktitle={Proceedings of the Annual ACM Symposium on Theory of Computing (STOC)},
  pages={1094--1105},
  year={2024}
}

@article {LLP10,
    AUTHOR = {Levin, David A. and Luczak, Malwina J. and Peres, Yuval},
     TITLE = {{G}lauber dynamics for the mean-field {I}sing model: cut-off,
              critical power law, and metastability},
   JOURNAL = {Probab. Theory Related Fields},
  FJOURNAL = {Probability Theory and Related Fields},
    VOLUME = {146},
      YEAR = {2010},
    NUMBER = {1-2},
     PAGES = {223--265},
}

@article {JS93,
    AUTHOR = {Jerrum, Mark and Sinclair, Alistair},
     TITLE = {Polynomial-time approximation algorithms for the {I}sing
              model},
   JOURNAL = {SIAM J. Comput.},
  FJOURNAL = {SIAM Journal on Computing},
    VOLUME = {22},
      YEAR = {1993},
    NUMBER = {5},
     PAGES = {1087--1116},
}

@inproceedings{RW99,
  author       = {Dana Randall and
                  David Wilson},
  title        = {Sampling Spin Configurations of an {I}sing System},
  booktitle    = {Proceedings of the Annual {ACM-SIAM} Symposium on Discrete Algorithms (SODA)},
  pages        = {959--960},
  year         = {1999},
}

@article {stam1959fisher,
    AUTHOR = {Stam, A. J.},
     TITLE = {Some inequalities satisfied by the quantities of information
              of {F}isher and {S}hannon},
   JOURNAL = {Information and Control},
  FJOURNAL = {Information and Control},
    VOLUME = {2},
      YEAR = {1959},
     PAGES = {101--112},
      ISSN = {0019-9958},
   MRCLASS = {94.00},
  MRNUMBER = {109101},
}

@article{mitra2024fast,
	title={Fast Convergence of $\Phi$-Divergence Along the Unadjusted Langevin Algorithm and Proximal Sampler},
	author={Siddharth Mitra and Andre Wibisono},
	journal={arXiv preprint arXiv:2410.10699},
	year={2024}
}

\appendix

\section{Spectral/Entropic Stability from Spectral Independence}
\label{app:spe-ent-stab}

In this section, we formally prove \Cref{lem:variance-stability-SI,lem:EI-entropic-stability}, 
establishing spectral/entropic stability from spectral independence under the respective noising processes.
We refer the readers to a recent work \cite{CCCYZ25+} for more discussion on spectral/entropic stability with respect to the continuous-time down walk.

\subsection{\texorpdfstring{Spectral stability from SI under field dynamics (Proof of~\Cref{lem:variance-stability-SI})}{Spectral stability from SI under field dynamics (Proof of Lemma~\ref{lem:variance-stability-SI})}}\label{sec:append-variance}
\begin{proof}
We slightly abuse the derivative notation  $\frac{\dif}{\dif \theta}$ to denote the limit operator $\lim_{\theta \to \eta^{+}} \frac{1}{\theta - \eta}$. 
We first claim that for any $S \in \+X \subseteq 2^{[n]}$, the derivative of $\Var[Q_{\eta \to \theta}(S,\cdot)]{f_{\theta}}$ satisfies:
\begin{align}\label{eq:derivative-hardcore}
    \frac{\dif \Var[Q_{\eta \to \theta}(S,\cdot)]{f_\theta}}{\dif \theta} \bigg|_{\theta = \eta^+} =  \frac{1}{1-\eta}\cdot \sum_{v \in [n] \setminus S} \Pr[R \sim Q_{\eta \to 1}(S,\cdot)]{v \in R} \cdot \tp{f_{\eta}(S \cup \{v\}) - f_{\eta}(S)}^2.
\end{align}
We first examine the derivative of the probability $Q_{\eta \to \theta}(S,T)$ for $S,T \in \+X\subseteq 2^{[n]}$ ($S \subseteq T$):
\begin{align}\label{eq:hardcore-Q-formula}
Q_{\eta \to \theta}(S,T) = \frac{\Pr[]{X_{1-\eta} = S \text{ and } X_{1-\theta} = T}}{\Pr[]{X_{1-\eta} = S}} = \frac{\sum_{T \subseteq R} \mu(R) (1-\theta)^{\abs{R} - \abs{T}} (\theta-\eta)^{\abs{T}-\abs{S}}}{\sum_{S \subseteq R} \mu(R) (1-\eta)^{\abs{R}-\abs{S}}}.
\end{align}
The derivative of the numerator in~\eqref{eq:hardcore-Q-formula} satisfies
\begin{align*}
& \left.\frac{\dif}{\dif \theta} \tp{\sum_{T \subseteq R} \mu(R) (1-\theta)^{\abs{R}-\abs{T}}(\theta-\eta)^{\abs{T}-\abs{S}}  } \right|_{\theta = \eta^+} \\
={}&
\begin{cases}
\sum\limits_{S \subseteq R} -\tp{\abs{R}-\abs{S}} \mu(R) (1-\eta)^{\abs{R}-\abs{S}-1}, & \text{if } S = T,\\
\sum\limits_{T \subseteq R} \mu(R) (1-\eta)^{\abs{R}-\abs{S}-1}, & \text{if } \abs{T \setminus S} = 1,\\
0, & \text{otherwise.}
\end{cases}
\end{align*}
Thus, the derivative of $Q_{\eta \to \theta}(S,T)$ is
\begin{align*}
\frac{\dif Q_{\eta \to \theta}(S,T)}{\dif \theta} \Bigg|_{\theta = \eta^+} =
\begin{cases}
\frac{1}{1-\eta} \E[R \sim Q_{\eta \to 1}(S,\cdot)]{\abs{S}-\abs{R}} & \text{if } S=T,\\
\frac{1}{1-\eta} \Pr[R \sim Q_{\eta \to 1}(S,\cdot)]{T \subseteq R} & \text{if } \abs{T\setminus S}=1,\\
0&\text{otherwise}.
\end{cases}
\end{align*}
Consequently, the transition rule of the Markov kernel $Q_{\eta \to \eta + h}(\cdot,\cdot)$ can be expressed as:
\begin{align}\label{eq:field-kernel}
    Q_{\eta \to \eta + h}(S,T) =
    \begin{cases}        
    1-\frac{1}{1-\eta}\sum_{v \in [n] \setminus S}\Pr[R \sim Q(S,\cdot)]{v \in R} \cdot h + o(h) & \text{if }S = T,\\
        \frac{1}{1-\eta} \Pr[R \sim Q_{\eta \to 1}(S,\cdot)]{v \in R} \cdot h + o(h) & \text{if }\abs{T\setminus S}=1,\\
        o(h) & \text{otherwise.}
    \end{cases}
\end{align}
Similarly, there exists $m(S,\theta) \neq \infty$ such that the derivative of $f_\theta(S) = \E[Q_{\theta \to 1}(S,\cdot)]{f}$ satisfies
\begin{align*}
\frac{\dif f_\theta(S)}{\dif \theta} = \frac{\dif}{\dif \theta} \tp{\sum_{S \subseteq T} \frac{\mu(T)(1-\theta)^{\abs{T}-\abs{S}}}{\sum_{S \subseteq R} \mu(R) (1-\theta)^{\abs{R}-\abs{S}}} f(T)} = m(S,\theta).
\end{align*}
Therefore, $f_{\eta + h}(S)$ satisfies
\begin{align}\label{eq:field-expected}
    f_{\eta + h}(S) = f_\eta(S) + m(S,\eta)  \cdot h + o(h).
\end{align}
Thus, from~\eqref{eq:field-kernel},~\eqref{eq:field-expected}, and the fact that $\E[Q_{\eta \to \eta+h}(S,\cdot)]{f_{\eta + h}}=f_{\eta}(S)$, the variance $\Var[Q_{\eta \to \eta + h}(S,\cdot)]{f}$ is bounded as:
\begin{align*}
\nonumber\Var[Q_{\eta \to \eta + h}(S,\cdot)]{f_{\eta + h}} 
&= \sum_{S \subseteq T} Q_{\eta \to \eta+h}(S,T)\tp{f_{\eta+h}(T) - f_{\eta}(S)}^2\\
&= \frac{1}{1-\eta}\cdot \sum_{v \in [n] \setminus S} \Pr[R \sim Q_{\eta \to 1}(S,\cdot)]{v \in R} \cdot \tp{f_{\eta}(S \cup \{v\}) - f_{\eta}(S)}^2 \cdot h + o(h).
\end{align*}
This proves~\eqref{eq:derivative-hardcore}.

Assuming $C(\eta)$-spectral independence of the distribution $Q_{\eta \to 1}(S,\cdot)$, we have the following for all test functions $g = \frac{\widetilde{\nu}}{Q_{\eta \to 1}(S,\cdot)}$, where $\widetilde{\nu}$ is absolutely continuous with respect to $Q_{\eta \to 1}(S,\cdot)$:
\begin{align}\label{eq:SI-hardcore}
    \sum_{v \in [n] \setminus S} p_v \tp{\frac{q_v}{p_v}-1}^2 \le \sum_{v \in [n] \setminus S} p_v \tp{\frac{q_v}{p_v}-1}^2 + (1-p_v) \tp{\frac{1-q_v}{1-p_v}-1}^2 \le C \cdot \Var[Q_{\eta \to 1}(S,\cdot)]{g},
\end{align}
where $p_v = \Pr[R \sim Q_{\eta \to 1}(S,\cdot)]{v \in R}$ and $q_v = \Pr[R \sim \widetilde{\nu}]{v \in R}$. 
Assuming the test function $g = f/f_\eta(S)$ in~\eqref{eq:SI-hardcore}, the marginal probability $q_v$ is given by:
\begin{align}\label{eq:q-hardcore}
q_v = \sum_{S \cup \{v\} \subseteq R} Q_{\eta \to 1}(S,R) \cdot \frac{f(R)}{f_\eta(S)} \overset{(\star)}{=} \sum_{S \cup \{v\} \subseteq R} p_v \cdot Q_{\eta \to 1}(S \cup \{v\},R) \cdot \frac{f(R)}{f_\eta(S)} = p_v \cdot \frac{f_\eta(S \cup \{v\})}{f_\eta(S)},
\end{align}
where equality $(\star)$ follows the definition of $Q_{\eta \to 1}(S,\cdot)$. Specifically, $Q_{\eta \to 1}(S,\cdot)$ is the distribution that samples $R$ according to the law of $(1-\eta) * \mu$ conditioned on $S \subseteq R$.

Combining~\eqref{eq:derivative-hardcore},~\eqref{eq:SI-hardcore} and~\eqref{eq:q-hardcore}, we have
\begin{align*}
    \frac{\dif \Var[Q_{\eta \to \theta}(S,\cdot)]{f_\theta}}{\dif \theta} \bigg|_{\theta = \eta^+} \le \frac{C(\eta)}{1-\eta} \cdot \Var[Q_{\eta \to 1}(S,\cdot)]{f}.
\end{align*}
This establishes the spectral stability claimed by the lemma.
\end{proof}

\subsection{\texorpdfstring{Entropic stability from SI under proximal sampler (Proof of~\Cref{lem:EI-entropic-stability})}{Entropic stability from SI under proximal sampler (Proof of Lemma~\ref{lem:EI-entropic-stability})}}\label{sec:append-entropy}

We remark that de Bruijn’s
identity (see~\cite{stam1959fisher} and~\cite[Lemma 4]{mitra2024fast}) relates the derivative of KL divergence $\frac{\dif}{\dif t}D_{\mathrm{KL}}(\nu_t \parallel \mu_t)$ with the Fisher information $\mathsf{FI}(\nu_t \parallel \mu_t)$, and \Cref{lem:EI-entropic-stability} can be shown by calculating the Fisher information $\mathsf{FI}(\nu_t \parallel \mu_t)$ (to obtain \eqref{eq:derivative-Ising} below) and applying~\cite[Lemma 40]{CE22}. For completeness, we present the following self-contained proof for \Cref{lem:EI-entropic-stability}.

\begin{proof}[Proof of~\Cref{lem:EI-entropic-stability}]
Let $\mu=\mu_{J,\*h}$ be the Gibbs distribution of Ising model specified by interaction matrix $J = L^{\intercal} L \in \mathbb{R}^{n \times n}$ and external fields $\*h \in \mathbb{R}^n$. Fix $\eta \in [0,1)$ and $\*x \in \mathbb{R}^r$ be a feasible solution. Let $\widetilde{\mu} := Q_{\eta \to 1}(\*x,\cdot)$ and $f = \frac{\nu}{\widetilde{\mu}}$ be the test function, where $\nu$ is a distribution absolutely continuous with respect to $\widetilde{\mu}$, which is also absolutely continuous with respect to $\mu_{J,\*h}$.
We first claim that for any $\*x \in \mathbb{R}^n$, the derivative of $\Ent[Q_{\eta \to \theta}(\*x,\cdot)]{f_\theta}$ satisfies:
\begin{align}\label{eq:derivative-Ising}
    \frac{\dif \Ent[Q_{\eta \to \theta}(\*x,\cdot)]{f_\theta}}{\dif \theta}\bigg|_{\theta = \eta^+} = \frac{1}{2(1-\eta)^2}(\*b(\widetilde{\mu})-\*b(\nu))^{\intercal} J (\*b(\widetilde{\mu})-\*b(\nu)),
\end{align}
where $\*b(\mu):=\E[\*x \sim \mu]{\*x}$ denotes the mass center of distribution $\mu$.

We first obtain the law of the denoising process $Q_{\eta \to \eta + t}(\*x,\cdot)$ for all $0 < t < 1-\eta$. By the definition of denoising process, for any measurable set $A \subseteq \mathbb{R}^r$ we have
\begin{align*}
    Q_{\eta \to \eta+t}(\*x,A) &= \Pr[]{X_{1-\eta-t} \in A \mid X_{1-\eta} = \*x}\\
    (\text{law of total probability})\quad &= \E[\*z \sim \widetilde{\mu}]{\Pr[]{X_{1-\eta - t} \in A \mid X_{1-\eta} = \*x \text{ and } X_0 = \*z}}.
\end{align*}
By Bayes' formula, i.e. $\Pr[]{Y \mid X,Z} = \frac{\Pr[]{X \mid Y,Z} \Pr[]{Y \mid Z}}{\Pr[]{X \mid Z}}$, the density function $p_{\eta,\*x,\*z,t}$ for the conditional distribution $(X_{1-\eta-t} \mid X_{1-\eta} = \*x \text{ and } X_0 = \*z)$ satisfies
\begin{align*}
    p_{\eta,\*x,\*z,t}(\*y) \propto &\exp\tp{-\frac{1}{2(\eta + t)(1-\eta-t)} (\*y- (\eta + t)L \*z)^{\intercal}(\*y- (\eta + t)L \*z)} \\
    &\exp\tp{ - \frac{1}{2 \eta t (\eta+t)} ((\eta+t) \*x - \eta \*y)^{\intercal} ((\eta+t) \*x - \eta \*y)}\\
    &\exp\tp{\frac{1}{2\eta(1-\eta)}(\*x - \eta L \*z)^{\intercal} (\*x-\eta L \*z)}\\
    \propto& \exp\tp{-\frac{1-\eta}{2(1-\eta-t)t} \tp{\*y - \frac{t}{1-\eta} L \*z - \frac{1-\eta-t}{1-\eta} \*x}^{\intercal} \tp{\*y - \frac{t}{1-\eta} L \*z - \frac{1-\eta-t}{1-\eta} \*x}}.
\end{align*}
Thus, the distribution $(X_{1-\eta-t} \mid X_{1-\eta} = \*x \text{ and } X_0 = \*z)$ satisfies a normal distribution with mean $\frac{t}{1-\eta} L \*z + \frac{1-\eta-t}{1-\eta} \*x$ and covariance $\frac{(1-\eta-t)t}{1-\eta} I_r$.
Therefore, the entropy satisfies
\begin{align}\label{eq:entropy-f}
    \Ent[Q_{\eta \to \eta+t}(\*x,\cdot)]{f_{\eta + t}} = \E[\substack{\*z \sim \widetilde{\mu}\\ \*e \sim \+N(\*0,I)}]{f_{\eta + t}\tp{\*\gamma} \log f_{\eta + t}\tp{\*\gamma}},
\end{align}
where $\*\gamma := \*x + \frac{t}{1-\eta}\tp{L \*z - \*x} - \sqrt{\frac{t(1-\eta-t)}{1-\eta}} \cdot \*e$.
Similarly, for any fixed $\*x$ and $\*\gamma$, the distributions $Q_{\eta +t \to 1}(\*\gamma,\cdot)$ and $\widetilde{\mu}:= Q_{\eta \to 1}(\*x,\cdot)$ satisfies:
\begin{align*}
    \frac{Q_{\eta +t \to 1}(\*\gamma,\*y)}{Q_{\eta \to 1}(\*x,\*y)} &\propto \frac{\exp\tp{-\frac{1}{2(\eta+t)(1-\eta-t)} (\*\gamma-(\eta+t) L \*y)^{\intercal}(\*\gamma-(\eta+t) L \*y) }}{\exp\tp{-\frac{1}{2\eta(1-\eta)} (\*x-\eta L \*y)^{\intercal}(\*x-\eta L \*y) }} \\
    &\propto \exp\tp{-\frac{1}{2t(1-\eta-t)(1-\eta)} (\*\gamma^\star)^{\intercal} (\*\gamma^\star) },
\end{align*}
where $\*\gamma^\star = t L \*y + (1-\eta-t) \*x - (1-\eta) \*\gamma = t L (\*y-\*z) + \sqrt{t(1-\eta)(1-\eta-t)} \*e$. Thus, we have
\begin{align*}
f_{\eta+t}(\*\gamma) &= \frac{\sum_{\*y}{Q_{\eta+t \to 1}(\*\gamma,\*y)} f(y)}{\sum_{\*y} Q_{\eta+t \to 1}(\*\gamma,\*y)}
= \frac{\sum_{\*y} Q_{\eta \to 1}(\*x,\*y) \exp\tp{-\frac{1}{2t(1-\eta-t)(1-\eta)} (\*\gamma^\star)^{\intercal} (\*\gamma^\star) } f(y) }{\sum_{\*y} Q_{\eta \to 1}(\*x,\*y) \exp\tp{-\frac{1}{2t(1-\eta-t)(1-\eta)} (\*\gamma^\star)^{\intercal} (\*\gamma^\star) } }\\
&= \frac{\E[\*y \sim \nu]{\exp\tp{-\frac{1}{2t(1-\eta-t)(1-\eta)} (\*\gamma^\star)^{\intercal} (\*\gamma^\star) }}}{\E[\*y \sim \widetilde{\mu}]{\exp\tp{-\frac{1}{2t(1-\eta-t)(1-\eta)} (\*\gamma^\star)^{\intercal} (\*\gamma^\star) }}},
\end{align*}
Hence, the entropy of $f$ in~\eqref{eq:entropy-f} can be further expressed as
\begin{align}\label{eq:entropy-f-new}
\Ent[Q_{\eta \to \eta+t}(\*x,\cdot)]{f_{\eta + t}} = \E[\substack{\*z \sim \widetilde{\mu} \\\*e \sim \+N(0,I)}]{\frac{N_{\*z,\*e}\tp{g(t)}}{D_{\*z,\*e}\tp{g(t)}} \log \frac{N_{\*z,\*e}\tp{g(t)}}{D_{\*z,\*e}\tp{g(t)}}},
\end{align}
where $g(t) = \frac{t}{(1-\eta)(1-\eta-t)}$, and $N_{\*z,\*e}(t)$ is defined by
\begin{align*}
N_{\*z,\*e}(t) := \E[\*x \sim \nu]{\exp\tp{-\frac{1}{2} (\*x-\*z)^{\intercal} L^{\intercal} L (\*x-\*z) \cdot t - (\*x-\*z)^{\intercal} L^{\intercal} \*e \cdot \sqrt{t}}},
\end{align*}
and $D_{\*z,\*e}(t)$ is defined accordingly by taking expectation over $\mu$. By expanding $N_{\*z,\*e}$ around the origin, we have
\begin{align}\label{eq:expand-N}
    \nonumber N_{\*z,\*e}(t) &= 1- \E[\*x \sim \nu]{(\*x-\*z)^{\intercal} L^{\intercal} \*e} \cdot \sqrt{t} \\
    &- \frac{1}{2} \E[\*x \sim \nu]{ (\*x-\*z)^{\intercal} L^{\intercal} L (\*x-\*z) - (\*x - \*z)^{\intercal} L^{\intercal} \*e \*e^{\intercal} L (\*x-\*z)} \cdot t + o(t)
\end{align}
Similarly, We can also expand $D_{\*z,\*e}(t)$ around the origin to obtain an analogue of~\eqref{eq:expand-N}. Thus, 
\begin{align}\label{eq:expand-N-D}
    \nonumber\frac{N_{\*z,\*e}(t)}{D_{\*z,\*e}(t)} &= 1 + \tp{\E[\*x \sim \widetilde{\mu}]{(\*x-\*z)^{\intercal} L^{\intercal} \*e} - \E[\*x \sim \nu]{(\*x-\*z)^{\intercal} L^{\intercal} \*e} } \cdot \sqrt{t}\\
    \nonumber &-\frac{1}{2}\E[\*x \sim \nu]{ (\*x-\*z)^{\intercal} L^{\intercal} L (\*x-\*z) - (\*x - \*z)^{\intercal} L^{\intercal} \*e \*e^{\intercal} L (\*x-\*z)} \cdot t\\
    \nonumber &+\frac{1}{2}\E[\*x \sim \nu]{ (\*x-\*z)^{\intercal} L^{\intercal} L (\*x-\*z) - (\*x - \*z)^{\intercal} L^{\intercal} \*e \*e^{\intercal} L (\*x-\*z)} \cdot t\\
    \nonumber &-\frac{1}{2}\E[\*x \sim \widetilde{\mu}, \*y \sim \nu]{(\*x-\*z)^{\intercal} L^{\intercal} \*e \*e^{\intercal} L (\*y-\*z)} \cdot t\\
    &+\frac{1}{2}\E[\*x \sim \widetilde{\mu}, \*y \sim \mu]{(\*x-\*z)^{\intercal} L^{\intercal} \*e \*e^{\intercal} L (\*y - \*z)} \cdot t + o(t).
\end{align}
Let $A_{\*z,\*e}$ denote the coefficients of $\sqrt{t}$ in~\eqref{eq:expand-N-D}, and $B_{\*z,\*e}$ denote the coefficients of $t$ in~\eqref{eq:expand-N-D}. Then,
\begin{align}\label{eq:derivative-f-Ising}
   \frac{N_{\*z,\*e}(t)}{D_{\*z,\*e}(t)} \log \frac{N_{\*z,\*e}(t)}{D_{\*z,\*e}(t)} =  A \sqrt{t} + B t + \frac{1}{2} A^2 t + o(t).
\end{align}
By a straightforward calculation\footnote{We remark that the $o(t)$ in~\eqref{eq:derivative-f-Ising} relies on $\*e$ (denoted by $o_{\*e}(t)$). Though, the remainder term can be bounded by $t$ times $R(\*e) =\exp\tp{C\cdot \norm{\*e}_1} \cdot \mathrm{poly}(\norm{\*e}_1)$, where $C$ is some universal constant. The expectation $\E[]{R(\*e)}$ is finite, implying that the expectation of $o_{\*e}(t)$ is indeed $o(t)$.}, we have
\begin{align*}
\E[\substack{\*z \sim \widetilde{\mu}\\ \*e \sim \+N(0,I)}]{A_{\*z,\*e}} = \E[\substack{\*z \sim \widetilde{\mu}\\ \*e \sim \+N(0,I)}]{B_{\*z,\*e}} = 0 \quad \text{and} \quad \E[\substack{\*z \sim \widetilde{\mu}\\ \*e \sim \+N(0,I)}]{A^2_{\*z,\*e}} = (\*b(\widetilde{\mu})-\*b(\nu))^{\intercal} J (\*b(\widetilde{\mu})-\*b(\nu)),
\end{align*}
where $\*b(\mu) := \E[\*x \sim \mu]{\*x}$ denotes the center of $\mu$ and $J = L^{\intercal} L$. Together with~\eqref{eq:entropy-f-new}, we have
\begin{align*}
\frac{\dif \Ent[Q_{t \to 1}(x,\cdot)]{f_t}}{\dif t} \bigg|_{t = \eta^+} = \frac{1}{2(1-\eta)^2}(\*b(\widetilde{\mu})-\*b(\nu))^\intercal J (\*b(\widetilde{\mu})-\*b(\nu)).
\end{align*}
This proves~\eqref{eq:derivative-Ising}.

Recall that $\widetilde{\mu}:=Q_{\eta \to 1}(\*x,\cdot)$ is $C\tp{\eta}$-spectrally independent under all external fields, i.e.
\begin{align*}
    \forall \*\lambda \in \mathbb{R}_{\ge 0}^n,\quad \mathrm{Cov}(\*\lambda * \widetilde{\mu}) \preceq C(\eta)\cdot I.
\end{align*}
By~\cite[Lemma 40]{CE22}, we have
\begin{align}\label{eq:CE-stability}
\frac{1}{2} (\*b(\nu) - \*b(\widetilde{\mu})) J (\*b(\nu) - \*b(\widetilde{\mu})) \le \alpha \Ent[\widetilde{\mu}]{f},
\end{align} where $\alpha = C(\eta) \norm{J}_{2}$.~\Cref{lem:EI-entropic-stability} then follows from~\eqref{eq:derivative-Ising} and~\eqref{eq:CE-stability}.
\end{proof}

\section{\texorpdfstring{Missing Proofs from  \Cref{sec:LB}}{Missing Proofs from  Section \ref{sec:LB}}}
\label{app:missing-LB}
In this section, we present the formal proofs of  \Cref{lem:distant-alpha-hardcore} and \Cref{lem:distant-alpha-Ising}.
We start with the proof of \Cref{lem:distant-alpha-Ising} for the Ising model in \Cref{sec:append-Ising}, as it follows standard methods and is comparatively straightforward. 
We then address \Cref{lem:distant-alpha-hardcore} for the hardcore model in \Cref{sec:append-hardcore},  where the proof is more intricate and technically challenging due to the need for symmetrization.

\subsection{\texorpdfstring{Proof of \Cref{lem:distant-alpha-Ising}}{Proof of Lemma~\ref{lem:distant-alpha-Ising}}}\label{sec:append-Ising}
As in the proof of~\cite{GSV16}, we approximate $\log \widetilde{\alpha}_{s,t}$ via an optimization problem. For convenience, let $H(x) = - x \log x$ be the entropy function.
\begin{lemma}[\cite{GSV16}]\label{lem:approx-Ising}
    For any integers $0 \le s,t \le n$, the logarithm of $\widetilde{\alpha}_{s,t}$ satisfies
    \begin{align*}
    \abs{\log \widetilde{\alpha}_{s,t} - n\underbrace{\tp{\Delta \max_{\*\theta} f_N(\*\theta) - (\Delta-1) f_D(\*\rho)}}_{=: U(\*\rho)}} = O_{\Delta}\tp{\log n},
    \end{align*}
    where $\*\rho = \tp{\frac{s}{n},\frac{t}{n}}$, $f_N(\*\theta) = \sum_{i=1}^4 H(\theta_i) + 2\beta (\theta_1+\theta_4)$, $f_D(\*\rho) = \sum_{i=1}^2 H(\rho_i)+H(1-\rho_i)$. Furthermore, the maximum is taken over all $\*\theta \in [0,1]^4$ satisfying
    \begin{align*}
            \theta_1 + \theta_2 + \theta_3 + \theta_4 = 1, \quad
            \theta_2 + \theta_4 = \rho_1, \quad
            \theta_3 + \theta_4 = \rho_2.
    \end{align*}
\end{lemma}

\begin{proof}
    The coefficient $[x^sy^t]N$ of the generating function $N$ in~\eqref{eq:num-Ising} can be bounded by
    \begin{align*}
    \max_{\*\zeta} \binom{n}{\*\zeta} \exp\tp{2\beta \cdot \tp{\zeta_1+\zeta_4}} \le [x^s y^t] N \le n^4 \max_{\*\zeta} \binom{n}{\*\zeta} \exp\tp{2\beta \cdot (\zeta_1+\zeta_4)},
    \end{align*}
    where $\binom{n}{\*\zeta} = \frac{n!}{\prod_{i=1}^4 \zeta_i!}$ is the multinomial coefficient, and the maximum is taken over all $\*\zeta \in \mathbb{Z}^4$ with
    \begin{align}\label{eq:req-zeta}
        \zeta_1 + \zeta_2 + \zeta_3 + \zeta_4 = n, \quad
        \zeta_2 + \zeta_4 = s, \quad
        \zeta_3 + \zeta_4 = t.
    \end{align}
    By the definition of $\widetilde{\alpha}_{s,t}$ in~\eqref{eq:alpha-express-Ising}, the upper bound for $\log \widetilde{\alpha}_{s,t}$ follows immediately from Stirling's approximation and that $\*\theta = \*\zeta/n$. For the lower bound, suppose $\max_{\*\theta} f_N(\theta)$ achieves the maximum at $\*\theta = \*\theta_{\*\rho}$ for $\*\rho=\tp{\frac{s}{n},\frac{t}{n}}$. We may round $n \*\theta_{\*\rho}$ to the nearest integer point $\*\zeta$ satisfying~\eqref{eq:req-zeta}. Hence, the difference between $f_N(\*\theta_{\*\rho})$ and $ f_N\tp{\*\zeta/n}$ can be bounded by $O(\log n/n)$, and the lower bound for $\log \widetilde{\alpha}_{s,t}$ can be established by Stirling's approximation.
\end{proof}

\begin{lemma}[\cite{GSV16}]\label{lem:critical-Ising}
    $U(\*\rho)$ (see \Cref{lem:approx-Ising}) achieves maximum value at its unique critical point $\*\rho^\star = \tp{\frac{1}{2},\frac{1}{2}}$. Furthermore, the function $U(\rho,\rho)$ has a unique critical point at $\rho = \frac{1}{2}$.
\end{lemma}

\begin{proof}[Proof of~\Cref{lem:critical-Ising}]
The former part was proved in~\cite[Lemma 11, Lemma 12]{GSV16}. We now aim to prove the latter. By definition, $U(\rho,\rho)$ corresponds to the following optimization program:

\begin{equation}\label{eq:optimization-Ising}
\begin{aligned}
\max_{\*\theta} \quad & \Delta f_N(\*\theta) - (\Delta-1) f_D(\*\rho)\\
\textrm{s.t.} \quad & 
\theta_1+\theta_2+\theta_3+\theta_4 = 1\\
&\theta_2 + \theta_4 = \rho\\
&\theta_3 + \theta_4 = \rho
\end{aligned}
\end{equation}
Applying the method of Lagrange multiplier, there exist functions $P(\rho),Q(\rho) \ge 0$ satisfying
\begin{align*}
    \theta_1 = \exp\tp{2\beta} P,\quad \theta_2=\theta_3=PQ, \quad \theta_4=\exp\tp{2\beta} PQ^2, \quad \text{ where } Q^\Delta = \tp{\frac{\rho}{1-\rho}}^{\Delta-1}.
\end{align*}
Then, combining~\eqref{eq:optimization-Ising}, we have the following equations 
\begin{align*}
  Q^\Delta = \tp{\frac{\rho}{1-\rho}}^{\Delta-1} = \tp{\frac{\theta_2 + \theta_4}{\theta_1 + \theta_3}}^{\Delta-1} = \tp{\frac{Q + \exp\tp{2\beta} Q^2}{\exp\tp{2\beta} + Q}}^{\Delta-1} = Q^{\Delta-1} \tp{\frac{1 + \exp\tp{2\beta} Q}{\exp\tp{2\beta} + Q}}^{\Delta-1}.
\end{align*}
In particular, we have
\begin{align*}
  Q = \tp{\frac{1 + \exp\tp{2\beta} Q}{ \exp\tp{2\beta} + Q}}^{\Delta-1}
  \quad \text{and} \quad
  \frac{\rho}{1-\rho} = Q^{\frac{\Delta}{\Delta-1}}.
\end{align*}
Thus, $Q$ should be the fixed point of critical Ising model, implying $Q=1$ and $\rho = \frac{1}{2}$.
\end{proof}
\begin{proof}[Proof of~\Cref{lem:distant-alpha-Ising}]
By~\Cref{lem:approx-Ising}, it suffices to show that there exist $\epsilon,\eta > 0$ satisfying
\begin{align*}
    U(\*\rho) \le -\epsilon n^{-1/2} + U\tp{\*\rho^\star},
\end{align*}
where $\*\rho$ satisfies $\abs{\rho_1+\rho_2-1}>n^{-1/4}$ and $\abs{\rho_1-\rho_2} \le \eta n^{-1/4}$.
By~\Cref{lem:critical-Ising}, we may assume without loss of generality\footnote{ \Cref{lem:critical-Ising} indicates that the maximum value only achieves at the point $(1/2,1/2)$, implying the values $U(\rho_1,\rho_2)$ outside the $(0.25,0.75)^2$ box is less than $U(1/2,1/2)-\epsilon$ for some constant $\epsilon > 0$.} that $0.25 < \rho_1,\rho_2 < 0.75$. Note that for any $0.25 < x_1,x_2 < 0.75$,
\begin{align*}
    0 \le 2H\tp{\frac{x_1+x_2}{2}}-H(x_1)-H(x_2) \le \min_{0.25 \le x \le 0.75} \abs{H''(x)} \cdot \abs{x_1-x_2}^2 \le 4 \abs{x_1-x_2}^2.
\end{align*}
Therefore, $U(\rho_1,\rho_2) \le U(\rho,\rho) + 8(\Delta-1)\eta n^{-1/2}$, where $\rho$ is the average of $\rho_1$ and $\rho_2$.
Since $\abs{\rho_1+\rho_2-1}>n^{-1/4}$, we have $\abs{2\rho - 1} > n^{-1/4}$.
Let $\rho_{l/r} = \frac{1}{2}(1 \pm n^{-1/4})$, by~\Cref{lem:critical-Ising}, we have
\begin{align*}
    U(\rho,\rho) \le \max \set{U\tp{\rho_l,\rho_l},U\tp{\rho_r,\rho_r}} = U\tp{\frac{1}{2},\frac{1}{2}} - Q(\Delta) n^{-1/2} + O_{\Delta}\tp{\frac{\log n}{n}},
\end{align*}
and the last inequality follows from~\Cref{lem:approx-Ising} and~\Cref{lem:center-alpha-Ising} ($Q$ is the constant in \Cref{lem:center-alpha-Ising}).
Hence, we can pick $\eta = \frac{Q(\Delta)/4}{8(\Delta-1)}$ and get
\begin{align*}
U(\rho_1,\rho_2) \le U\tp{\frac{1}{2},\frac{1}{2}} - \frac{Q(\Delta)}{2} n^{-1/2}
\end{align*}
when $n$ is sufficiently large. This concludes the proof.
\end{proof}

\subsection{\texorpdfstring{Proof of~\Cref{lem:distant-alpha-hardcore}}{Proof of Lemma~\ref{lem:distant-alpha-hardcore}}}\label{sec:append-hardcore}

As in the previous proof, we approximate $\log \alpha_{A,B,C}$ via an optimization problem. For convenience, let $H(x) = - x \log x$ be the entropy function.
\begin{lemma}\label{lem:approx-hardcore}
  Let $\*p = \tp{1,2\lambda^{1/\Delta},2\lambda^{1/\Delta},2\lambda^{2/\Delta},\lambda^{2/\Delta},\lambda^{2/\Delta}}$.
  Then, we define $f_N(\*\theta)$ and $f_D(\*\rho)$ as
    \begin{align*}
        f_N(\*\theta) = \sum_{i=1}^6 H(\theta_i) + \theta_i \log p_i \quad \text{and} \quad f_D(\*\rho) = \sum_{i=1}^4 H(\rho_i)
    \end{align*}
    where $H(x) = - x \log x$ is the entropy function.

    For integers $0 \le A,B,C \le 2n$ with $A+B+2C \le 2n$, the logarithm of $\alpha_{A,B,C}$ satisfies
    \begin{align*}
    \abs{\log \alpha_{A,B,C} - n\underbrace{\tp{\Delta \max_{\*\theta} f_N(\*\theta) - 2(\Delta-1) f_D(\*\rho)}}_{=: U(\*\rho)}} = O_{\Delta}\tp{\log n},
    \end{align*}
    where $\*\rho = \tp{1-\frac{A+B+C}{2n}, \frac{A}{2n}, \frac{B}{2n}, \frac{C}{2n}}$, and the maximum is taken over all $\*\theta \in [0,1]^6$ satisfying
    \begin{align}\label{eq:requirement-theta-hardcore}
            \sum_{i=1}^6 \theta_i = 1, \quad
            \theta_2 + 2 \theta_5 = 2\rho_2, \quad
            \theta_3 + 2 \theta_6 = 2\rho_3, \quad 
            \theta_4 = 2\rho_4.
    \end{align}
\end{lemma}

\begin{remark}
    We remark that $U(\*\rho)$ is defined on 
    \begin{align*}
    \Omega = \{\*\rho \in [0,1]^4 \mid \rho_1+\rho_2+\rho_3+\rho_4 = 1 \text{ and } \rho_2+\rho_3+2\rho_4 \le 1\}.
    \end{align*}
    The latter constraint is to guarantee the existence of $\*\theta$ satisfying~\eqref{eq:requirement-theta-hardcore}.
\end{remark}
\begin{proof}[Proof of~\Cref{lem:approx-hardcore}]
  By formula~\eqref{eq:hardcore-alpha-coefficients} of $\alpha_{A,B,C}$, it holds that
  \begin{align*}
    \alpha_{A, B, C} = \frac{([x^Ay^Bz^C] \hat{N})^\Delta}{([x^Ay^Bz^C]D)^{\Delta-1}},
  \end{align*}
  where $\hat{N}(x,y,z) = N(\lambda^{1/\Delta}x, \lambda^{1/\Delta}y, \lambda^{2/\Delta}z)$.
  According to the definition of $N$ in~\eqref{eq:num-hardcore}, we have
  \begin{align*}
    \hat{N} = (p_1 + p_2 x + p_3 y + p_4 z + p_5 x^2 + p_6 y^2)^n.
  \end{align*}
  Hence, the coefficient of the generation function $\hat{N}$ can be bounded by
    \begin{align*}
        \max_{\*\zeta} \binom{n}{\*\zeta} \prod_{i=1}^6 p_i^{\zeta_i} \le [x^Ay^Bz^C]\hat{N} \le n^6 \max_{\*\zeta} \binom{n}{\*\zeta} \prod_{i=1}^6 p_i^{\zeta_i},
    \end{align*}
    where the maximum is taken over all parameters $\*\zeta \in \{0,1,\ldots,n\}^{6}$ satisfying
    \begin{align}\label{eq:requirement-zeta-hardcore}
        \sum_{i=1}^6 \zeta_i = n, \quad \zeta_2  + 2 \zeta_5 = A, \quad \zeta_3 + 2\zeta_6 = B,  \quad \zeta_4 = C.
    \end{align}
    The upper bound of $\log \alpha_{A,B,C}$ follows from the Stirling's approximation, where $\*\theta = \*\zeta/n$.
    The function $\max_{\*\theta} f_N(\*\theta)$ is the approximation for the coefficient $[x^Ay^Bz^C]\hat{N}$; and $f_D(\*\rho)$ is the approximation for the coefficient $[x^Ay^Bz^C]D$.
    For the lower bound, suppose $\max_{\*\theta} f_N(\*\theta)$ achieves the maximum value at $\*\theta = \*\theta_{\*\rho}$ for the given $A,B,C$. We may round $n\*\theta$ to the nearest integer point $\*\zeta$ satisfying~\eqref{eq:requirement-zeta-hardcore}. Hence, the difference between $f_N(\*\theta_{\*\rho})$ and $f_N(\*\zeta/n)$ can be bounded by $O(\log n/n)$, and the lower bound can be established together with Stirling's approximation.
\end{proof}

\begin{lemma}\label{lem:critical-hardcore}
Let $U^{\mathrm{sym}}(\rho_2,\rho_4) = U(1-2\rho_2-\rho_4,\rho_2,\rho_2,\rho_4)$ be the symmetric version of $U(\*\rho)$ defined on $\Omega = \{(\rho_2,\rho_4) \in [0,1]^2 \mid 2\rho_2+2\rho_4 \le 1 \}$. Then $U^{\mathrm{sym}}(\rho_2,\rho_4)$ achieves maximum value at its unique critical point at $(\rho_2,\rho_4) = \tp{\frac{\Delta-1}{\Delta^2},\frac{1}{\Delta^2}}$. 
\end{lemma}

\begin{proof}[Proof of~\Cref{lem:critical-hardcore}]
We first show that $U^{\mathrm{sym}}(\rho_2,\rho_4)$ has a unique critical point at $(\rho_2,\rho_4) = \tp{\frac{\Delta-1}{\Delta^2},\frac{1}{\Delta^2}}$.
By definition, $U^{\mathrm{sym}}(\rho_2,\rho_4)$ corresponds to the following optimization program:

\begin{equation}\label{eq:optimization-hardcore}
\begin{aligned}
\max_{\*\theta} \quad & \Delta f_N(\*\theta) - 2(\Delta-1) f_D(\*\rho)\\
\textrm{s.t.} \quad & 
\theta_1+\theta_2+\theta_3+\theta_4 + \theta_5 + \theta_6 = 1\\
&\theta_2 + 2\theta_5 = 2\rho_2\\
&\theta_3 + 2\theta_6 = 2\rho_2\\
&\theta_4 = 2\rho_4.
\end{aligned}
\end{equation}
By the method of Lagrange multiplier, there exist $\*\theta = \*\theta^\star$ and constants $P,Q_1,Q_2,R \ge 0$, such that they satisfy \eqref{eq:optimization-hardcore} and the following equations
\begin{align}\label{eq:pre-lagrange-hardcore}
\theta_1 = p_1 P, \quad \theta_2 &= p_2 PQ_1, \quad \theta_3 =  p_2 PQ_2, \quad \theta_4 = p_4 PR, \quad \theta_5 = p_5PQ_1^2, \quad \theta_6 =p_5 PQ_2^2.
\end{align}
Moreover, $f_N(\*\theta^\star) = \max_{\*\theta}f_N(\*\theta)$ (we use the fact that $p_2 = p_3$ and $p_5 = p_6$).

We first show that $\theta^\star$, $P, Q_1, Q_2, R$ are uniquely determined by \eqref{eq:optimization-hardcore} and \eqref{eq:pre-lagrange-hardcore}, so that they can be seen as implicit functions of $\rho_2$ and $\rho_4$.
By~\eqref{eq:optimization-hardcore}, we have $\theta_2 + 2\theta_5 = \theta_3 + 2\theta_6$.
Let $x = \lambda^{1/\Delta}$ so that $p_2 = 2x$ and $p_5 = x^2$.
We have
\begin{align*}
  P(2xQ_1 + (xQ_1)^2) = P(2xQ_2 + (xQ_2)^2),
\end{align*}
which implies $P(xQ_1 + 1)^2 = P(xQ_2 + 1)^2$.
Since, $Q_1, Q_2 \geq 0$, we can conclude that $Q_1 = Q_2 =: Q$.
Hence \eqref{eq:pre-lagrange-hardcore} can be simplifies to
\begin{align}\label{eq:lagrange-hardcore}
\theta_1 = p_1 P, \quad \theta_2 = \theta_3 =  p_2 PQ, \quad \theta_4 = p_4 PR, \quad \theta_5 = \theta_6=p_5 PQ^2.
\end{align}
By \eqref{eq:optimization-hardcore} and \eqref{eq:lagrange-hardcore}, we have
\begin{align*}
  \frac{\rho_4}{\rho_2} &= \frac{\theta_4}{\theta_2 + 2\theta_5}
  = \frac{p_4 R}{p_2Q + p_5Q^2} \\
  \frac{1}{\rho_2} &= \frac{\sum_{i=1}^6 \theta_i}{\theta_2 + 2\theta_5}
  = \frac{p_1 + 2p_2Q + p_4R + 2p_5 Q^2}{p_2Q + p_5Q^2} \\
  &= 1 + \frac{p_1}{p_2 Q + p_5 Q^2} + \frac{\rho_4}{\rho_2},
\end{align*}
where, recall $x = \lambda^{1/\Delta}$ so that $p_2 = 2x$ and $p_5 = x^2$, the later equation implies that
\begin{align*}
  \frac{p_1\rho_2 + 1 - \rho_2 - \rho_4}{1 - \rho_2 - \rho_4} = (1 + xQ)^2.
\end{align*}
Hence $Q$ and $R$ are uniquely determined by $\rho_2$ and $\rho_4$.
By $\sum_{i=1}^6 \theta_i = 1$ and \eqref{eq:lagrange-hardcore}, $P$ is also uniquely determined by $\rho_2$ and $\rho_4$.
Again, by  \eqref{eq:lagrange-hardcore}, $\*\theta^\star$ are uniquely determined by $\rho_2$ and $\rho_4$.

Now, we can conclude that $\*\theta^\star$, $P, Q, R$ are implicit functions of $\rho_2$ and $\rho_4$ determined by \eqref{eq:optimization-hardcore} and \eqref{eq:lagrange-hardcore}.
Combining \eqref{eq:optimization-hardcore}, \eqref{eq:lagrange-hardcore}, and the fact that $p_2 = p_3$, $p_5 = p_6$, we have
\begin{align*}
  \max_{\*\theta} f_N(\*\theta) = f_N(\*\theta^\star)
  &= \sum_{i=1}^6 -\theta^\star_i \log \frac{\theta^\star_i}{p_i} 
  = - \tp{\sum_{i=1}^6 \theta^\star_i \log P + 2(\theta^\star_2 + 2\theta^\star_5)\log Q + \theta^\star_4\log R} \\
  &= - (\log P + 4\rho_2\log Q + 2\rho_4 \log R).
\end{align*}
This also implies
\begin{align*}
  U^{\mathrm{sym}}(\rho_2,\rho_4) &= -\Delta\tp{\log P + 4\rho_2 \log Q + 2 \rho_4 \log R} - 2(\Delta-1) f_D(\*\rho), \\
  \text{where} \quad \*\rho &= \tp{1-2\rho_2-\rho_4,\rho_2,\rho_2,\rho_4}.
\end{align*}
Therefore, by taking derivative on both sides of the first constraint in~\eqref{eq:optimization-hardcore}, we have
\begin{align}
  \nonumber
0 = \sum_{i=1}^{6} \frac{\partial \theta^\star_i}{\partial \rho_2} &= \tp{\sum_{i=1}^6\theta_i^\star} \cdot \frac{1}{P} \cdot \frac{\partial P}{\partial \rho_2} + \tp{\theta_2+\theta_3+2\theta_5+2\theta_6}\cdot \frac{1}{Q} \cdot \frac{\partial Q}{\partial \rho_2} + \cdot \frac{\theta_4}{R}\cdot \frac{\partial R}{\partial \rho_2}\\
\label{eq:rel}
&=\frac{1}{P} \cdot \frac{\partial P}{\partial \rho_2} + \frac{4\rho_2}{Q} \cdot \frac{\partial Q}{\partial \rho_2} + \frac{2\rho_4}{R} \cdot \frac{\partial R}{\partial \rho_2}, 
\end{align}
where the last equation follows from the constraints in~\eqref{eq:optimization-hardcore}. Similarly, a same type of equation holds for $\rho_4$ as well. Finally, combining~\eqref{eq:rel}, the critical point $(\rho_2,\rho_4)$ of $U^{\mathrm{sym}}$ satisfies
\begin{align*}
\frac{\partial U^{\mathrm{sym}}}{\partial \rho_2} = -\tp{4 \Delta \log Q+2(\Delta-1)\frac{\partial f_D(\*\rho)}{\partial \rho_2}} &= 0\\
\frac{\partial U^{\mathrm{sym}}}{\partial \rho_4} = -\tp{2 \Delta \log R+2(\Delta-1)\frac{\partial f_D(\*\rho)}{\partial \rho_4}} &= 0
\end{align*}
By a straightforward calculation, we have
\begin{align*}
    \log Q &= \frac{\Delta-1}{\Delta} \log \frac{\rho_2}{1-2\rho_2-\rho_4} = \frac{\Delta-1}{\Delta}\log \frac{\lambda^{1/\Delta}Q+\lambda^{2/\Delta} Q^2}{1+2\lambda^{1/\Delta}Q+\lambda^{2/\Delta}R}\\
    \log R &= \frac{\Delta-1}{\Delta} \log \frac{\rho_4}{1-2\rho_2-\rho_4} = \frac{\Delta-1}{\Delta} \log \frac{\lambda^{2/\Delta}R}{1+2\lambda^{1/\Delta}Q+\lambda^{2/\Delta}R},
\end{align*}
where we use \eqref{eq:lagrange-hardcore} the fact that \eqref{eq:optimization-hardcore} implies $2(1 - 2\rho_2 - \rho_4) = 2\theta_1 + \theta_2 + \theta_3 + \theta_4$.
Let $y_2 = \lambda^{1/\Delta} Q$ and $y_4 = \lambda^{2/\Delta} R$. It is equivalent to the following system:
    \begin{align} \label{eq:system-y}
       \begin{cases}
            y_2 
            = \lambda \tp{\frac{ 
            1 + y_2
            }{
            1 + 2 y_2 + y_4 
            }}^{\Delta-1} &\hspace{2cm} (a)\\
            y_4
            = \lambda^2 \tp{\frac{
            1
            }{
            1 + 2 y_2 + y_4 
            }}^{\Delta-1} &\hspace{2cm} (b)
        \end{cases}
    \end{align}
Let $d=\Delta-1$.
It is direct to verify that $(y_2,y_4) = (\hat{x}_d, \hat{x}_d^2)$ is a solution of the above system, where $\hat{x}_d$ is the unique fixed point of the tree recursion: $\hat{x}_d = \lambda(\frac{1}{1+\hat{x}_d})^{d}$.
Since $\lambda = \lambda_c(\Delta) = \frac{d^d}{(d-1)^{d+1}}$, we have $(y_2, y_4) = (\frac{1}{d-1}, \frac{1}{(d-1)^2})$.
By \eqref{eq:lagrange-hardcore} and \eqref{eq:optimization-hardcore}, one can verify that the point $(\rho_2, \rho_4) = (\frac{\Delta-1}{\Delta^2}, \frac{1}{\Delta^2})$ is a critical point.

In this paragraph, we verify that the system in \eqref{eq:system-y} has at most $1$ solution.
 We resolve $y_4$ via (b) in~\eqref{eq:system-y} and plug it into (a) in~\eqref{eq:system-y} to obtain
\begin{align*}
    (y_2+1) \left(\left(\frac{\lambda }{y_2}\right)^{1/d}-1\right)=y_2 \left(\lambda  (y_2+1)^{-d}+1\right).
\end{align*}
By taking logarithm on both sides and rearranging terms, it holds that 
\begin{align*}
    F(y_2) := \log\tp{(y_2+1) \left(\left(\frac{\lambda }{y_2}\right)^{1/d}-1\right)} - \log\tp{y_2 \left(\lambda  (y_2+1)^{-d}+1\right)} = 0.
\end{align*}
To prove the uniqueness of $y$, it suffices to show that $F'(y) \le 0$ for all $y \in [0,\lambda]$. Note that
\begin{align*}
    F'(y) &= -\frac{\tp{\frac{\lambda}{y}}^{1/d}}{d y\tp{\left(\frac{\lambda }{y}\right){}^{1/d} - 1}}+\frac{d \lambda }{\left(y+1\right) \left(\left(y+1\right){}^d+\lambda \right)}-\frac{1}{y^2+y} \\
    &\leq -\frac{1}{d y}+\frac{d \lambda }{\left(y+1\right) \left(\left(y+1\right){}^d+\lambda \right)}-\frac{1}{y^2+y} \\
    &= \frac{(d+1) \lambda  \left((d-1) y-1\right)-\left(y+1\right){}^d \left(d+y+1\right)}{d y \left(y+1\right) \left(\left(y+1\right){}^d+\lambda \right)}.
\end{align*}
When $(d - 1)y - 1 \leq 0$, we have $F'(y) \leq 0$.
Otherwise, to make $F'(y) \leq 0$, we need
\begin{align}\label{eq:lambda-bound}
    \lambda \leq \frac{\left(y+1\right){}^d \left(d+y+1\right)}{(d+1)  \left((d-1) y-1\right)}.
\end{align}
It can be verified easily that~\eqref{eq:lambda-bound} holds when $d=2$. We assume without loss of generality that $d \ge 3$.
Since $y \in [\frac{1}{d-1}, \lambda]$, we can simplify the r.h.s. as follow
\begin{align*}
    \frac{\left(y+1\right){}^d \left(d+y+1\right)}{(d+1)  \left((d-1) y-1\right)} \geq \lambda_c(\Delta) \frac{d-1}{d+1} \frac{d + \lambda + 1}{(d-1)\lambda - 1}.
\end{align*}
The r.h.s. is greater than $\lambda$ for $\lambda = \lambda_c(\Delta)$ and $d \geq 3$. Hence, we have shown the uniqueness of critical point.

Finally, we prove that $U^{\mathrm{sym}}$ cannot achieves the maximum value on the boundary. 
Recall $y_2 = \lambda^{1/\Delta} Q$ and $y_4=\lambda^{2/\Delta} R$.
By~\eqref{eq:optimization-hardcore} and~\eqref{eq:lagrange-hardcore}, it holds that
\begin{align} \label{eq:y-bound}
\frac{y_2 + 2y_2^2}{1+y_2} = \frac{2\rho_2}{1-2\rho_4-2\rho_2} \quad \text{ and } \quad \frac{y_4}{y_2+y_2^2} = \frac{\rho_4}{\rho_2}.
\end{align}
Recall that
\begin{align}
    \nonumber
  \frac{\partial U^{\mathrm{sym}}}{\partial \rho_2}
  &= -4\log \tp{\frac{\rho_2}{1-2\rho_2-\rho_4} \cdot \tp{\frac{1-2\rho_2-\rho_4}{\rho_2} \cdot Q}^{\Delta}} \\
    \label{eq:U-rho-2}
    &= -4\log \tp{\frac{\rho_2}{\lambda\tp{1-2\rho_2-\rho_4}} \cdot \tp{\frac{1+2y_2 + y_4}{1+ y_2}}^{\Delta}}\\
  \nonumber
  \frac{\partial U^{\mathrm{sym}}}{\partial \rho_4}
    &= -2\log \tp{\frac{\rho_4}{1-2\rho_2-\rho_4}\cdot \tp{\frac{1-2\rho_2-\rho_4}{\rho_4} \cdot R}^\Delta }\\
  \label{eq:U-rho-4}
   &= -2\log \tp{\frac{\rho_4}{\lambda^2\tp{1-2\rho_2-\rho_4}} \cdot \tp{1+2y_2+y_4}^\Delta }.
\end{align}
We are now ready to verify that $U^{\mathrm{sym}}$ does not achieve its maximum value on the boundary, i.e., the function does not achieve its maximum value when $\rho_2 \rho_4 (1-2\rho_2-2\rho_4) \to 0^+$. We follow a similar proof as in~\cite{GSV16}.
\paragraph{Case 1: $\rho_4 \to 0^+$} Note that $1+2y_2+y_4$ can be upper bounded by a constant depending on $\rho_2$.
This can be confirmed by~\eqref{eq:y-bound}.
Therefore, by \eqref{eq:U-rho-4}, $\frac{\partial U^{\mathrm{sym}}}{\partial \rho_4} \to +\infty$ as $\rho_4 \to 0^+$.
\paragraph{Case 2: $\rho_2\to 0^+$} We first show that $y_4$ can be bounded by a constant. Let $K = \frac{2\rho_2}{1-2\rho_2-2\rho_4}$. By~\eqref{eq:y-bound} and $y_2 \ge 0$, it holds that $y_2 = \frac{1}{4} \tp{\sqrt{K^2 + 6K + 1}+K-1} = \frac{K}{1-K + \sqrt{K^2 + 6K + 1}} \le K$. Therefore, $y_4 \le 2y_2 \cdot \frac{\rho_4}{\rho_2} \le \frac{2\rho_4}{1-2\rho_2-2\rho_4}$. Hence, $\frac{1+2y_2+y_4}{1+y_2}$ is upper bounded by a constant.
According to \eqref{eq:U-rho-2}, we have $\frac{\partial U^{\mathrm{sym}}}{\partial \rho_2} \to +\infty$ as $\rho_2 \to 0^+$.
\paragraph{Case 3: $2\rho_2+2\rho_4 \to 1^-$} By the argument in Case 2, $y_2 \to +\infty$ as $2\rho_2+2\rho_4 \to 1^-$. Therefore, by \eqref{eq:U-rho-4}, $\frac{\partial U^{\mathrm{sym}}}{\partial \rho_4} \to -\infty$ as $2\rho_2+2\rho_4 \to 1^-$.
\end{proof}

We are now ready to prove~\Cref{lem:distant-alpha-hardcore}.

\begin{proof}[Proof of~\Cref{lem:distant-alpha-hardcore}]
By~\Cref{lem:approx-hardcore}, it suffices to prove there exist constants $\epsilon,\eta >0$ satisfying
\begin{align}\label{eq:require-U}
U(\*\rho) \le  -\epsilon n^{-2/3} + U(\*\rho^\star),
\end{align}
where $\*\rho^\star = \tp{\frac{(\Delta-1)^2}{\Delta^2},\frac{\Delta-1}{\Delta^2},\frac{\Delta-1}{\Delta^2},\frac{1}{\Delta^2}}$, and $\*\rho$ satisfies 
\begin{enumerate}
\item $\abs{\rho_2 +\rho_3 - \frac{2(\Delta-1)}{\Delta^2}} > n^{-1/3}$ or $\abs{\rho_4 - \frac{1}{\Delta^2}} > n^{-1/3}$;
\item $\abs{\rho_2 - \rho_3} \le \eta n^{-1/3}$.
\end{enumerate}
By~\Cref{lem:critical-hardcore} and the continuity of function $U$, there exist constants $0<\alpha\leq 1/2$ and $\beta > 0$ such that for every $\rho_4 \in [0,1/2]$, $\sup_{\rho_2 \in [0,\alpha]}U^{\-{sym}}(\rho_2,\rho_4) \leq U^{\-{sym}}(\rho_2^\star, \rho_4^\star) - \beta$.
This implies that there exists constant $\kappa$ such that $0 < 2\kappa \leq \alpha$, $4(\Delta-1)\sqrt{2\kappa} \leq \beta/2$; and when $n$ is sufficiently large, $\min\set{\rho_2,\rho_3} \leq \kappa$, we have $\max\set{\rho_2,\rho_3} \leq \kappa + \eta n^{-1/3} \leq 2\kappa$ and 
\begin{align*}
  U(\*\rho)
  &\leq U^{\-{sym}}((\rho_2+\rho_3)/2,\rho_4) + 2(\Delta-1) (2H(\tp{\rho_2 + \rho_3}/2)) \\
    &\leq U^{\-{sym}}(\rho_2^\star, \rho_4^\star) - \beta + 2(\Delta-1)(2\sqrt{2\kappa}) \leq U^{\-{sym}}(\rho_2^\star, \rho_4^\star) - \beta/2.
\end{align*}
Hence~\eqref{eq:require-U} holds when $\min \set{\rho_2,\rho_3} \le \kappa$. 
Therefore, we now assume that $\rho_2,\rho_3 \ge \kappa$.
Note that for any $x_1,x_2 \ge \kappa$,
\begin{align*}
0 \le 2H(x_1+x_2)-H(x_1)-H(x_2) \le \min_{x \ge \kappa} \abs{H''(x)} \abs{x_1-x_2}^2 \le \kappa^{-1} \abs{x_1-x_2}^2.
\end{align*}
Hence, $U(\*\rho) \le U^{\mathrm{sym}}\tp{\frac{\rho_2+\rho_3}{2}, \rho_4} + \kappa^{-1} \eta^2 n^{-2/3}$. By~\Cref{lem:critical-hardcore} and the constraints of $\*\rho$, 
\begin{align*}
U^{\mathrm{sym}}\tp{\frac{\rho_2+\rho_3}{2},\rho_4} \le \max_{(\rho'_2,\rho'_4) \in \partial B\tp{\tp{\frac{\Delta-1}{\Delta^2},\frac{1}{\Delta^2}},n^{-1/3}}} U^{\mathrm{sym}}(\rho_2',\rho_4') \le U^{\mathrm{sym}}\tp{\frac{\Delta-1}{\Delta^2},\frac{1}{\Delta^2}} - \epsilon' n^{-2/3}, 
\end{align*}
where the last inequality follows from~\Cref{lem:approx-hardcore} and~\Cref{lem:center-alpha-hardcore}, and $\epsilon' > 0$ is a constant only relying on $\Delta$. This concludes the proof with $\eta = \sqrt{\epsilon'/2\kappa}$ and $\epsilon = \epsilon'/2$.
\end{proof}

\section{A Refined SI Bound for the Hardcore Model}
\label{sec:SI-hardcore}
In this section, we present a refined analysis of the spectral independence for the hardcore model satisfying $\delta$-uniqueness,
proving~\Cref{lem:sub-SI-hardcore}.
Let $G = (V, E)$ be a graph with maximum degree $\Delta \geq 3$.
Let $\mu$ denote the Gibbs distribution of the hardcore model on $G$ with fugacity $\lambda > 0$.

For convenience, define $d := \Delta - 1$. The tree recursion for the hardcore model is defined as
\begin{align*}
    f(x) := \lambda \tp{\frac{1}{1+x}}^d.
\end{align*}
Let $\delta^* \in (0, 1)$ be a parameter.
The hardcore model on graph $G$ with fugacity $\lambda$ is said to be \emph{$\delta^*$-unique} if the following inequality holds:
\begin{align} \label{eq:def-delta-unique}
    \abs{f'(\hat{x})} = d\lambda \tp{\frac{1}{1+\hat{x}}}^{d+1} = \frac{d \hat{x}}{1 + \hat{x}} \leq 1 - \delta^*,
\end{align}
where $\hat{x}$ represents the fixed point of the tree recursion (i.e., the unique positive solution of the equation $x = f(x)$).

It is worth noting that in the literature, the condition of being $\delta$-unique  is often assumed  for all degrees up to $\Delta$.
However, for the hardcore model, assuming $\delta$-uniqueness for degrees up to $\Delta$ is equivalent to assuming it holds for the maximum degree  $\Delta$.

The following spectral independence bound is known for the hardcore model that is $\delta$-unique.
\begin{lemma}[\text{\cite{chen2023rapid}}] \label{lem:hardcore-SI-CLV}
    If the hardcore model is $\delta^*$-unique, then $\mu$ is $\frac{2}{\delta^*}\frac{\Delta \lambda}{1 + \lambda}$-spectrally independent.
\end{lemma}

\begin{remark}
In~\cite{chen2023rapid}, the authors introduce the notion of an $(\alpha, c)$-potential function, the existence of which implies $c/\alpha$-spectral independence for the distribution.
They show that the potential function proposed in~\cite{li2012correlation} is indeed an $(\alpha,c)$-potential function with desirable parameters.
Specifically, \Cref{lem:hardcore-SI-CLV} is implied by combining~\cite[Theorem 18]{chen2023rapid} (which proves $\alpha = \delta^*/2$) and ~\cite[Proof of Lemma 36 H.1]{chen2023rapid} (which proves $c = \frac{\lambda\Delta}{1+\lambda}$).  
\end{remark}


\begin{proof}[Proof of~\Cref{lem:sub-SI-hardcore}]
    First, we will show that $\frac{\Delta\lambda}{1+\lambda} \leq \e$.
    Note that $\lim_{\Delta \to \infty} \frac{\Delta\lambda_c(\Delta)}{1+\lambda_c(\Delta)} = \e$, and $\frac{\Delta\lambda}{1+\lambda} \leq \frac{\Delta\lambda_c(\Delta)}{1+\lambda_c(\Delta)}$.
    Therefore, it suffices to show that $T(\Delta) := \frac{\Delta\lambda_c(\Delta)}{1+\lambda_c(\Delta)}$ is an increasing function.
    
    Calculating the derivative $(\log T)'$, we have
    \begin{align*}
        (\log T(\Delta)') = \frac{1}{\Delta }+\frac{(\Delta -1) \left((\Delta -2) \log \left(\frac{1}{\Delta -2}+1\right)-2\right)}{(\Delta -2) \left(\left(\frac{1}{\Delta -2}+1\right)^{\Delta }+\Delta -1\right)}.
    \end{align*}
    Without loss of generality, we assume $(\Delta-2)\log(\frac{1}{\Delta-2}+1) - 2 \leq 0$ (otherwise, the lemma holds).
    Note that $(\frac{1}{\Delta-2} + 1)^\Delta \geq 1 + \frac{\Delta}{\Delta-2} + \binom{\Delta}{2} \frac{1}{(\Delta-2)^2}$, which follows from applying a generalization of the Bernoulli inequality  $(x+1)^r \geq 1 + rx + \frac{r(r-1)}{2} x^2$ that holds for $r \geq 2, x \geq 0$.
    We can apply this to the denominator in the above equation  to obtain:
    \begin{align*}
        (\log T(\Delta)') \geq \frac{-2 \Delta +2 (\Delta -2)^2 \log \left(\frac{1}{\Delta -2}+1\right)+5}{\Delta  (2 \Delta -3)}.
    \end{align*}
    Since $2\Delta \geq 3$, it suffices to prove:
    \begin{align*}
        \log\tp{1 + \frac{1}{\Delta-2}} \geq \frac{2\Delta-5}{2(\Delta-2)^2}.
    \end{align*}
    By Taylor's expansion for $\log(1+x)$ at $x = 0$, we have $\log(1+x) \geq x - \frac{x^2}{2}$ for $x \in (0,1)$, leading to:
    \begin{align*}
        \log\tp{1 + \frac{1}{\Delta-2}} \geq \frac{1}{\Delta-2} - \frac{1}{2}\tp{\frac{1}{\Delta-2}}^2 = \frac{2\Delta -5}{2(\Delta-2)^2}.
    \end{align*}
    This confirms that $T(\Delta)$ is an increasing function, proving  $\frac{\Delta\lambda}{1+\lambda} \leq \e$.
    
    According to~\Cref{lem:hardcore-SI-CLV}, we can complete the proof by showing that $2\delta^*(1+\frac{1}{\Delta-2}) \geq \delta$.
    The remaining proof follows a similar structure to the proof of~\cite[Lemma C.1]{anari2020spectral}.
    By~\eqref{eq:def-delta-unique}, to ensure the hardcore model is $\delta^*$-unique we need $\hat{x} \leq \frac{1-\delta^*}{d - (1 - \delta^*)}$.
    Since $\hat{x} = f(\hat{x})$ for the fixed point $\hat{x}$, we have $\lambda = \hat{x}(1+\hat{x})^d$,
    meaning that  $\lambda$ is an increasing function of the fixed point $\hat{x}$.
    Thus, to make the hardcore model $\delta^*$-unique, it suffices to require
    \begin{align*}
        \lambda \leq \lambda^* := \frac{1-\delta^*}{d - (1 - \delta^*)} \tp{1 + \frac{1-\delta^*}{d - (1 - \delta^*)}}^d.
    \end{align*}
    By the definition of $\lambda_c(\Delta) = \frac{1}{d-1}\tp{1 + \frac{1}{d-1}}^d$, we have
    \begin{align*}
        \frac{\lambda^*}{\lambda_c} 
        &= (1-\delta^*)\tp{1-\frac{\delta^*}{d-1+\delta^*}}^{d+1} \\
        &\geq (1-\delta^*)\tp{1 - \frac{(d+1)\delta^*}{d-(1-\delta^*)}} \tag{Bernoulli inequality} \\
        &\geq 1 - \tp{1+\frac{d+1}{d-(1-\delta^*)}}\delta^*
        \geq 1 - \tp{1 + \frac{d+1}{d-1}}\delta^*.
    \end{align*} 
    This implies that if $\delta = (1 + \frac{d+1}{d-1})\delta^* = 2(1 + \frac{1}{\Delta-2})\delta^*$ and $\lambda \leq (1 - \delta)\lambda_c(\Delta)$, then $\lambda \leq \lambda^*$ and the hardcore model is $\delta^*$-unique.
    Hence, we establish that $2\delta^*(1 + \frac{1}{\Delta-2}) \geq \delta$.
    This proves \Cref{lem:sub-SI-hardcore} as discussed above.
\end{proof}


\end{document}